\documentclass[11pt]{article}
\usepackage{enumerate,bm}
\usepackage[numbers]{natbib}
\usepackage{url}
\usepackage{color}

\usepackage{setspace}
\usepackage[normalem]{ulem}
\usepackage{graphicx,psfrag,epsf}
\graphicspath{{./PlotsArXiv/}} 

\usepackage{sidecap}
\usepackage{subfig}
\usepackage{times}
\usepackage{multirow,pdflscape}
\usepackage{amsmath, amssymb, amsfonts, amsthm}
\usepackage{myarticle-macrosShared}
\usepackage{latexsym}
\usepackage{mathtools}
\usepackage{bbm}
\usepackage{booktabs} 
\usepackage{pbox}
\usepackage{setspace}
\usepackage{enumerate,bm}

\newcommand{\Expboot}[1]{\mathbf{E}^*\left(#1\right)}

\newcommand{\polyLog}{\textrm{polyLog}}
\newcommand{\thetaHat}{\widehat{\theta}}
\date{June 18, 2016}
\title{The bootstrap, covariance matrices and PCA in moderate and high-dimensions}
\author{Noureddine El Karoui\thanks{
    The authors gratefully acknowledge the support of grant NSF DMS-1510172. N. El Karoui thanks Profs. H-T Yau and Michael Brenner for discussions while visiting Harvard's Center of Mathematical Sciences and Applications; \textbf{Keywords}: bootstrap, principal component analysis, random matrices. \textbf{AMS 2010 MSC}: Primary 62G09; Secondary: 62H25}\hspace{.2cm}\\
    and \\
    Elizabeth Purdom \\
    Department of Statistics, University of California, Berkeley}

\begin{document}
\maketitle

\begin{abstract}
We consider the properties of the bootstrap as a tool for inference concerning the eigenvalues of a sample covariance matrix computed from an $n\times p$ data matrix $X$. We focus on the modern framework where $p/n$ is not close to 0 but remains bounded as $n$ and $p$ tend to infinity.

Through a mix of numerical and theoretical considerations, we show that the bootstrap is not in general a reliable inferential tool in the setting we consider. However, in the case where the population covariance matrix is well-approximated by a finite rank matrix, the bootstrap performs as it does in finite dimension. 
	
\end{abstract}
	
\section{Introduction}

The bootstrap \cite{EfronBootstrap1979AoS} is a central tool of applied statistics, enabling inference by assessing the variability of the statistics of interest directly from the data and without explicit appeal to asymptotic theory. The appeal of the bootstrap is especially great when asymptotic theoretical derivations are difficult and/or can be done only under quite restrictive assumptions. For instance, consider the case of Principal Components Analysis (PCA). The classic text of \cite{anderson63,anderson03} (Chapter 13) gives limit theory for the eigenvalues and eigenvectors of the sample covariance matrix when the data is drawn from a normal population. These limit results are non-trivial to derive, even in the Gaussian case, and depend, for instance, on assumptions regarding the multiplicity of the eigenvalues of the population covariance matrix. Furthermore, it is clear, using approximation arguments from \cite{KatoPerturbTheory}, that these limit results are not valid for a broad class of distributions. For instance, they do not apply to populations distributions with kurtosis not equal to 3. The modern theory of PCA which aims for better finite-sample approximations by relaxing the assumption that $p/n \rightarrow 0$ is much more difficult technically and relies on very strong assumptions about the geometry of the dataset (see \cite{imj,JohnstoneReview07,nekCorrEllipD}, follow-up papers, and Section \ref{subsec:remindersHighDCovMats} below for a short summary).

Remarkably, from a theoretical standpoint, it has been shown that in many situations the bootstrap  estimates  the distribution of the statistics of interest accurately, at least with sufficient sample sizes (see \cite{BickelFreedmanTheoryBootAoS81,HallBootstrapAndEdgeworthExpansion92} for classic references). For the specific example of estimating the eigenvalues of the sample covariance matrix and PCA, numerous papers have been written about the properties of the bootstrap \cite{BeranSrivastava85,BootstrapEigenvaluesAlemayehu88,EatonTyler91,DuembgenNondiffFuncsAndBootPTRF93,HallPaulBootstrapEigenvalues2009}. The main results of these papers is that the bootstrap works in an asymptotic regime that assumes that the sample size grows to infinity while the dimension of the data is fixed, with the additional provision that the population covariance has eigenvalues of multiplicity one. When the assumption of multiplicity equal to one does not hold, subsampling techniques \cite{PolitisRomanoWolfSubsampling99} can be used to correctly estimate the distributions of interest by resampling. We note however that these subsampling techniques also require the statistician to have subsamples of size that is infinitely large compared to the dimension of the data.

Given the limitations of existing asymptotic theory and these theoretical results on bootstrapping of eigenvalues, it is not surprising that the bootstrap is a natural tool to use in connection with PCA and inferential questions therein. The bootstrap is mostly used in this context to assess variability of eigenvalues, for instance to come up with principled cutoff selections in PCA and related methods such as factor analysis. For recent examples of an applied nature, we refer the reader to \cite{TimmermanBootstrapPC2007,FastExactBootstrapPCA14,bootstrapScreeTests2006,bootstrapCIsInPCAChemometrics13}.  Another application of the bootstrap is of course in bagging \cite{BreimanBaggingPaper96}; a well known instance of bagging related to high-dimensional covariance estimation is in resampled portfolio selection \cite{MichaudBook98}.

\paragraph{Our framework: $p/n$ not close to zero} 

The theoretical assumptions that support the use of the bootstrap make the fundamental assumption that the dimension $p$ is much smaller than  $n$ (i.e. $p/n\rightarrow 0$). The modern asymptotic theory of PCA, referenced above, has shown that relaxing that assumption -- for example by assuming that $p/n<1$ but does not tend to $0$ -- leads to dramatically different theoretical behavior of the eigenvalues and eigenvectors. This leads us to  question what effect this assumption has on the performance of the bootstrap in the situation where $p/n$ is not close to 0. 

In addition its theoretical interest, the asymptotic analysis in this framework tends to yield very accurate finite-sample approximations\cite{imj} and hence gives accurate information concerning the practical performance of the methods we consider. Furthermore, in statistical practice $p/n$ is rarely very close to 0, and hence classical approximations, which rely heavily on that assumption, may lead to theoretical results and interpretations that differ quite drastically from what is observed by practitioners. 

However, when $p/n$ is not close to 0, developing theoretical results is still quite technically difficult. It requires a large variety of tools, whether one is concerned with the properties of the bulk of the eigenvalues (\cite{mp67,wachter78,silverstein95}), or the largest ones (\cite{imj,nekGencov,LeeSchnelli14}) -- which are particularly important in Principal Component Analysis (PCA). These considerations motivate our exploration of the bootstrap as an alternative, data-driven way, to perform inferential tasks for spectral properties of large covariance matrices.

\paragraph{Contributions of the paper} 

The paper is divided into two main sections. In Section \ref{sec:SimulationStudy}, we study the performance of the bootstrap by simulations in the context of PCA. We assess whether the bootstrap recovers the sampling distribution of various statistics of interest, for instance the largest eigenvalue of a covariance matrix. We also consider the simpler problem of whether the bootstrap estimates of bias and variance can be used to accurately measure the bias and the variance of various statistics. Most of our results are negative. Only when the largest eigenvalues become quite large compared to the rest does the bootstrap provide accurate inference. Furthermore, the behavior of the bootstrap is very unpredictable: for instance, in two setups that are nearly similar from a population standpoint, the bootstrap estimate of bias is itself biased, but in one case it underestimates the true bias and in the other it overestimates it. 

In Section \ref{sec:BootTheoryMain} we provide theoretical results that help explain this behavior. Those results concern two different aspects of the bootstrap. The first results are about the behavior of the bootstrapped empirical distribution of \emph{all} the eigenvalues of the sample covariance matrix $\SigmaHat$. We show that in the framework we consider ($p/n$ not very close to zero)  the bootstrapped empirical distribution is biased and asymptotically non-random. We then consider the bootstrap behavior of only the largest eigenvalues of $\SigmaHat$. We show that when the population covariance $\Sigma$ has some very large eigenvalues, far separated from the other eigenvalues, the bootstrap distribution of those large eigenvalues correctly approximates the sampling distribution of the large eigenvalues of $\SigmaHat$. 

The results of this paper confirm that the bootstrap works when the problem is very low-dimensional or can be approximated by a very low-dimensional problem, but is untrustworthy when the problem is genuinely high-dimensional. As such, the current paper complements the findings of the paper \cite{NEKEliBootRegression15} that was concerned with the bootstrap for linear regression models.

We now give basic notation and background regarding the bootstrap and estimation of covariance matrices in high dimensions.
 
\subsection{Notations and default conventions} 
If $X$ is an $n\times p$ data matrix, we call $\SigmaHat$ its associated covariance matrix, i.e 
$$
\SigmaHat=\frac{1}{n-1}(X-\bar{X})\trsp (X-\bar{X})\;.
$$
We also use the notation $\tilde{X}\triangleq (X-\bar{X})$. 

We call empirical spectral distribution of a $p\times p$ symmetric matrix $M$ the probability measure such that $dF_p(x)=\frac{1}{p}\sum_{i=1}^p \delta_{\lambda_i(M)}$, where $\lambda_1(M)\geq \lambda_2(M)\geq \ldots \geq \lambda_p(M)$ are the ordered eigenvalues of $M$. We also use the notation $\lambda_{\max}(M)$ for the largest eigenvalue of the matrix $M$.

For $z\in\mathbb{C}^+$, i.e $z=u+iv$, where $v>0$, we call $m_p(z)$ the Stieltjes transform of the distribution $F_p$, i.e 
$$
m_p(z)=\int \frac{1}{x-z}dF_p(x)=\frac{1}{p}\trace{(\SigmaHat-z\id_p)^{-1}}\;.
$$

We use the notation $\weakCV$ to denote weak convergence of probability distributions. 

$\opnorm{M}$ is the operator norm of $M$, i.e its largest singular value. $\norm{w}_\infty=\max_{1\leq i\leq p}|w_i|$ is the $\ell_\infty$-norm of the vector $w$.  We say that the sequence $u_n=\polyLog(n)$ if $u_n$ grows at most like a polynomial in $\log(n)$. 

In doing asymptotic analysis, we work under the assumptions that $p/n\tendsto r$, $r\in(0,\infty)$.

We call the \emph{Gaussian phase transition} the value $1+\sqrt{p/n}$ (see the comments after Theorem \ref{thm:JohnstoneTW} for details).
\subsection{Review of Theoretical results about High-dimensional Covariance Matrices}\label{subsec:remindersHighDCovMats}
We review briefly two key results concerning the spectral properties of high-dimensional covariance matrices. More details and more general background results are in the Supplementary Material. 

The first result (from \cite{imj}) concerns the distribution of the largest eigenvalue of $\SigmaHat$ in the case that might be considered the ``null'' case for PCA -- the predictors are all independent with covariance matrix equal to $\id_p$. 

\begin{theorem}[Johnstone, '01]\label{thm:JohnstoneTW}
Suppose the design matrix $X$ has $X_{i,j}\iid {\cal N}(0,1)$. Call $\SigmaHat$ the sample covariance matrix of $X$. Let $\lambda_1(\SigmaHat)$ be the largest eigenvalue of $\SigmaHat$. Then, if $p/n\tendsto r$, $r\in(0,\infty)$,
\begin{align*}
n^{2/3}\frac{\lambda_1(\SigmaHat)-\mu_{n,p}}{\sigma_{n,p}}&\WeakCv \text{TW}_1\;,\\
\mu_{n,p}\tendsto (1+\sqrt{r})^2
\end{align*}	
\end{theorem}

$\text{TW}_1$ refers to the Tracy-Widom distribution appearing in the study of the Gaussian Orthogonal ensemble (GOE); details about its density can be found  in \cite{imj} for instance. Further details about $\mu_{n,p}$ and $\sigma_{n,p}$ are in the appendix; they both converge to finite, non-zero limit.  This result implies, among other things, that the standard estimate $\lambda_1(\SigmaHat)$ is a biased estimator of the true $\lambda_1(\Sigma)$ when $p/n$ is not close to zero, overestimating the true size of $\lambda_1(\Sigma)$. 

From the point of view of PCA, $\lambda_{1}(\Sigma)>1$ corresponds to the scenario of the ``alternative" hypothesis, and an important question is how well can we differentiate when the data came from the alternative distribution rather than the null. \cite{bbap} shows that the distribution of  $\lambda_1(\SigmaHat)$ given in Theorem \ref{thm:JohnstoneTW} also describes the distribution of $\lambda_1(\SigmaHat)$ when $\Sigma$ is a finite rank perturbations of the $\id_p$, provided none of the eigenvalues of $\Sigma$ are too separated from each other. In practical terms, this signifies that $\lambda_1(\SigmaHat)$ has - asymptotically - the exact same distribution under the null as under the alternative and therefore no ability to differentiate the null and the alternative, provided the alternative is not far away from the null. 

A related significant result also due to \cite{bbap} gives the point at which the alternative hypothesis is sufficiently removed from the null so that the distribution of $\lambda_1(\SigmaHat)$ is stochastically different from that of $\lambda_1(\SigmaHat)$ under the null. In that paper, under slightly different distributional assumptions (see Supplementary Material), the authors show that if the largest eigenvalue is changed from 1 to $\lambda_{1}(\Sigma)>1+\sqrt{p/n}$ and the other eigenvalues remain the same, then $\lambda_1(\SigmaHat)$ has Gaussian fluctuations and they are of order $n^{-1/2}$. See also \cite{debashis}.  The value $1+\sqrt{p/n}$ is therefore called the \emph{Gaussian phase transition}.
Part of our simulation study investigates whether the bootstrap is capable of capturing this statistically interesting phase transition. 

Another very important result concerns the distribution of eigenvalues of the sample covariance matrix. If we call $F_p$ the empirical spectral distribution of $\SigmaHat$, we have 
\begin{theorem}[Marchenko-Pastur, '67]
Suppose $X_{i,j}$ are i.i.d with mean 0, variance 1, and a fourth moment. Then, as $p/n\tendsto r \in (0,1)$, 
$$
F_p \WeakCv F_r\;, a.s \;.
$$
Furthermore, $F_r$, the so-called Marchenko-Pastur distribution, has the density $f_r$ with
$$	
f_r(x)=\frac{\sqrt{(r_+-x)(x-r_-)}}{2\pi r x} 1_{x\in (r_-,r_+)}\;,
$$
where $r_\pm=(1\pm \sqrt{r})^2$.
\end{theorem}
Marchenko and Pastur's paper \cite{mp67} contains many more results, including the case where the population covariance is not $\id_p$, the case where $r>1$, etc. The previous result can be interpreted as saying that the histogram of sample eigenvalues is asymptotically non-random, and its limiting shape, which depends on the ratio $p/n$, is characterized by the Marchenko-Pastur distribution.

\section{Simulation Study}\label{sec:SimulationStudy}

We investigate via simulation the behavior of the bootstrap for the top eigenvalue of the standard sample covariance matrix, $\frac{1}{n-1}\tilde{X}'\tilde{X}$, where $\tilde{X}$ refers to the matrix of centered data values. For simplicity, we consider the case where only the top eigenvalue $\lambda_1$ is allowed to vary;  we assume the remaining eigenvalues are all equal to $1$. Therefore in this simple setting, the inferential question is to determine whether the top eigenvalue $\lambda_1$ differs from $1$. 

In our simulations, we generate data $X$ for several distributions of $X$, and perform the bootstrap for the top eigenvalue of the sample covariance matrix. In what follows, we focus on when the observations $X_i$ come from either a multivariate Normal distribution or an elliptical distribution with exponential weights (see Supplementary Text, Section \ref{Supp:Simulation} for details on this and other elliptical distributions we considered). Many theoretical results on random matrices in high dimensions do not yet extend to the case of elliptical distributions where the geometry of the data is more complex than under standard setups. Therefore, simulations under the elliptical distribution, while still an idealization, represents a simulation scenario that is somewhat more realistic than that of standard models. 

Regarding inference on $\lambda_1$, we consider two different types of questions for which the bootstrap could be used. The first is to estimate basic features of the estimate $\hat{\lambda}_1$, such as its variance or bias. The second is to use the bootstrap to perform inference on $\lambda_1$, for example to create confidence intervals for $\lambda_1$; there exist several methods of creating bootstrap confidence intervals which we evaluate. The results of our simulations shows that the bootstrap performs quite badly for all of these tasks as the ratio $p/n$ grows, unless the top eigenvalue $\lambda_1$ is very large and thus highly separated from the other eigenvalues. Our theoretical work explains these phenomena.

\paragraph{Estimating Bias} As noted above in the background (see Subsection \ref{subsec:remindersHighDCovMats} and the Supplementary Material), the eigenvalues of the sample covariance are biased in high dimensions for estimating the true eigenvalues. Before considering the bootstrap estimates of bias, we first note the importance of this bias in understanding the behavior of $\hat{\lambda}_1$. The bias can be substantial unless $\lambda_1$ is quite large, and the bias is clearly evident even for low ratios of $r=p/n$ if $\lambda_1$ is close to $1$. Furthermore, this bias is more pronounced for elliptical distributions than the normal distribution (which can be explained through the results of \cite{nekGencov,Onatski08}). For example, for the null setting of $\lambda_1=1$, with a ratio of $p/n$ as low as $0.01$, we see a bias in $\hat{\lambda}_1$, overestimating the true $\lambda_1$ by $17\%$ for the normal distribution, and $49\%$ for the elliptical distribution with exponential weights (Supplemental Table \ref{tab:bootBiasNorm}-\ref{tab:bootBiasEllipUnif}). As the ratio of $p/n$ grows, the bias increases, with $\hat{\lambda}_1$ overestimating $\lambda_1$ by $1.88$ and $14.93$ when $\lambda_1=1$ for the normal distribution and the elliptical distribution with exponential weight, respectively when $p/n=0.5$. The bias declines as $\lambda_1$ grows and becomes more separated from the remaining eigenvalues, especially relative to the size of $\lambda_1$; Subsection \ref{subsec:bootExtremeEigenvalues} provides an explanation for this phenomenon. But the bias remains for large ratios of $p/n$ even for $\lambda_1$ well beyond the Gaussian phase transition $(1+\sqrt{p/n})$, especially for non-normal distributions; for $p/n=0.5$, when $\lambda_1$ is as large as $1+11\sqrt{p/n}\approx 8.78$, the bias of $\hat{\lambda}_1$ is $8.08$ when $X$ follows an elliptical distribution with exponential weights, and $2.01$ for an elliptical distribution with normal weights. Any use of $\hat{\lambda}_1$ as an estimate of $\lambda_1$ must grapple with the problem of such a highly non-consistent estimator, making bootstrap methods for estimating the bias highly relevant.

The standard bootstrap estimate of bias is given by, $\bar{\lambda}_1^*-\hat{\lambda}_1$, where $\bar{\lambda}_1^*=\frac{1}{B}\sum_{b}\hat{\lambda}^{*b}_1$ is the mean of the bootstrap estimates of $\lambda_1$. Unfortunately, we see in simulations that this bootstrap estimate of bias is not a reliable estimate of the bias of $\hat{\lambda}_1$ unless $\lambda_1$ is quite large relative to the other eigenvalues.  We demonstrate these results in Figure \ref{fig:bootBias} where we plot boxplots of the standard bootstrap estimates of bias over 1,000 simulations for different ratios of $p/n$ and values of $\lambda_1$. For $X_i$ following a normal distribution, the bootstrap estimate of bias remains poor for large $p/n$ even for $\lambda_1$ past the phase transition, e.g. $\lambda_1=1+3 \sqrt{p/n}$ ($\lambda_1\approx 3.1$ for $p/n=0.5$), and only for much larger values of $\lambda_1$ does the bootstrap estimate of bias start to approach the true bias in high dimensions (Supplementary Table \ref{tab:bootBiasNorm}). Further the bootstrap estimate of bias is inconsistent: depending on the true value of $\lambda_1$ the bootstrap either under- or over-estimates the bias. Another important feature of the bootstrap shown in our simulations is that when $X_i$ follows an elliptical distribution with exponential weights, while the mean performance of the bootstrap estimate is still poor, it is the extremely high variance of the bootstrap that is even more problematic for the bootstrap estimate of bias. Indeed, as we discuss below, the bootstrap distribution of the top eigenvalue seems to change dramatically from simulations to simulations, thus creating highly variable estimates.

\begin{figure}[t]
	\centering
	\subfloat[][$Z\sim$ Normal]{\includegraphics[width=.45\textwidth]{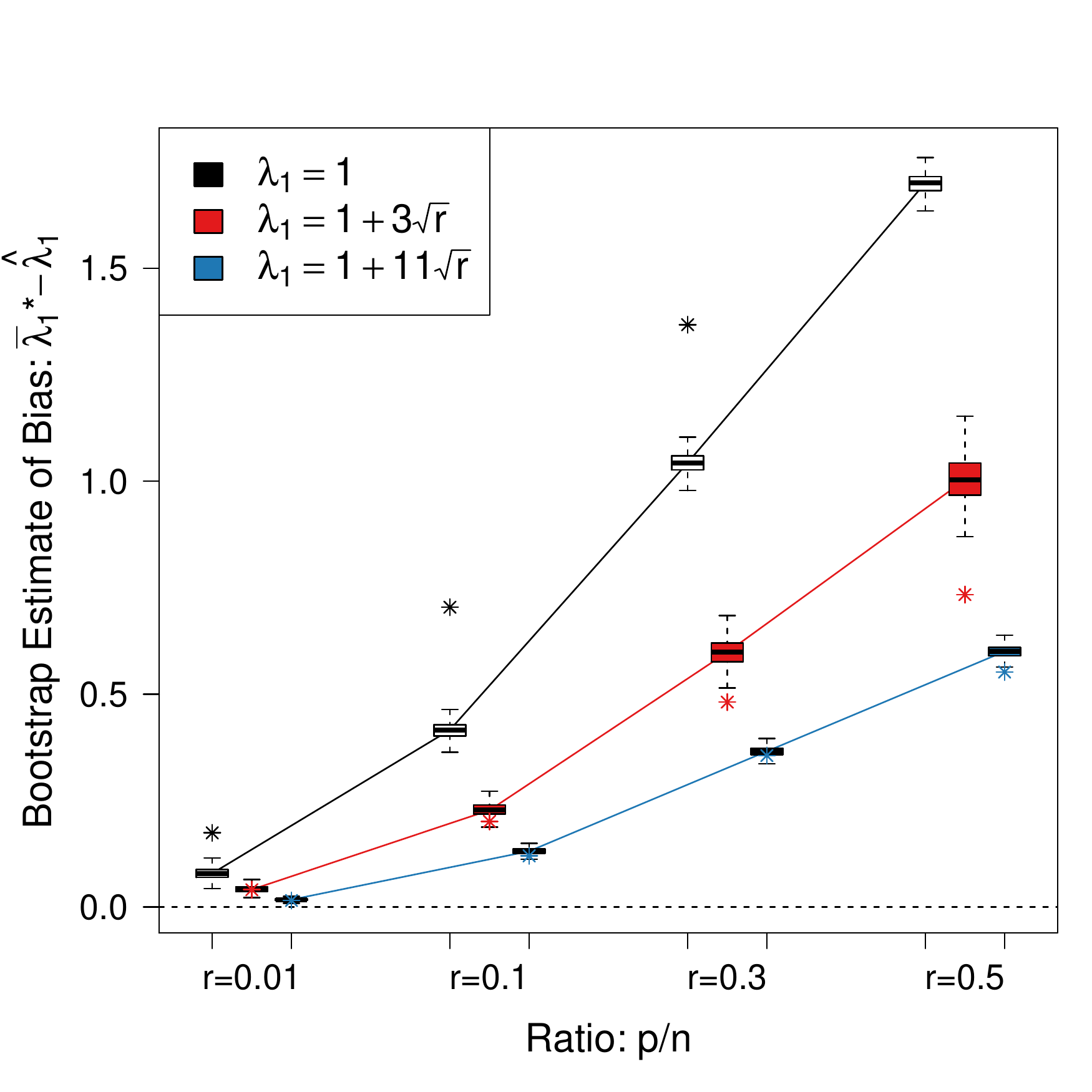} \label{subfig:bootBias:Norm} }
	\subfloat[][$Z\sim$ Ellip. Exp]{\includegraphics[width=.45\textwidth]{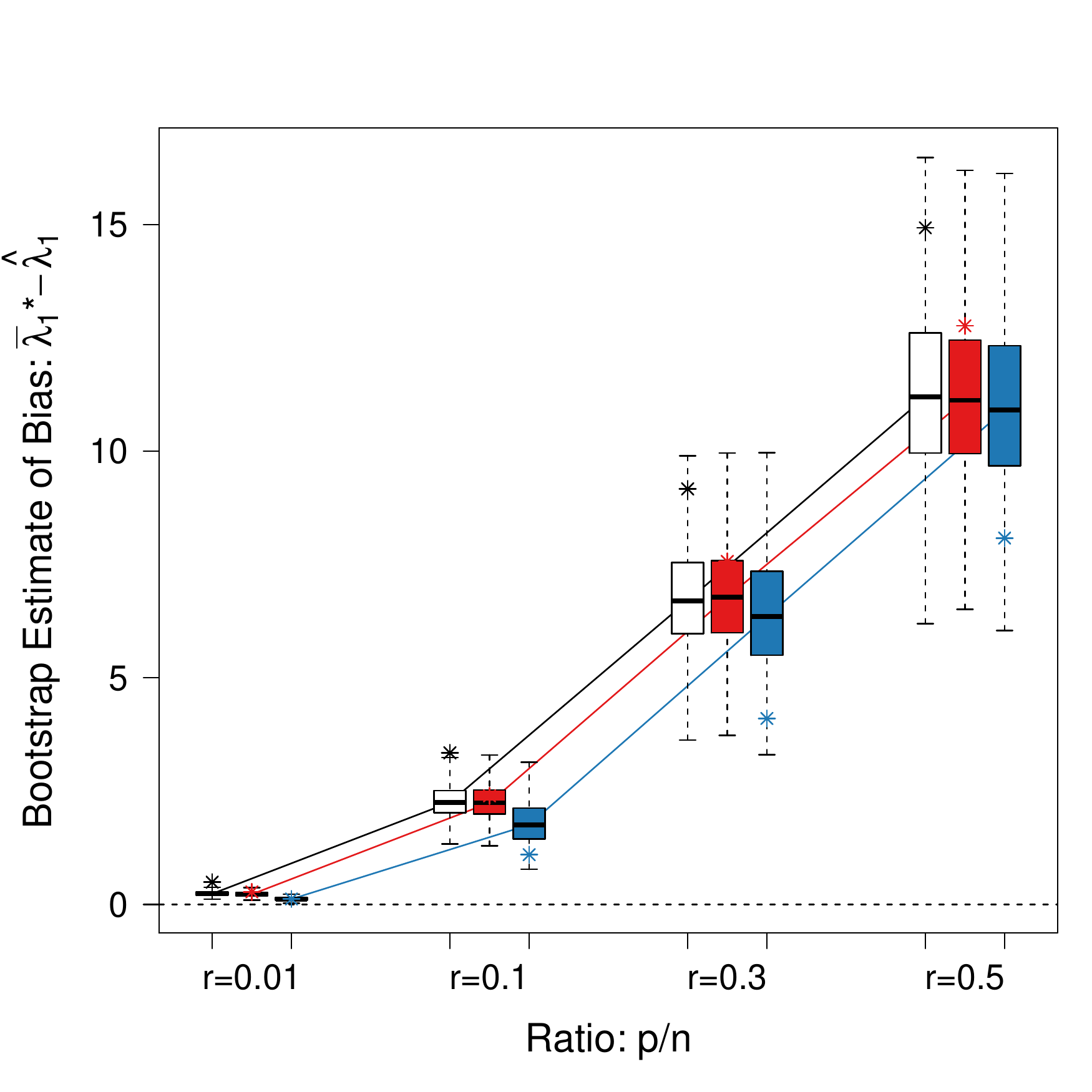} \label{subfig:bootBias:EllipExp} }
\caption{
 \textbf{Bias of Largest Bootstrap Eigenvalue, n=1,000: } Plotted are boxplots of the difference of $\bar{\lambda}_1^*$, the average bootstrap value of $\lambda_1$ over 999 bootstrap samples, minus the estimate $\hat{\lambda}_1$. This is repeated over 1000 simulations. $\bar{\lambda}_1^*-\hat{\lambda}_1$  is also the standard bootstrap estimate of bias. Each group of boxplots along the x-axis corresponds to a different ratio $r$ of $p/n$; different colors of the boxplot correspond to different values of the true $\lambda_1:$ $1$, $1+3\sqrt{r}$,$1+11\sqrt{r}$; for larger values of $\lambda_1$ see Supplementary Figure \ref{fig:bootBigBias}. The asterisk (*) in the plot corresponds to the true bias, $\hat{\lambda}_1-\lambda_1$ as evaluated over 1,000 simulations.  See Supplemental Tables \ref{tab:bootBiasNorm}-\ref{tab:bootBiasEllipUnif} for the median values of these boxplots and for those of larger $\lambda_1$ values.
}\label{fig:bootBias}
\end{figure}	

\paragraph{Estimating the Variance} We see similar problems in our simulations for the boostrap estimate of variance (Figure \ref{fig:bootVar}). Specifically, the bootstrap dramatically overestimates the variance of $\hat{\lambda}_1$ when $\lambda_1$ is close to 1. When the $X_i$'s are normally distributed and $\lambda_1=1$, the bootstrap estimates the variance to be four times larger than the true variance for $p/n=0.1$, and grows to be up to 60 times larger than the true variance when $p/n=0.5$ (Supplementary Table \ref{tab:bootVar}). Even when $\lambda_1=1+3\sqrt{p/n}$, well beyond the Gaussian phase transition and hence in a relatively easy setup, the bootstrap estimate of variance is inflated to 1.5 to 2.2 times as large as the truth, for $p/n=0.3$ and $0.5$, respectively. Only for large values of $\lambda_1$ does the bootstrap inflation of the variance become minimal. Again, when the $X_i$'s follow an elliptical distribution with exponential weights, the behavior of the bootstrap estimate of variance is dominated by the variability in the estimate, because the distribution of $\hat{\lambda}_1^*$ is so erratic.

\begin{figure}[t]
	\centering
	\subfloat[][$Z\sim$ Normal]{\includegraphics[width=.45\textwidth]{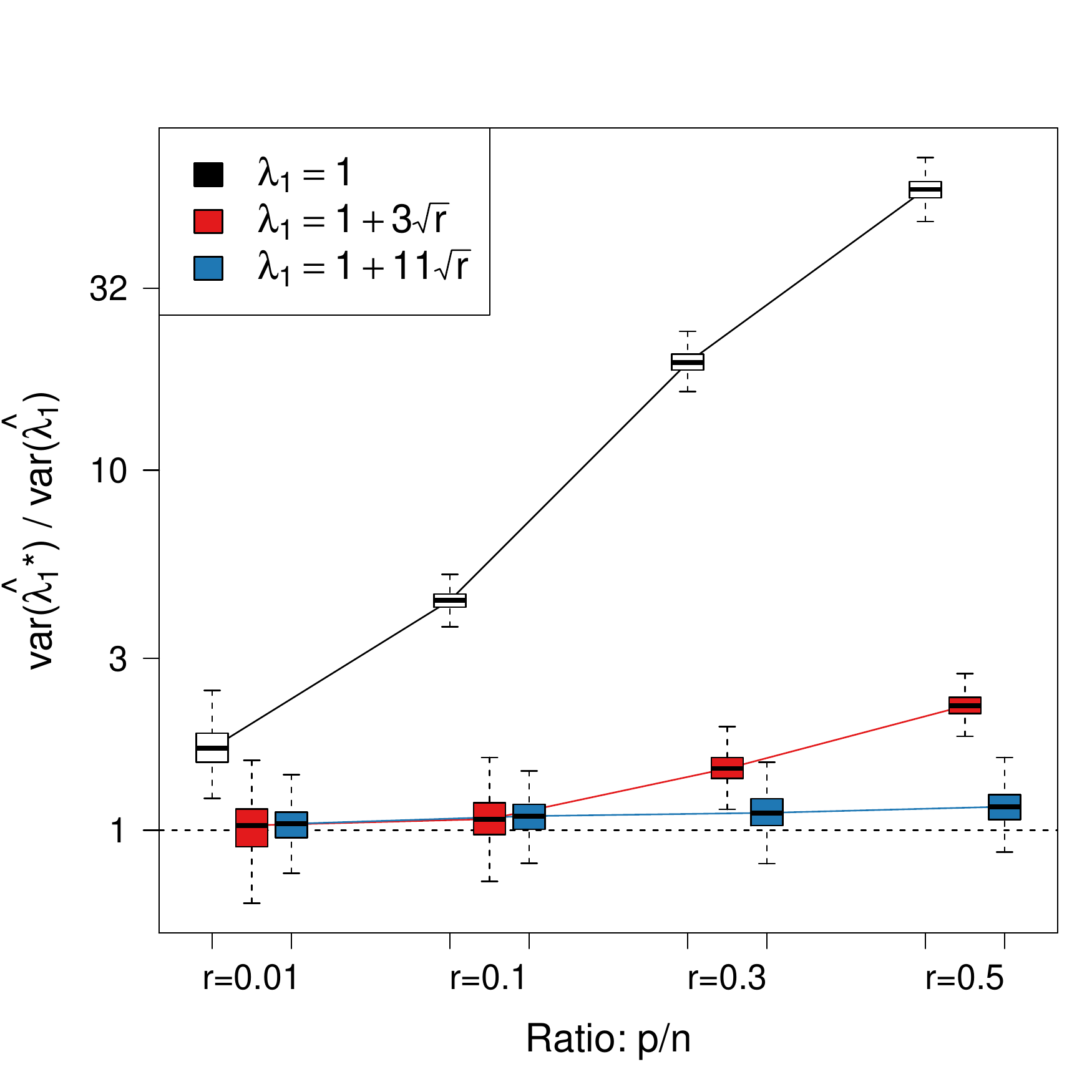} \label{subfig:bootVar:Norm} }
	\subfloat[][$Z\sim$ Ellip. Exp]{\includegraphics[width=.45\textwidth]{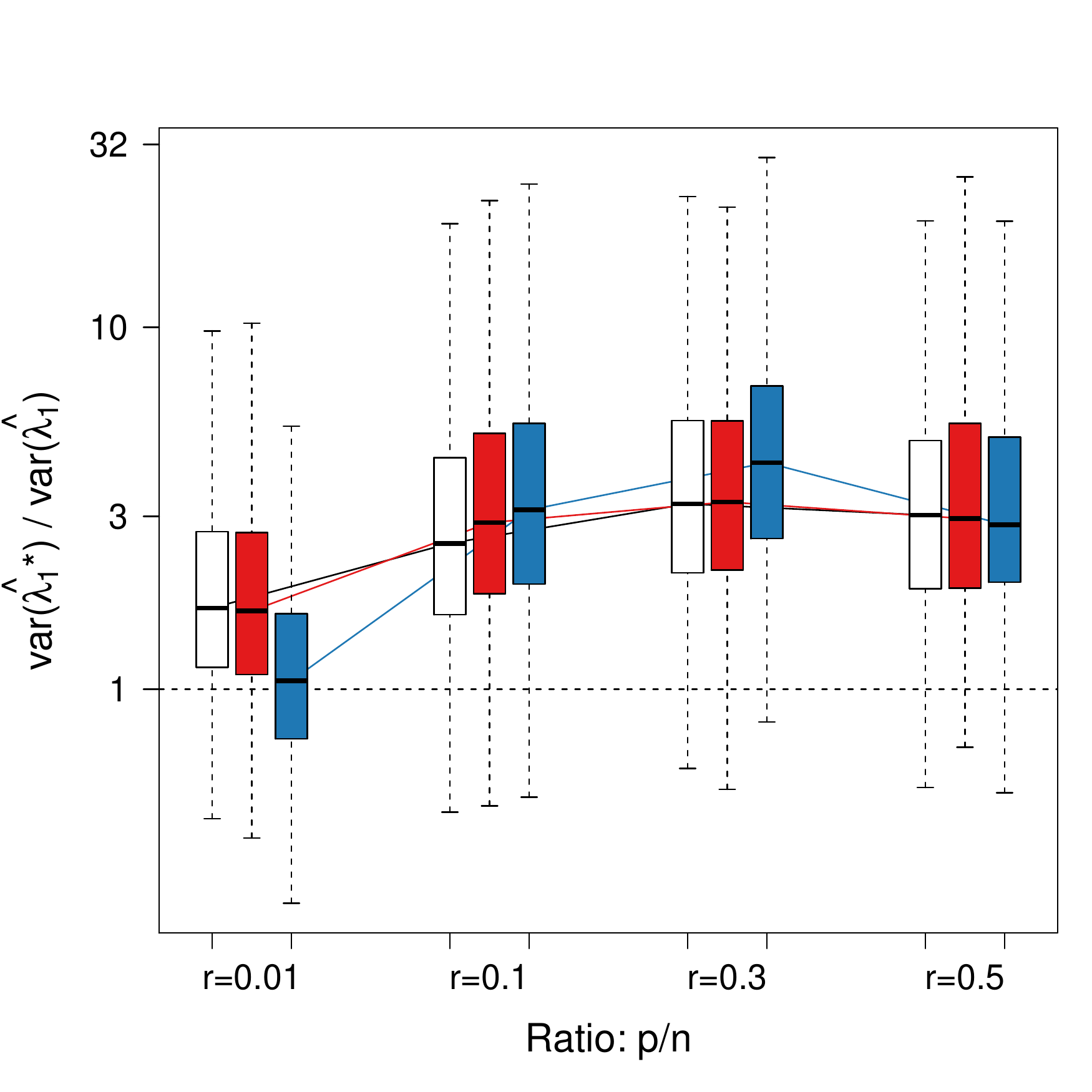} \label{subfig:bootVar:EllipExp} }
\caption{
 \textbf{Ratio of Bootstrap Estimate of Variance to True Variance for Largest Eigenvalue, n=1,000:  } Plotted are boxplots of the bootstrap estimate of variance ($B=999$) as a ratio of the true variance of  $\hat{\lambda}_1$; boxplots represent the bootstrap estimate of variance over 1000 independent simulations. Each group of boxplots along the x-axis corresponds to a different ratio $r$ of $p/n$; different colors of the boxplot correspond to different values of the true $\lambda_1:$ $1$ (white), $1+3\sqrt{r}$ (red), and $1+11\sqrt{r}$ (blue); for larger values of $\lambda_1$ see Supplementary Figure \ref{fig:bootBigVar}. See Supplemental Table \ref{tab:bootVar} for the median values of boxplots.
}\label{fig:bootVar}
\end{figure}

\paragraph{Confidence Intervals for $\lambda_1$} Standard techniques for creating confidence intervals for $\lambda_1$ are clearly problematic in high dimensions since $\hat{\lambda}_1$ is biased and  not a consistent estimator for $\lambda_1$. If $\lambda_1(\Sigma)$ is beyond the Gaussian phase transition, it is nonetheless fairly straightforward to construct those intervals in the Wishart data setting. However, it is still useful to consider the performance of bootstrap confidence intervals because of what they highlight about the behavior of the bootstrap for the top eigenvalue. Another reason is that  many practitioners might use the bootstrap to make statements about whether top eigenvalues are separated from the rest of the eigenvalue spectrum.

Bootstrap confidence intervals can be created in multiple ways \cite{DavisonHinkley97}. Common techniques include 1) a simple normal confidence interval around $\hat{\lambda}_1$ using the bootstrap estimate of variance, 2) the percentile method, which uses the percentiles of the bootstrap distribution of $\hat{\lambda}_1^{*b}$, or 3) a bias-corrected confidence interval. Based on our earlier discussions, it is not surprising that none of these methods for estimating confidence intervals will have the proper coverage probability.  Examining the actual bootstrap distributions of $\hat{\lambda}_1^*-\hat{\lambda}_1$ from multiple simulations (Figure \ref{fig:bootDensityNull}), we see clearly the incorrect bias estimation and the overestimation of variance that results from using the bootstrap, features that will also invalidate confidence intervals constructed from the percentiles of the distribution of $\hat{\lambda}^*_1$. We also see that when the $X_i$'s follow an elliptical distribution with exponential weights, the bootstrap distributions do not appear to be converging to a limit for small values of $\lambda_1$, at least for $n=1000$ -- an even greater problem in using the bootstrap in these settings.

\begin{figure}[t]
	\centering
	\subfloat[][$Z\sim$ Normal, r=0.01]{\includegraphics[type=pdf,ext=.pdf,read=.pdf,width=.3\textwidth]{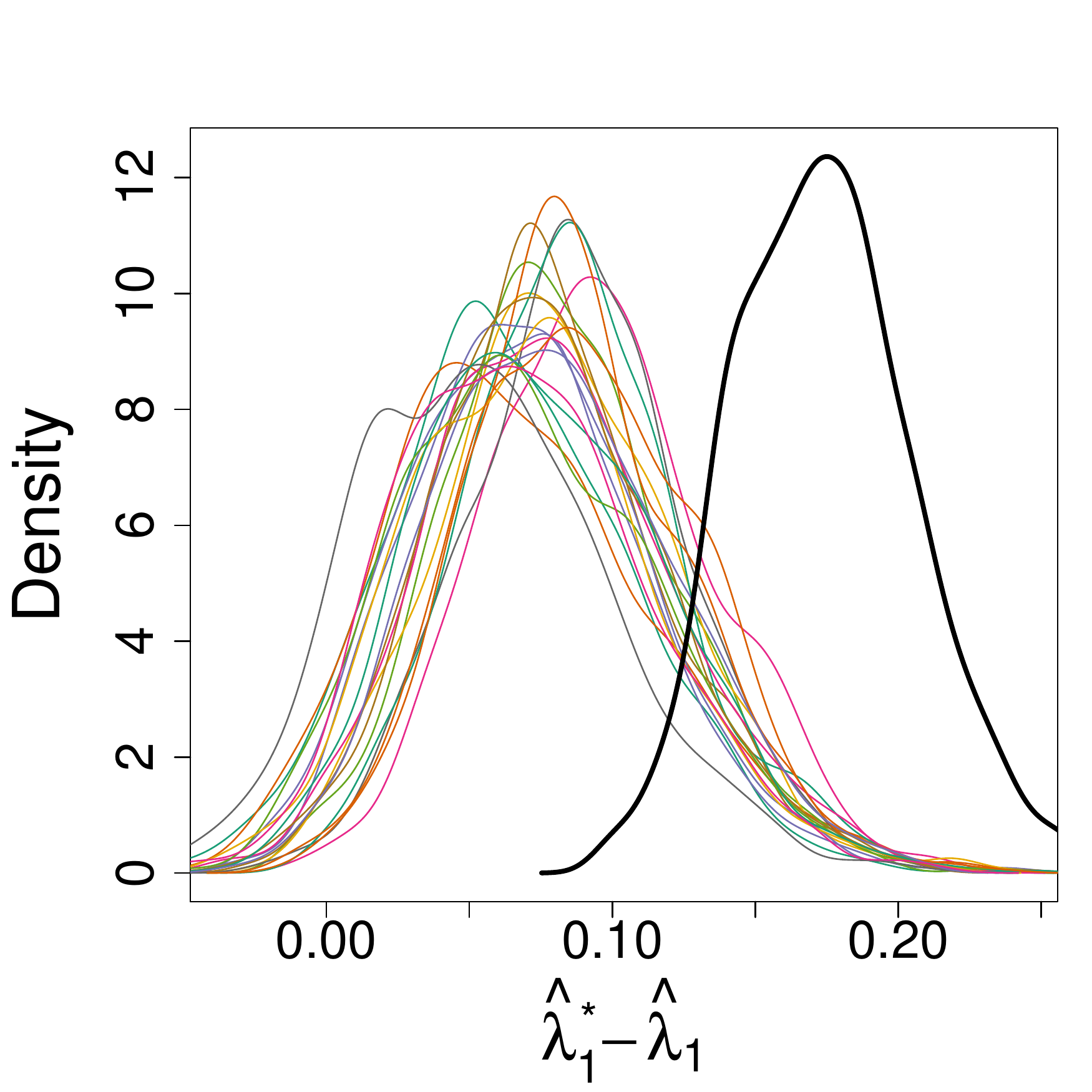} \label{subfig:bootDensityNull:Norm0.01} }
	\subfloat[][$Z\sim$ Normal, r=0.3]{\includegraphics[type=pdf,ext=.pdf,read=.pdf,width=.3\textwidth]{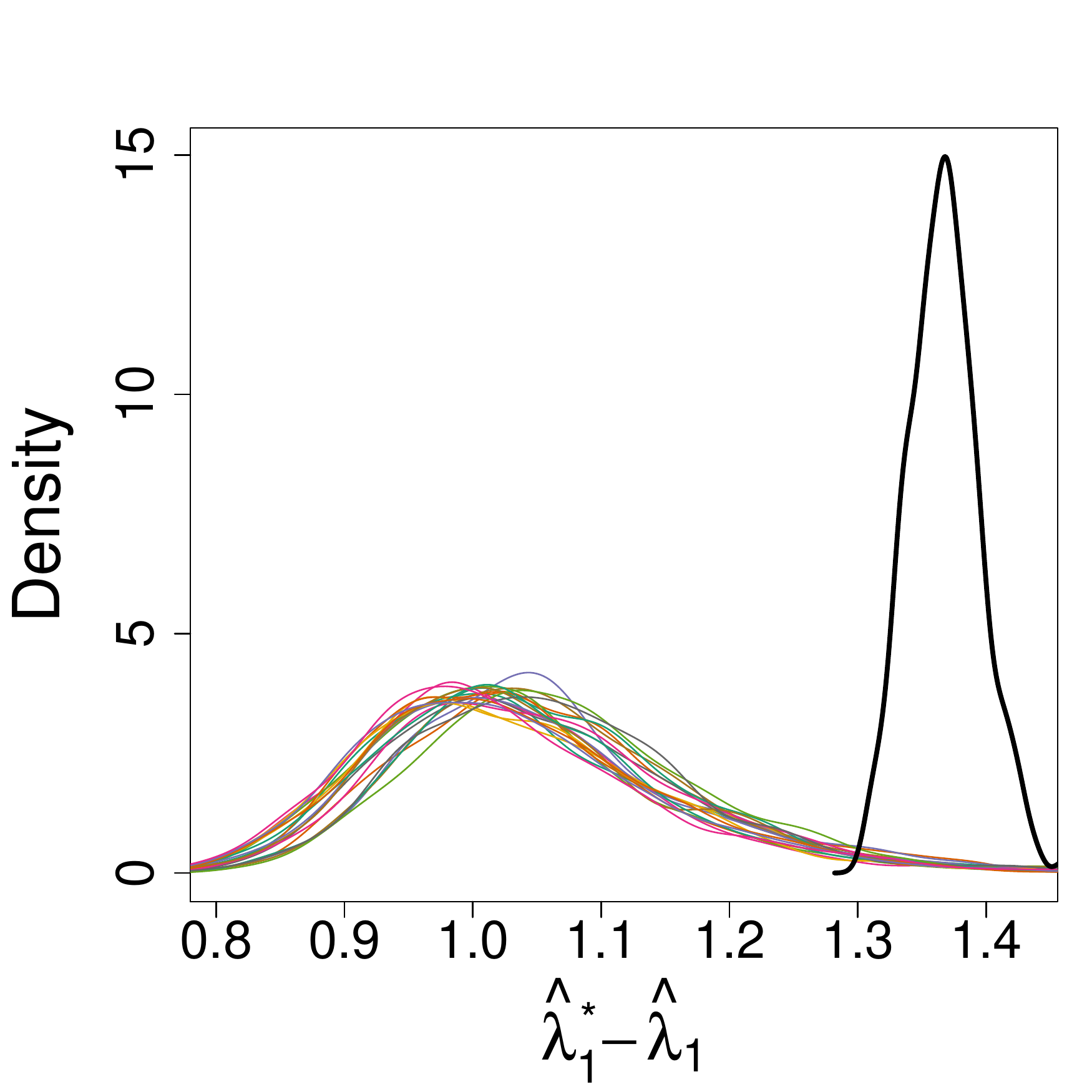} \label{subfig:bootDensityNull:Norm0.3} }
	\\
	\subfloat[][$Z\sim$ Ellip. Exp, r=0.01]{\includegraphics[type=pdf,ext=.pdf,read=.pdf,width=.3\textwidth]{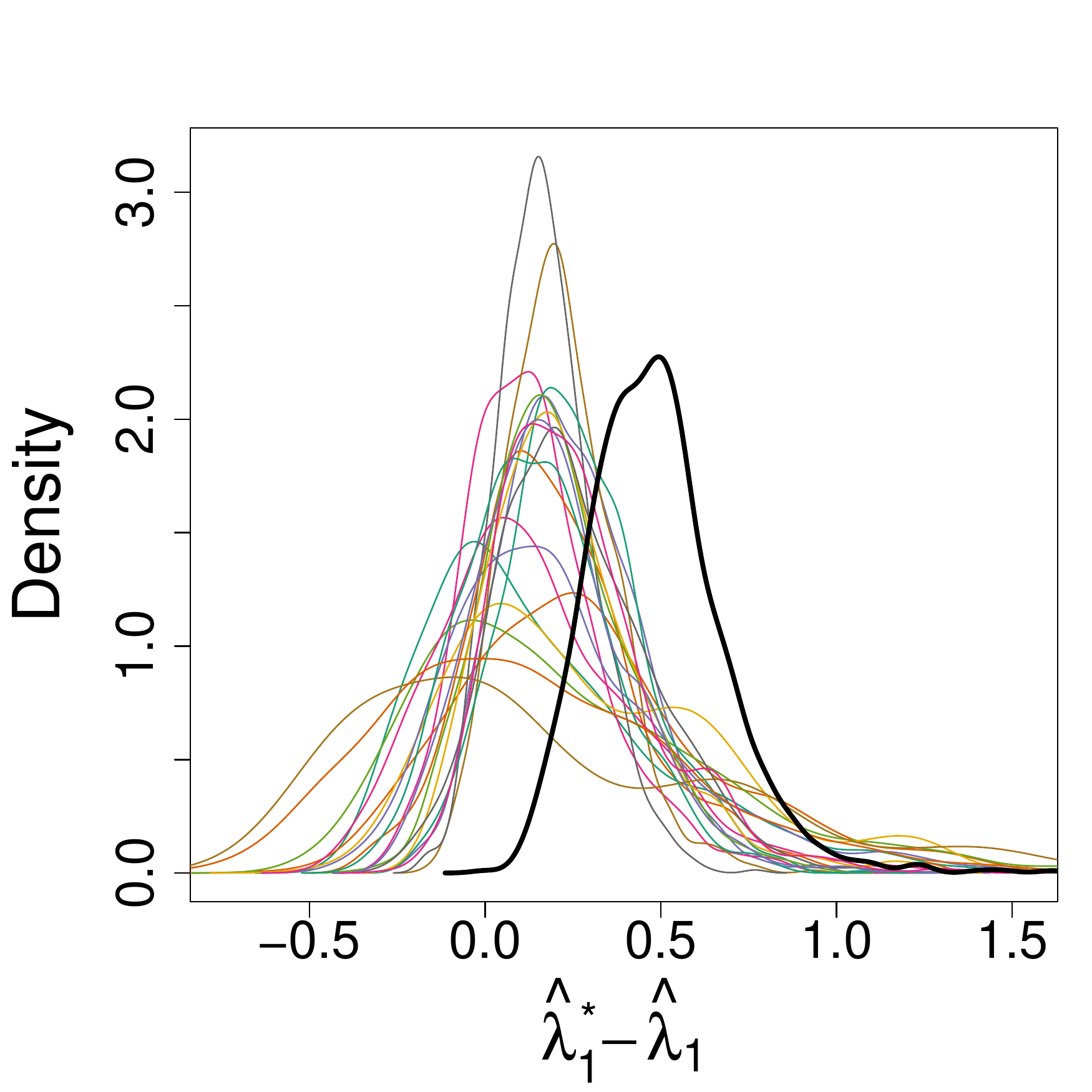} \label{subfig:bootDensityNull:EllipExp0.01} }
	\subfloat[][$Z\sim$ Ellip. Exp, r=0.3]{\includegraphics[type=pdf,ext=.pdf,read=.pdf,width=.3\textwidth]{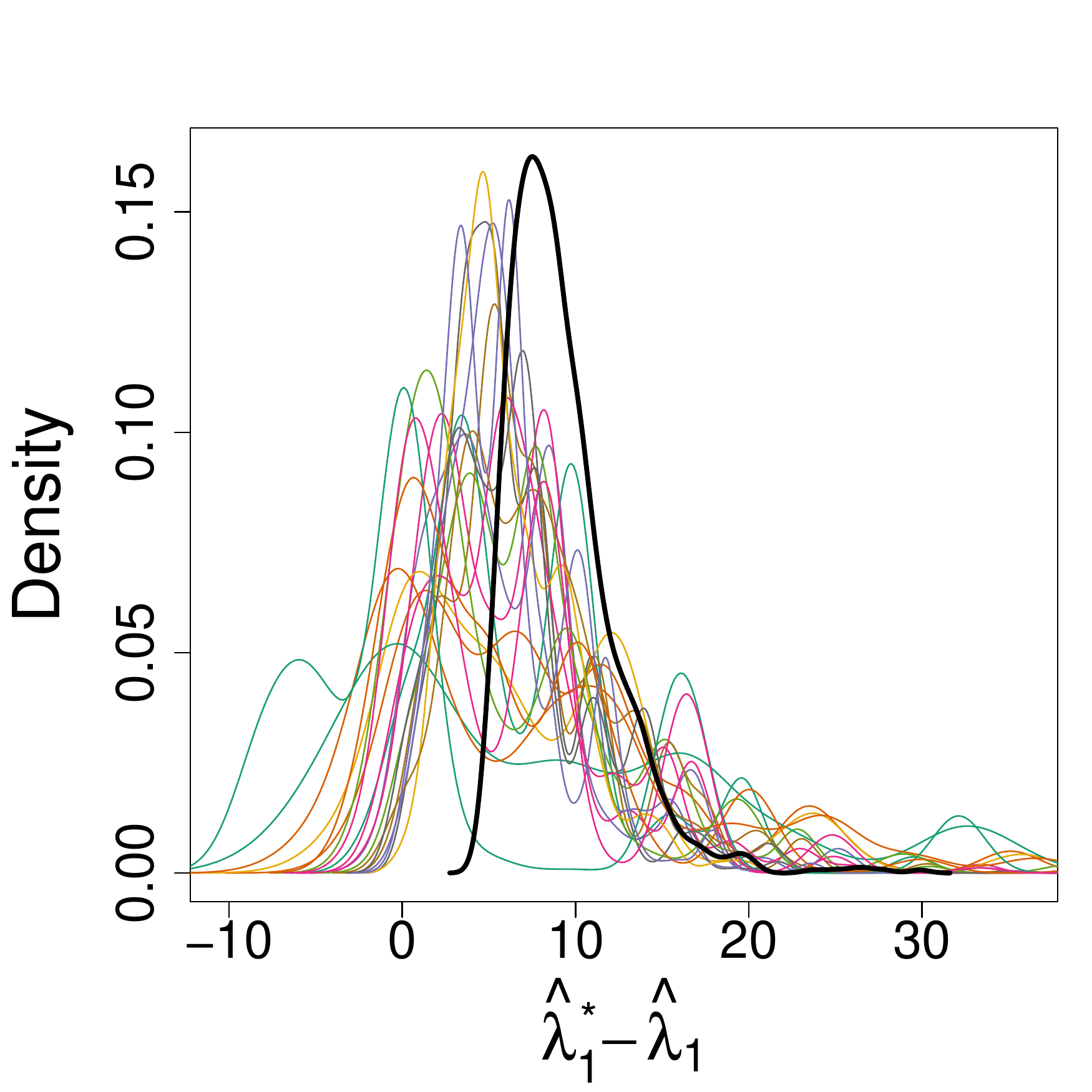} \label{subfig:bootDensityNull:EllipExp0.3} }
\caption{
 \textbf{Bootstrap distribution of $\hat{\lambda}^*_1$ under the null ($\lambda_1=1$), n=1,000:  } Plotted are the estimated density of twenty simulations of the bootstrap distribution of $\hat{\lambda}^{*b}_1-\hat{\lambda}_1$, with $b=1,\ldots,999$. The solid black line line represents the distribution of $\hat{\lambda}_1-\lambda_1$ over 1,000 simulations.  For similar figures for the larger value of $\lambda_1=1+3\sqrt{r}$, see Supplementary Figure \ref{fig:bootDensityAlt}.
}\label{fig:bootDensityNull}
\end{figure}

As expected, the resulting bootstrap confidence intervals are not useful in inference on the true value of $\lambda_1$. Bootstrap confidence intervals based on the percentile estimates do not cover the true value with any kind of reasonable probability until $\lambda_1$ becomes quite large  (Figure \ref{fig:bootCI} and Supplementary Table \ref{tab:bootCITrue}); this is undoubtedly because of the bias in the estimate of $\lambda_1$ making the percentile method inappropriate for constructing confidence intervals. Bootstrap confidence intervals based on normal intervals around $\hat{\lambda}_1$ using the bootstrap estimate of variance \emph{do} cover the true value of the $\lambda_1$ with high probability. However, this coverage is due to the fact that the bootstrap estimate of variance is much larger than the true variance of $\hat{\lambda}_1$, as seen above, and thus result in overly large confidence intervals. As a result, the normal-based bootstrap confidence intervals also incorrectly cover the putative null hypothesis ($\lambda_1=1$) with high probability when the alternative is true, particularly for $X_i$ following an elliptical distribution (Supplementary Table \ref{tab:bootCINull}). In short, such bootstrap confidence intervals based on the bootstrap estimate of variance suffer from lack of power once the distribution of $X_i$ deviates from strictly normal because of the large size of the confidence interval.

\begin{figure}[t]
	\centering
	\subfloat[][\centering$Z\sim$ Normal\par Percentile Intervals]{\includegraphics[type=pdf,ext=.pdf,read=.pdf,width=.3\textwidth]{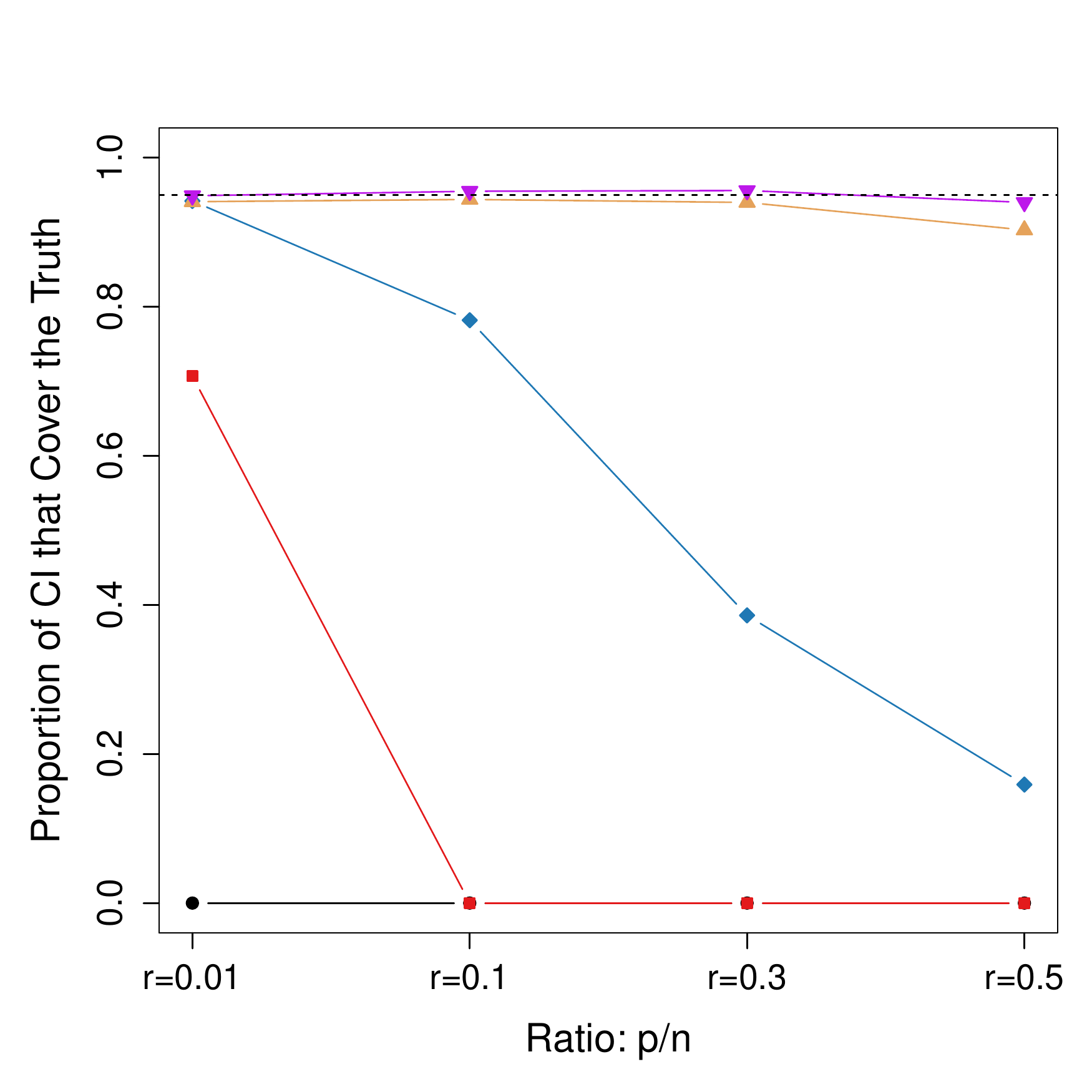} \label{subfig:bootCI:NormPerc} }
	\subfloat[][\centering$Z\sim$ Normal\par Normal-based Intervals]{\includegraphics[type=pdf,ext=.pdf,read=.pdf,width=.3\textwidth]{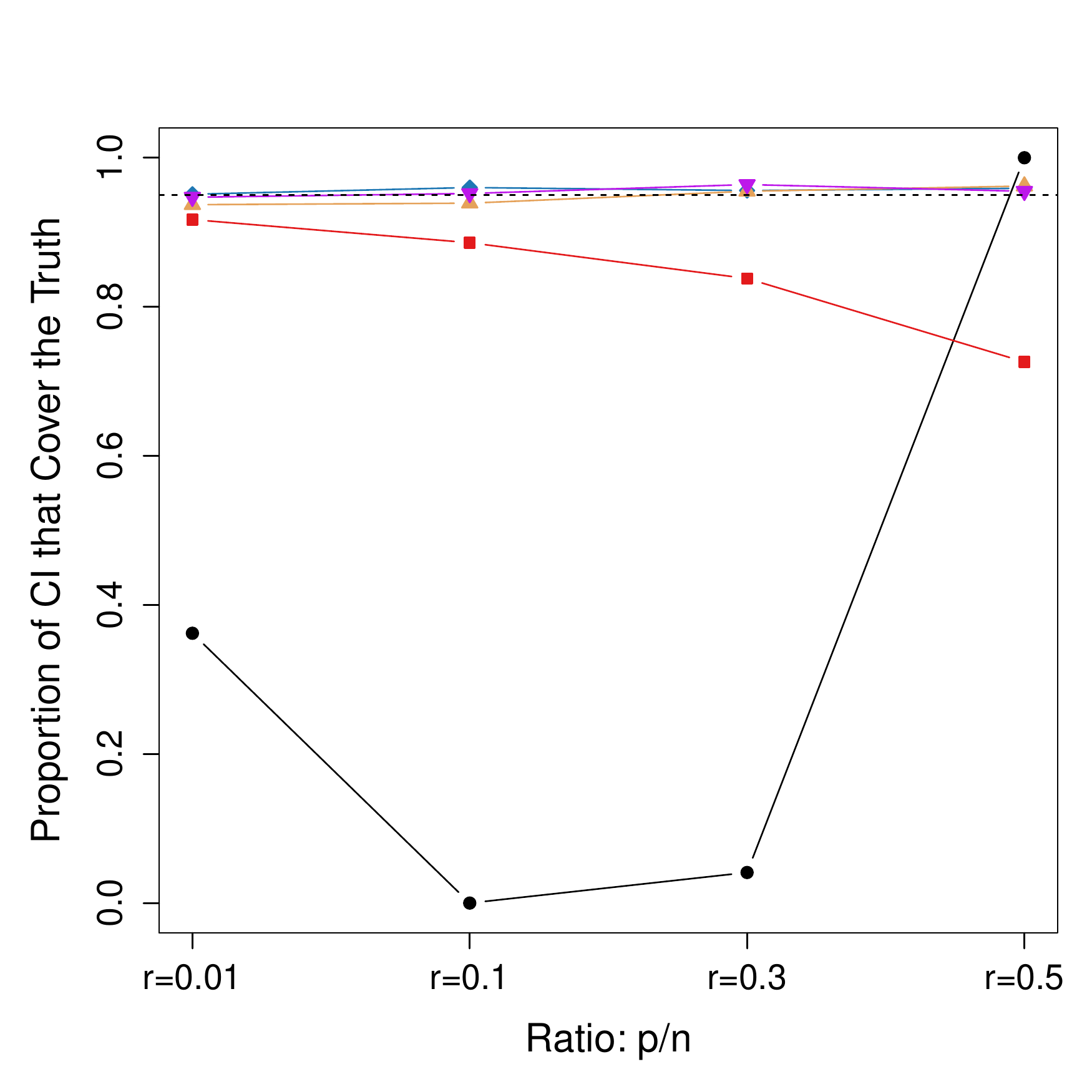} \label{subfig:bootCI:NormNormal} }
	\\
	\subfloat[][\centering$Z\sim$ Ellip. Exp\par Percentile Intervals]{\includegraphics[type=pdf,ext=.pdf,read=.pdf,width=.3\textwidth]{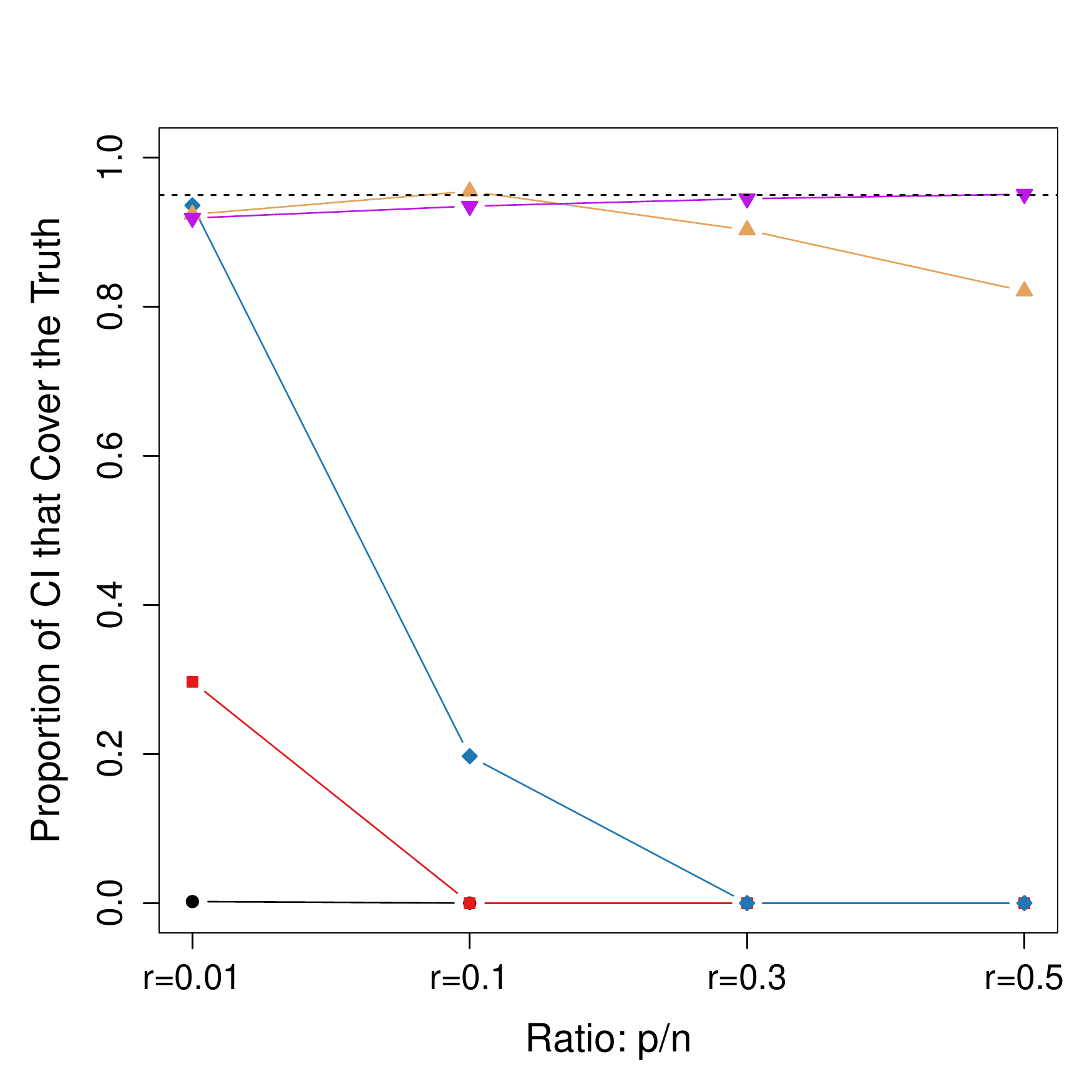} \label{subfig:bootCI:EllipExpPerc} }
	\subfloat[][\centering $Z\sim$ Ellip. Exp \par Normal-based Intervals]{\includegraphics[type=pdf,ext=.pdf,read=.pdf,width=.3\textwidth]{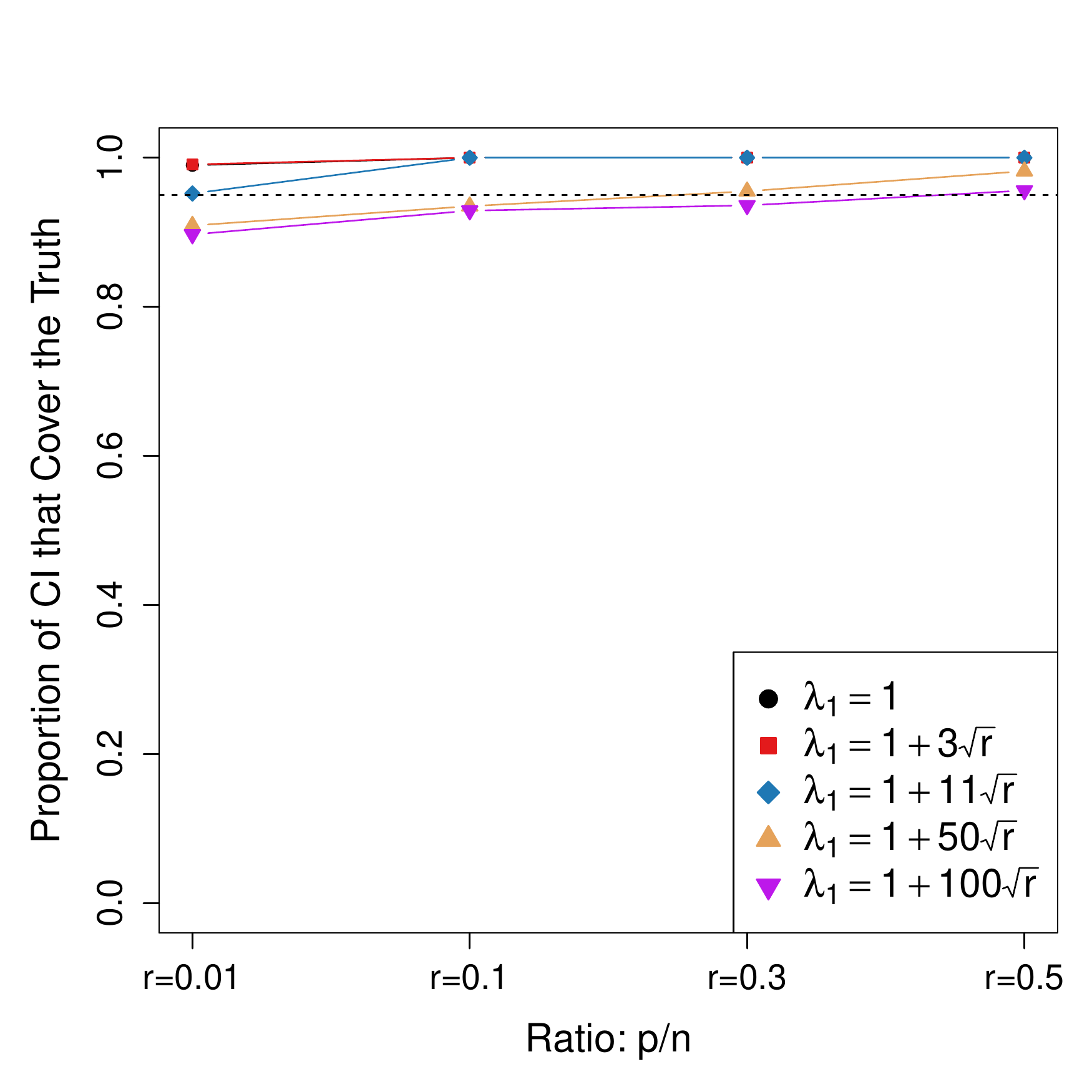} \label{subfig:bootCI:EllipExpNorm} }
\caption{
 \textbf{95\% CI Coverage, $n=1,000$:  } Plotted are the percentage of the confidence intervals that cover the true $\lambda_1$ (out of 1,000 simulations), for different values of $r=p/n$ and for different true values of $\lambda_1$. The plotted values can be found in Supplementary Table \ref{tab:bootCITrue}, and the percentage of these same intervals that cover the null value $\lambda_1=1$ can be found in Supplementary Table \ref{tab:bootCINull}. Additional elliptical distributions can be seen in \ref{fig:bootCIMoreEllip}.
}\label{fig:bootCI}
\end{figure}

\subsection{Statistics for Detecting Gaps in the Eigenvalue Spectrum}

The behavior of the top eigenvalues is often studied in theoretical work, but in practice, examination of the eigenvalues of the sample covariance matrix is largely done to find gaps in the eigenvalue spectrum. Such gaps might indicate a logical point at which to reduce the dimension of the data.  Again, we focus for simplicity on detecting the separation of just the top eigenvalue from the remainder. Then a natural statistic is the gap statistic, $\hat{\lambda}_1-\hat{\lambda}_2$, where large values of the statistic are meant to suggest that there is a large difference between the first and second eigenvalue. 

These statistics are difficult to understand theoretically, with limit distributions that are even less standard than the Tracy-Widom distribution (see \cite{SoshnikovUniversalityWigner99}, which explains joint distributional results, \cite{dieng04}, and \cite{nekGencov} for applications), which again makes them good candidates for using the bootstrap for inference. We consider the performance of the bootstrap for these statistics via the same simulation structure that we applied to the largest eigenvalue, above. 

The gap statistic $\hat{\lambda}_1-\hat{\lambda}_2$ is also a biased estimate for the true population value. And unlike $\hat{\lambda}_1$, the direction of the bias for the gap statistic differs depending on the value of $\lambda_1$. For $\lambda_1=1$, the gap statistic overestimates the true difference, while for  $\lambda_1>1$, the gap statistic underestimates the true difference (Supplementary Tables \ref{tab:bootBiasNormGap}-\ref{tab:bootBiasEllipUnifGap}); how large $\lambda_1$ needs to be before the bias becomes negative depends on the distribution of $X$. 

As in the case of the top eigenvalue, the bootstrap estimate of bias does not accurately estimate this bias (Supplementary Figure \ref{fig:bootBiasGap}). The bootstrap under-estimates the absolute size of the bias, and for elliptical distributions can misspecify the direction of the bias (Supplementary Tables \ref{tab:bootBiasNormGap}-\ref{tab:bootBiasEllipUnifGap}, Supplementary Figure \ref{fig:bootBiasGap}). As with the top eigenvalue, the disparity in the bootstrap estimate of bias improves as the top eigenvalue becomes more separated from the bulk. Estimating the variance of the gap statistic with the bootstrap shows similar problems, with the bootstrap widely over-estimating the variance of $(\lambda_1(\SigmaHat)-\lambda_2(\SigmaHat))$ in high dimensions (Supplementary Figure \ref{fig:bootVarGap}). Bootstrap confidence intervals also suffer from the same problem as those of the top eigenvalue: percentile CIs have low coverage of the truth in high dimensions and normal-based CI being much wider than necessary because of the over estimation of the variance.

\paragraph{Gap Ratio Statistic} Another alternative that tries to normalize the gap statistic is the gap ratio, $(\lambda_1-\lambda_2)/(\lambda_2-\lambda_3)$. Onatski \cite{OnatskiTalkAtMIT06} proposed tests based on these statistics to avoid having to estimate $\sigma_{n,p}$ and $\mu_{n,p}$ in Theorem \ref{thm:JohnstoneTW}, when $\Sigma$ is a multiple of the identity. In the scenario we are evaluating, this population quantity is not well defined ($\lambda_2-\lambda_3=0$), but the estimate and its distribution are well defined, and the statistic is  a tool for deciding whether $\lambda_1$ and $\lambda_2$ are well separated. Again, the bootstrap estimate gives poor estimates of various features of the actual distribution of the gap ratio statistic: the bootstrap estimate is biased and can either under or over estimate the variance, depending on the true value of $\lambda_1$ (Supplementary Figures \ref{fig:bootBiasGapRatio} and \ref{fig:bootVarGapRatio}). The bootstrap distribution does not appear to be converging, i.e the bootstrap distribution seems to change with each $X$ (Supplementary Figure \ref{fig:bootDensityGapRatioEllipExp}).

\section{Theoretical results}\label{sec:BootTheoryMain}
The problems with the bootstrap can be explained in part by the difference between the spectral behavior of weighted and unweighted covariance matrices when $p/n$ is not small. Specifically, bootstrapping the observations (rows) of $X$ is equivalent to randomly reweighting the observations which changes the spectral distribution of $\SigmaHat$. For example, if $X_i$'s are normally distributed, randomly weighting the $X_i$'s transforms the data to an elliptical distribution, which leads to a very different spectral distribution of eigenvalues when $p/n$ is not close to 0 (see Theorem \ref{supp:thm:bulkEllipDist} in the Supplementary material). Similarly, the distribution of the largest eigenvalues are dramatically affected by reweighting; the one exception to this rule is the situation where the largest eigenvalues of $\Sigma$ are very separated from the rest.

In what follows we provide theoretical results that help explain the results of our numerical simulations and also complete them. The first set of results concerns the impact of bootstrapping on the spectral distribution of a sample covariance matrix. We explain that this creates bias in the setting we consider and it helps explain some of the misbehavior of the bootstrap we observed in the numerical study. We then consider the case of extreme eigenvalues, in the case where the largest population eigenvalues are well-separated from the bulk of the eigenvalues. We show that then the bootstrap works asymptotically under certain conditions. This helps explain why the performance of the bootstrap improves in our numerical study when we increase the largest population eigenvalue.

\subsection{Bootstrapped empirical distribution}

\subsubsection{A theoretical result}
\begin{lemma}\label{lemma:bootD}
Suppose $\{X_i\}_{i=1}^n$ are fixed vectors in $\mathbb{R}^p$. Suppose $w_i$'s are independent random variables. 
Consider $S_w=\frac{1}{n}\sum_{i=1}^n w_i X_i X_i\trsp$ and call 
$$
m(z)=\frac{1}{p}\trace{(S_w-z\id_p)^{-1}}, \text{ for } z=u+iv \in \mathbb{C}^+\;.
$$

Then 
$$
P(|m(z))-\Exp{m(z)}|>t)\leq C \exp(-c p^2 v^2 t^2/n)\;,
$$
with $C=4$ and $c=1/16$ for instance.

The same result holds when $w_i$'s have $Mult(n,1/n)$ distribution. 
\end{lemma}

Naturally, the case of $Mult(n,1/n)$ corresponds to the standard bootstrap. We explore the statistical implications of this result in Section \ref{subsubsec:StatCsqBootBulk}.

\begin{corollary}\label{coro:howBootDPerforms}
When $p/n\tendsto \kappa \in (0,\infty)$, the Stieltjes transform $m(z)$ of the independently-weighted bootstrapped covariance matrix is asymptotically deterministic. The same is true with the standard bootstrap, where the weights have a Multinomial(n,1/n) distribution.

In particular, if $f$ is a bounded continuous function, as $n$ and $p$ tend to infinity, while $p/n\tendsto \kappa$, 
$$
\frac{1}{p}\sum_{i=1}^p f(\lambda_i^*)-\frac{1}{p}\sum_{i=1}^p \Expboot{f(\lambda_i^*)}\tendsto 0 \text{ in probability}\;,  
$$
where $\lambda_i^*$ are the decreasingly ordered bootstrapped eigenvalues and $\Expboot{\cdot}$ refers to expectation under the bootstrap distribution.
\end{corollary}
The corollary follows from our lemma simply by using the well-known fact that convergence of the Stieltjes transform implies convergence of the corresponding spectral distributions (see \cite{bai99}, \cite{geronimohill03}).

\subsubsection{Statistical consequences when $p/n$ is not small but remains bounded}\label{subsubsec:StatCsqBootBulk}
One of the main problem that the bootstrap exhibits in the setting we consider is bias. Let us give a concrete example.
\paragraph{Bias in bootstrap spectral distribution: the case of Gaussian data}
Suppose $X_i\iid {\cal N}(0,\Sigma)$ and suppose that $\SigmaHat=\frac{1}{n}\sum_{i=1}^n X_i X_i\trsp$. When bootstrapping, we effectively observe $\SigmaHat_w=\frac{1}{n}\sum_{i=1}^n w_i X_i X_i\trsp$, where $w_i$ is the number of times index $i$ is picked in our resample. It is known that the spectral distribution of $\SigmaHat$, $L_n(\SigmaHat)$, has a non-random limit, ${\cal L}(\Sigma)$ satisfying the so-called Marchenko-Pastur equation (see \cite{mp67},\cite{wachter78}, \cite{silverstein95},\cite{bai99} and Theorem \ref{supp:thm:MP} in the Supplementary Material). It is also known that $\SigmaHat_w$ when sampling both $w_i$'s and $X_i$'s has a non-random limiting distribution (for instance when $\{w_i\}_{i=1}^n$ are independent of $\{X_i\}_{i=1}^n$), ${\cal L}(\Sigma,w)$, which can be implicitly characterized by a pair of equations for a variant of the Stieltjes transform of $\SigmaHat_w$ (see \cite{BoutetKhorunVasilchuk96}, \cite{nekCorrEllipD} and Theorem \ref{supp:thm:bulkEllipDist} in the Supplementary Material). This limit distribution is in general hard to characterize analytically, however it is not the same as that of $\SigmaHat$, i.e ${\cal L}(\Sigma,w)\neq {\cal L}(\Sigma)$. Furthermore, it is clear by a simple conditioning argument, that, if $L_n(\SigmaHat_w)$ is the spectral distribution of $\SigmaHat_w$, 
$$
L_n(\SigmaHat_w) |\{X_i\}_{i=1}^n \weakCV {\cal L}(\Sigma,w) a.s 
$$
where the $a.s$ statement refers to the design matrix. 

$L_n(\SigmaHat_w) |\{X_i\}_{i=1}^n$ is the bootstrapped spectral distribution of $\SigmaHat$. Its limit, ${\cal L}(\Sigma,w)$ is different from that of $\SigmaHat$, ${\cal L}(\Sigma)$. \emph{Hence the bootstrapped distribution of the eigenvalues of $\SigmaHat$ is in general biased.} 

In connection with the results of \cite{PaulSilversteinExactSeparation09}, these results also help explain the problem of bias found in the bootstrapped extreme eigenvalues, when for instance $\Sigma=\id_p$, i.e no eigenvalues are well-separated from the bulk. 

\paragraph{Extensions} Much work has been done in random matrix theory to extend the domain of validity of the Marchenko-Pastur equation, which holds beyond the case of Gaussian data. The bootstrap bias problem remains the same, because the limiting properties of the matrices of interest are unaffected by the move from Gaussian to these more general models (see \cite{nekCorrEllipD}).

\paragraph{Bootstrap and geometry: an explanation of the problems with the bootstrap} At a high-level, one can intuitively think that bootstrapping moves the data in the setting considered here from a Gaussian setting to an elliptical one. It is well-known (\cite{DiaconisFreedmanProjPursuit84}, \cite{HallMarronNeemanJRSSb05}, \cite{nekCorrEllipD}) that in moderate and high-dimension elliptical distributions have completely different geometric properties than Gaussian ones and that this geometric features impact strongly the statistical behavior of many estimators, and in particular that of eigenvalues of sample covariance matrices (\cite{nekCorrEllipD}). As such, it is not that surprising that the bootstrap does not perform well: from an eigenvalue point of view, it is as if the bootstrap changed ``the geometry of the dataset" and this geometry has an important impact on their behavior. 

\emph{Bootstrapping is therefore not a good way to mimick the data generating process in this context.}

\paragraph{Other designs} Lemma \ref{lemma:bootD} and Corollary \ref{coro:howBootDPerforms} apply without restrictions on the design, which is one of their main strength. With further specifications, for instance assuming that the data is generated from an elliptical distribution, we could characterize precisely various aspects of the bootstrapped spectral distribution in relation to the empirical spectral distribution, using for instance Theorem \ref{supp:thm:bulkEllipDist}. However, this would be quite model-specific, and since we are focused on understanding more generally the problems with the bootstrap, we do not think these simple computations would serve our purpose. So we do not detail them here.

\subsection{Extreme eigenvalues}\label{subsec:bootExtremeEigenvalues}

\begin{mydef}
Suppose $\thetaHat_n(X_1,\ldots,X_n)$ is a statistic, $\thetaHat_n^*$ is its bootstrapped version. Suppose that $\thetaHat_n \weakCv T$. We say that \textit{the bootstrap is consistent in probability} if
$$
\thetaHat_n^* \weakCv_w T \; \; \text{ in } P_{X_1,\ldots,X_n}-\text{probability}.
$$
Here $\weakCv_w$ refers to weak convergence of $\thetaHat_n^*$ under the bootstrap weight distribution; the convergence in probability is with respect to the joint distribution of $X_1,\ldots,X_n$, which we denote $P_{X_1,\ldots,X_n}$. For simplicity, we often abbreviate $P_{X_1,\ldots,X_n}$ by $P$. 

We say that \textit{the bootstrap is strongly consistent} if
$$
\thetaHat_n^* \weakCv_w T \; \; \text{ a.s } P_{X_1,\ldots,X_n}.
$$

\end{mydef}

\subsubsection{Approximation results}

Call $S_n$ the sample covariance matrix of the data. We use the block notation
$$
S_n=\begin{pmatrix}
T_n& U_n\\
U_n\trsp  & n^{-\alpha} V_n
\end{pmatrix}
$$
Call $\Sigma_n$ the true covariance. We use the block notation
$$
\Sigma_n=
\begin{pmatrix}
\Sigma_{11} & \Sigma_{12}\\
\Sigma_{21} & n^{-\alpha} \Sigma_{22}
\end{pmatrix}\;.
$$
$T_n$ and $\Sigma_{11}$ are both assumed to be $q\times q$.

\vspace{2ex}
\textbf{Assumptions}
\vspace{-1ex}
\begin{itemize}
	\setlength\itemsep{-1.5em}
\item \textbf{A1} We assume that $\opnorm{\Sigma_{22}}=\gO(1)$ and that $\lambda_{min} (\Sigma_{11})>\eta>0$. We assume that $\Sigma_{11}$ is $q\times q$ with $q$ fixed. \\
\item \textbf{A2} $X_i$'s are i.i.d with $X_i=r_i Z_i$, where $Z_i\sim{\cal N}(0,\Sigma_n)$, and $0<\delta_0< r_i<\gamma_0$ is a bounded random variable independent of $Z_i$, with $\Exp{r_i^2}=1$. \\
\item \textbf{A3} The bootstrap weights $w_i$ have infinitely many moments, $\norm{w}_\infty=\gO(\polyLog(n))$ and $\Exp{w_i}=1$. These weights can either be independent or Multinomial$(n,1/n)$.\\
\item \textbf{A4} $p/n$ remains bounded as $n$ and $p$ tend to infinity. 
\end{itemize}

We then have the following theorem. 

\begin{theorem}\label{thm:keyApproxResults}
Under our assumptions A1-A4, if $\alpha>1/2+\eps$ for any $\eps>0$,
\begin{equation}\label{eq:keyApproxNonBoot}
\sup_{1\leq i \leq q} \sqrt{n} (\lambda_i(S_n)-\lambda_i(T_n))=\lo_P(1)\;.
\end{equation}

Furthermore, if $w$ denotes the vector of weights used in the bootstrap and the corresponding bootstrapped matrices are $S_n^*$ and $T_n^*$, we have

\begin{equation}\label{eq:keyApproxBoot}
\sup_{1\leq i \leq q} \sqrt{n} (\lambda_i(S^*_n)-\lambda_i(T^*_n))=\lo_{P,w}(1)\;.
\end{equation}

\end{theorem}

\subsubsection{Consequences for the bootstrap}
Recall some of the key results of \cite{BeranSrivastava85,BeranSrivastava87Correction} and \cite{EatonTyler91}, p.269: in the low-dimensional case, where $p$ is fixed and $n\tendsto \infty$, if all eigenvalues of $\Sigma$ are simple, the bootstrap distribution of the eigenvalues of $S_n$ is strongly consistent. On the other hand, it is known from these papers that the bootstrap distribution of the eigenvalues of $S_n$ is inconsistent when the eigenvalues of $\Sigma$ have multiplicities higher than 1. 

In light of these results, we have the following theorem

\begin{theorem}\label{thm:consistencyBootSimpleCase}
Suppose the eigenvalues of $\Sigma_{11}$ are simple and the assumptions A1-A4 of Theorem \ref{thm:keyApproxResults} hold. Then the bootstrap distribution of the $q$ largest eigenvalues of $S_n$ is consistent in probability. 
\end{theorem}

\subsubsection{Discussion of assumptions and remarks}

\paragraph{$\bm{\alpha>1/2+\eps}$} 
This assumption is not terribly restrictive: in fact under our assumptions, if $\Sigma_{22}=\id_{p-q}$, the fraction of variance explained by the top $q$ eigenvalues is, if $C=\trace{\Sigma_{11}}/q$ 
$$
\frac{\trace{\Sigma_{11}}}{\trace{\Sigma_{n}}}=\frac{qC}{qC+(p-q) n^{-\alpha}}\;.
$$
So if $1/2+\eps<\alpha<1$, in our asymptotics where $p/n$ remains bounded, this fraction of variance is asymptotically 0, since $pn^{-\alpha}\tendsto \infty$ and $qC$ is bounded. On the other hand, if $\alpha>1$, the fraction of variance explained by the top $q$ eigenvalues is approximately 1, which is the standard setting for the use of PCA. 

If $\alpha=1$, the fraction of variance varies between 0 and 1, depending on $C$. Of course, a similar analysis is possible and actually easy to carry out if $\Sigma_{22}$ is not a multiple of the identity. We leave those details to the interested reader. 
\vspace{-2ex}
\paragraph{Strong consistency of the bootstrap} We have chosen to present our results using convergence in probability statements, as we think they better reflect the questions encountered in practice. However, a quick scan through the proofs show that all the approximation results could be extended to a.s convergence: the random matrix results we rely on hold a.s, and the low-dimensional bootstrap results we use also hold a.s in low-dimension. 
\vspace{-2ex}
\paragraph{Distributional assumptions on $X_i$'s} The assumption that $Z_i\sim {\cal N}(0,\Sigma)$ is not critical: most of our arguments could be adapted to handle the case where $Z_i =\Sigma^{1/2} Y_i$, where $Y_i$ has independent, mean 0, variance 1 entries, with sufficiently many moments. This simply requires appealing to slightly different random matrix results that exist in the literature. Also, the first $q$ coordinates of $X_i$ could have a much more general distribution than the elliptical distributions we consider here, as our proof simply requires control of $\opnorm{V_n}$, which is where we appeal to random matrix theory. Doing this entails minor technical modifications to the proof, but since it might reduce clarity, we leave them to the interested reader.  
\vspace{-2ex}
\paragraph{Assumptions on $\Sigma$} The block representation assumptions are made for analytic convenience and can be easily dispensed of: eigenvalues are of course unaffected by rotations, so we simply chose to write $\Sigma$ in a basis that gave us this nice block format. As long as the ratio between the $q$-th eigenvalue of $\Sigma$ and $q+1$st is of order $n^{\alpha}$, our results hold. Furthermore, our results also handle a situation similar to ours, where for instance the top $q$ largest eigenvalues of $\Sigma$ grow like $n^{\alpha}$ and the $q+1$st is of order 1, by simple rescaling, using for instance the fact that $\trace{T_n}/\trace{\Sigma_{11}}\tendsto 1$ in probability. 
\vspace{-2ex}
\paragraph{The m-out-of-n bootstrap} As noted in \cite{BeranSrivastava87Correction,EatonTyler91}, subsampling approaches fix the problem of bootstrap inconsistency in the setting where $\Sigma_{11}$ has eigenvalues of multiplicity higher than 1. Our approximation results for the pair $(S_n^*,T_n^*)$ can be extended to subsampling approaches, and hence our results could be extended to cover these ideas. However, since this question is a bit distant from our main motivations, we do not treat it in detail.

\subsubsection{Eigenvalues that are not well separated}
Our results on the bootstrapped empirical spectral distribution apply here and in particular suggests that problems are likely to arise in practice. However, interesting questions concerning the fluctuation behavior of bootstrapped eigenvalues when the largest population eigenvalues are not well-separated from the bulk are very natural. For instance, is it the case that the bootstrap distribution of the largest eigenvalue of a sample covariance matrix is Tracy-Widom  when that is the case for the sampling distribution of the top eigenvalue? 

Though mathematically interesting, this question  does not seem so statistically important in light of the simulation study in Section \ref{sec:SimulationStudy}. For instance Figure \ref{fig:bootBias} indicates that the bootstrap estimate of bias of $\lambda_1(\SigmaHat)$ is itself very biased. Figure \ref{fig:bootDensityNull} and in particular Subfigure \ref{subfig:bootBias:EllipExp} suggests that the bootstrap distribution of our statistics is a poor approximation of the sampling distribution. It also suggests that the characteristics of the bootstrap distribution may depend strongly on the property of the design matrix, rendering their characterization mathematically intractable outside of simple situations of limited practical statistical interest. For this reason, we postpone this mathematically interesting and delicate question to possible future work.

\section{Conclusion}
We have investigated in this paper the properties of the bootstrap for spectral analysis of high-dimensional covariance matrices. We have shown numerically and through theoretical results that in general, the bootstrap does not provide accurate inferential results. The one exception of practical interest is the situation where we are interested in inference concerning only a few large eigenvalues of $\Sigma$, which are well-separated from the bulk of them and of multiplicity 1. In this case, the problem is effectively low dimensional (Theorem \ref{thm:keyApproxResults}), and the bootstrap works, because it is known to work in low-dimension when $\Sigma$ has eigenvalues of multiplicity 1. 

This confirms the findings of \cite{NEKEliBootRegression15}: in high-dimension, bootstrapping does not mimick the data generating process. Hence, standard bootstraps appear to work only for problems that are effectively low-dimensional.

\bigskip
\begin{center}
{\large\bf SUPPLEMENTARY MATERIAL}
\end{center}

\begin{description}

\item[Supplementary Text] More detailed description of simulations and proofs of the theorems stated in main text  (see below; see also authors' website for different formatting)

\item[Supplementary Figures] Supplementary Figures referenced in the main text (pdf; see also authors' website for different formatting)

\item[Supplementary Tables] Supplementary Tables referenced in the main text (pdf; see also authors' website for different formatting)

\end{description}

\clearpage

\appendix
\begin{center}
\textbf{\textsc{APPENDIX}}
\end{center}
\vspace{1cm}
\renewcommand{\thesection}{S\arabic{section}}
\renewcommand{\thelemma}{\arabic{section}-\arabic{lemma}}
\renewcommand{\thesubsection}{S\arabic{section}.\arabic{subsection}}
\renewcommand{\thesubsubsection}{S\arabic{section}.\arabic{subsection}.\arabic{subsubsection}}
\setcounter{equation}{0}  
\setcounter{lemma}{0}
\setcounter{corollary}{0}
\setcounter{section}{0}

\renewcommand{\thetable}{S\arabic{table}}   
\renewcommand{\thefigure}{S\arabic{figure}}
\setcounter{table}{0}
\setcounter{figure}{0}

\section{Description of Simulations and other Numerics} \label{Supp:Simulation}
For each of 1,000 simulations, we generate a $n\times p$ data matrix $X$. For each $X$, we calculate either the top eigenvalue (or the gap statistic) from the sample covariance matrix. Specifically, we perform the SVD of $X$ using the ARPACK numeric routines (implemented in the package \texttt{rARPACK} in R) to find the top five singular values of $X$ and get estimates $\hat{\lambda}_i$ by multiplying the singular values of $X$ by $1/n$. 

For each simulation, we perform bootstrap resampling of the $n$ rows of $X$ to get a bootstrap resample $X^{*b}$ and $\lambda_i^{*b}$; we repeat the bootstrap resampling $B=999$ times for each simulation, resulting in 999 values of $\lambda_i^{*b}$ for each simulation. 

In all simulations, we only consider the case where only $\lambda_1$ is allowed to differ from the rest of the eigenvalues. Therefore for all eigen values except the first, $\lambda_i=1$. For $\lambda_1$, we consider $\lambda_1=1+ c \sqrt{\frac{p}{n}}$ for the following $c$: 0, 0.9, 1.1, 1.5,   2.0,    3.0,    6.0,   11.0,   50.0,  100.0, and 1000.0 (not all of these values are shown in figures or tables accompanying this manuscript). The results shown in this manuscript set $n=1,000$, though $n=500$ was also simulated.

\paragraph{Generating $X$} We generate $X$ as $X=Z\Sigma$ where $\Sigma=V\Lambda V'$, and $\Lambda$ is a diagonal matrix of eigenvalues $\lambda_1\geq \ldots \geq \lambda_p$. We assume that there is no structure in the true eigenvectors, and generate $V$ as the right eigenvectors of the SVD of a nxp matrix with entries i.i.d $N(0,1)$.

$Z=DZ_0$ is a nxp matrix, with $Z_0$ having entries i.i.d. $N(0,1)$ and $D$ is a diagonal matrix $D$. If $D$ is the identity matrix, $Z$ will be i.i.d. normally distributed; otherwise $Z$ will be i.i.d with an elliptical distribution. We simulated under the following distributions for the diagonal entries of $D$ to create elliptical distribution for $Z$,
\begin{itemize}
\item $D_{ii}\sim N(0,1)$
 \item $D_{ii}\sim Unif(1/2, \frac{\sqrt{3}\sqrt{4-1/4}}{2}-\frac{1}{4})$
  \item $D_{ii}\sim Exp(\sqrt{2})$
\end{itemize}
In this manuscript, we concentrated only on $D_{ii}\sim Exp(\sqrt{2})$, the ``Elliptical Exponential'' distribution. This was because its behavior resulted in an elliptical distribution for $Z$ with properties the most different from when $Z$ is normal. The remaining choices for the distribution of $D$ result in elliptical distributions between that of the Elliptical Exponential and the Normal. The results from when $D_{ii}\sim Unif$ were generally fairly similar to when $Z$ is normal and results from when $D_{ii}\sim N(0,1)$ were more different, though not as extreme as the exponential weights.

\section{Review of existing results in random matrix theory}

\paragraph{Notations} We call $\lambda_1(M)$ or $\lambda_{max}(M)$ the largest eigenvalue of a symmetric matrix $M$. We call $\lambda_1(M)\geq \lambda_2(M)\geq \lambda_3(M)\geq \lambda_p(M)$ the ordered eigenvalues of the $p\times p$ matrix $M$. If $Z\sim {\cal N}_{\mathbb{C}}(0,\Sigma)$, $Z$ has a complex normal distribution, i.e $Z=\frac{1}{\sqrt{2}} (Z_1+i Z_2)$ where $Z_1$ and $Z_2$ are independent with $Z_i\sim {\cal N}(0,\Sigma)$. We call $\mathbb{C}^+$ the set of complex numbers with positive imaginary part.

\subsection{Bulk results}\label{supp:subsecRemindersBulkResults}
Bulk results are concerned with the spectral distribution of $\SigmaHat$, i.e the (random) probability measure with distribution
$$
dF_p(x)=\frac{1}{p}\sum_{i=1}^n \delta_{\lambda_i(\SigmaHat)}(x)\;.
$$

An efficient way to characterize the limiting behavior of $F_p$ is through its Stieltjes transform:
$$
\text{ for } z=u+iv \text{ with } v>0, \; \; m_p(z)=\frac{1}{p}\trace{(\SigmaHat-z\id_p)^{-1}}\;.
$$
Note that $m_p(z): \mathbb{C}^+\mapsto \mathbb{C}^+$. We have of course 
$$
m_{p}(z)=\int \frac{dF_p(\lambda)}{\lambda-z} =\frac{1}{p}\sum_{i=1}^p \frac{1}{\lambda_i(\SigmaHat)-z}\;.
$$

An important result in this area is the so-called Marchenko-Pastur equation \cite{mp67,silverstein95}, which states the following
\begin{theorem}\label{supp:thm:MP}
Suppose $X_i\iid \Sigma^{1/2}Z_i$, where $Z_i$ has $i.i.d$ entries, with mean 0, variance 1 and 4 moments. Suppose that the spectral distribution of $\Sigma$ has a limit $H$ in the sense of weak convergence of probability measures and $p/n\tendsto r\in (0,1)$. Then 
$$
F_p\weakCv F \text{a.s}\;,
$$
where $F$ is a deterministic probability distribution. 

Call $v_{p}(z)=(1-p/n)\frac{-1}{z}+\frac{p}{n} m_{p}(z)$. Then $v_p(z)\tendsto v_F(z)$ a.s. The Stieltjes transform of $F$ can be characterized through the equation 
$$
-\frac{1}{v_{F}(z)}=z-r\int \frac{\lambda
dH_{\infty}(\lambda)}{1+\lambda v_{F}(z)} \;, \forall z \in
\mathbb{C}^+
$$
\end{theorem}
At an intuitive level, this result means that the histogram of eigenvalues of $\SigmaHat$ is asymptotically non-random. Its shape is characterized by $F$ and its Stieltjes transform, $m$.

A generalization of this result to the case of elliptical predictors was obtained in \cite{nekCorrEllipD}. For the purpose of the current paper, the main result of \cite{nekCorrEllipD} states the following:
\begin{theorem}\label{supp:thm:bulkEllipDist}
Suppose $X_i\iid \Sigma^{1/2}Z_i$, where $Z_i$ has $i.i.d$ entries, with mean 0, variance 1 and 4 moments. Suppose that the spectral distribution of $\Sigma$ has a limit $H$, that $H$ has one moment and $p/n\tendsto r \in (0,\infty)$. Consider the matrix 
$$
B_n=\frac{1}{n}\sum_{i=1}^n w_i X_i X_i\trsp\;.
$$
Assume that the weights $\{w_i\}_{i=1}^n$ are independent of $X_i$'s. Call $\nu_n$ the empirical distribution of the weights $w_i$'s and suppose that $\nu_n\WeakCv \nu$. 
	
Then $B_n\WeakCv B$ a.s, where $B$ is a deterministic probability distribution; furthermore the Stieltjes transform of $B$, $m$, satisfies the system 
\begin{align*}
m(z)&=\int \frac{dH(\tau)}{\tau \int \frac{ w^2}{1+r
w^2 \gamma(z)}d\nu(w)-z} \; \text{ and }\\
\gamma(z)&=\int \frac{\tau dH(\tau)}{\tau \int
\frac{ w^2}{1+r
w^2 \gamma(z)}d\nu(w)-z} \;.
\end{align*}
where $\gamma(z)$ is the only solution of this equation mapping $\mathbb{C}^+$ into $\mathbb{C}^+$. 
\end{theorem}
Theorem \ref{supp:thm:bulkEllipDist} is interesting statistically because it shows that the limiting spectral distribution of weighted covariance matrices is completely different from that of unweighted covariance matrices, even when $\Exp{w_i}=1$. This is in very sharp contrast with the low-dimensional case. (Note that as shown in \cite{nekGencov}, Theorem \ref{supp:thm:bulkEllipDist} holds for many other distributions for $X_i$'s that the one mentioned in our statement.)

In the context of the current paper, this result is especially useful since bootstrapping a covariance matrix amounts to moving from an unweighted to a weighted covariance matrix.

\subsection{Edge results}\label{supp:subsecRemindersEdgeResults}
Edge results are concerned with the fluctuation behavior of the eigenvalues that are at the edge of the spectrum of the matrices of interest.

A now standard result is due to Johnstone \cite{imj}. 
\begin{theorem}\label{supp:thm:Johnstone}
Suppose $X_i\iid {\cal N}(0,\id_p)$. Assume that $p/n\tendsto r \in (0,1)$. Then 
$$
n^{2/3}\frac{\lambda_{max}(\SigmaHat)-\mu_{n,p}}{\sigma_{n,p}}\WeakCv TW_1\;.
$$		
\end{theorem}
We have for instance $\mu_{n,p}=(1+\sqrt{p/n})^2$ and $\sigma_{n,p}=(1+\sqrt{p/n}) (1+\sqrt{n/p})^{1/3}$. The result for the case $r\in(1,\infty)$ follows immediately by changing the role of $p$ and $n$; see \cite{imj} for details.

Following a question of Johnstone, Baik, Ben-Arous and P\'ech\'e obtained the following result \cite{bbap}.
\begin{theorem}\label{supp:thm:BBP}
Suppose $X_i\iid {\cal N}_{\mathbb{C}}(0,\Sigma)$. Suppose that $\lambda_{1}(\Sigma)=1+\eta \sqrt{p/n}$ and $\lambda_{i}(\Sigma)=1$ for $i>1$. 
\begin{enumerate}
\item If $0<\eta<1$, then 
$$
n^{2/3}\frac{\lambda_{1}(\SigmaHat)-\mu_{n,p}}{\sigma_{n,p}}\WeakCv TW_2\;.
$$		
\item On the other hand, if $\eta>1$, then 
$$
\sqrt{n}\frac{\lambda_{1}(\SigmaHat)-\mu_{\eta,n,p}}{\sigma_{\eta,n,p}}\WeakCv {\cal N}(0,1)\;.
$$
\end{enumerate}
\end{theorem}
Here, 
$$
\mu_{\eta,n,p}=\lambda_1 (1+\frac{\sqrt{p/n}}{\eta}) \text{ and } \sigma_{\eta,n,p}= \lambda_1 \sqrt{1-\eta^{-2}}\;.
$$
We note that we can rewrite the previous quantities solely as functions of $\lambda_1$, specifically 
$$
\mu_{\eta,n,p}=\lambda_1 \left(1+\frac{p/n}{\lambda_1-1}\right) \text{ and } \sigma_{\eta,n,p}= \lambda_1 \sqrt{1-\frac{n}{p}(\lambda_1-1)^{-2}}\;.
$$
This representation shows that $\mu_{\eta,n,p}$ is an increasing function of $\lambda_1$ on $(1+\sqrt{p/n},\infty)$ and therefore it would be easy to estimate $\lambda_1(\Sigma)$ from $\lambda_1(\SigmaHat)$. In particular, it is very simple to build confidence intervals in this context. 

\paragraph{Interpretation of Theorem \ref{supp:thm:BBP}} In other words, there is a phase-transition: if $\lambda_1$ is sufficiently large, i.e larger than $1+\sqrt{p/n}$ the largest eigenvalue of $\SigmaHat$ has Gaussian fluctuations. If it is not large enough, i.e smaller than $1+\sqrt{p/n}$, the fluctuations are Tracy-Widom, and in fact are the same if $\lambda_1(\Sigma)=1$. Statistically, it is hard to build confidence intervals for $\lambda_1$ in the latter case - but it is very easy to do so in the first case where $\lambda_1$ is sufficiently large. 

Similar results were obtained in \cite{nekGencov} for general $\Sigma$ in the complex Gaussian case and extended to the real case in \cite{LeeSchnelli14}. \cite{nek04} showed that Theorem \ref{supp:thm:Johnstone} holds when $p/n\tendsto 0$ and $p\tendsto \infty$ at any rate. See also the interesting \cite{debashis} and \cite{JohnstoneReview07}. We finally note that the main result in \cite{bbap} is slightly more general than Theorem \ref{supp:thm:BBP} but we just need that version for the current paper.

\section{Proofs}

\subsection{Proof of Lemma \ref{lemma:bootD}}

\begin{proof}[Proof of Lemma \ref{lemma:bootD}]
We recall the following result from a simple application of the Sherman-Morrison-Woodbury formula (see \cite{hj} and \cite{bai99}): if $M$ is a symmetric matrix, $q$ is a real vector $v>0$ and and $z=u+iv\in\mathbb{C}^+$, 
$$
|\trace{(M+qq\trsp -z\id_p)^{-1}}-\trace{(M-z\id_p)^{-1}}|\leq \frac{1}{v}\;.
$$

We use bounded martingale difference arguments as in \cite{GoetzeTikhomirov05}, \cite{PajorPasturPub09}, \cite{nekCorrEllipD}. 

$\bullet$ \textbf{Case 1: independent weights $w_i$}\\
Consider the filtration $\{{\cal F}_i\}_{i=0}^n$, with ${\cal F}_i=\sigma(w_1,\ldots,w_i)$ - the $\sigma$-field generated by $w_1,\ldots,w_i$ - and ${\cal F}_0=\emptyset$.

Call $S^{(i)}_w=S_w-\frac{1}{n}w_i X_i X_i\trsp-z\id_p$. In light of the result we just mentioned, 
$$
\frac{1}{p}|\trace{[S^{(i)}_w]^{-1}}-\trace{[S_w]^{-1}}|\leq \frac{1}{pv}\;.
$$
In particular, this implies, since $\Exp{\trace{[S^{(i)}_w]^{-1}}|{\cal F}_i}=\Exp{\trace{[S^{(i)}_w]^{-1}}|{\cal F}_{i-1}}$ that 
$$
\frac{1}{p}|\Exp{\trace{[S_w]^{-1}}|{\cal F}_i}-\Exp{\trace{[S_w]^{-1}}|{\cal F}_{i-1}}|\leq \frac{2}{pv}\;.
$$
Hence, $d_i=\frac{1}{p}\Exp{\trace{[S_w]^{-1}}|{\cal F}_i}-\frac{1}{p}\Exp{\trace{[S_w]^{-1}}|{\cal F}_{i-1}}$ is a bounded martingale-difference sequence. We can therefore apply Azuma's inequality (\cite{ledoux2001}, p. 68), to get 
$$
P(|m_p(z)-\Exp{m_p(z)}|>t)\leq C \exp(-c \frac{p^2 v^2 t^2}{n})\;.
$$
In \cite{nekCorrEllipD} it is shown that we can take $C=4$ and $c=1/16$. 

$\bullet$ \textbf{Case 2: multinomial weights}\\
In this case, the previous result cannot be applied directly because the weights are not independent, since they must sum to $n$. However, to draw according to a Multinomial($n,1/n$), we can simply pick an index from $\{1,\ldots,n\}$ uniformly and repeat the operation $n$ times independently. Let $I(k)$ be the value of the index picked on the $k$-th draw from our sampling scheme.
Clearly, the bootstrapped covariance matrix can be written as 
$$
S_w=\frac{1}{n}\sum_{k=1}^n X_{I(k)}X_{I(k)}\trsp\;.
$$ 
Consider the filtration $\{{\cal F}_i\}_{i=0}^n$, with ${\cal F}_i=\sigma(I(1),\ldots,I(i))$ - the $\sigma$-field generated by $I(1),\ldots,I(i)$ - and ${\cal F}_0=\emptyset$. Clearly $S_w$ is a sum of rank-1, independent matrices. So, if $S_w(k)=S_w-X_{I(k)}X_{I(k)}\trsp/n$, 
$$
\Exp{\trace{[S^{(i)}_w]^{-1}}|{\cal F}_i}=\Exp{\trace{[S^{(i)}_w]^{-1}}|{\cal F}_{i-1}}\;.
$$
The same argument as above therefore applies and the theorem is shown. 
\end{proof}

\subsection{Proof of Theorem \ref{thm:keyApproxResults} and \ref{thm:consistencyBootSimpleCase} }
\begin{proof}[Proof of Theorem \ref{thm:keyApproxResults}]
Recall that Wielandt's Theorem (see p.261 in \cite{EatonTyler91}) gives 
$$
\sup_{1\leq i \leq q} \, \, 0\leq \lambda_i(S_n)-\lambda_i(T_n)\leq \frac{\lambda_{max}(U_n U_n\trsp)}{\lambda_q(T_n)-n^{-\alpha}\lambda_{max}(V_n)}\;,
$$	
provided $\lambda_q(T_n)>n^{-\alpha}\lambda_{max}(V_n)$.

Recall that the Schur complement formula gives
$$
n^{-\alpha} V_n \succeq U_n\trsp T_n^{-1} U_n \succeq U_n\trsp U_n /\lambda_{max}(T_n)\;,
$$
where the second inequality is a standard application of Lemma V.1.5 in \cite{bhatia97}. Since $\lambda_{max}(U_n U_n\trsp)=\lambda_{max}(U_n\trsp U_n)$ by simply writing the singular value decomposition of $U_n$,  we conclude that 
$$
\lambda_{max}(T_n) n^{-\alpha} \opnorm{V_n}\geq \lambda_{max}(U_n U_n\trsp)\;.
$$

So we conclude that provided $\lambda_q(T_n)>n^{-\alpha}\lambda_{max}(V_n)$,
\begin{equation}\label{eq:keyBound}
\sup_{1\leq i \leq q} \, \, 0\leq \lambda_i(S_n)-\lambda_i(T_n)\leq n^{-\alpha} \frac{\lambda_{max}(V_n)}{\lambda_q(T_n)-n^{-\alpha}\lambda_{max}(V_n)}\;.
\end{equation}

$\bullet$ \textbf{Proof of Equation \eqref{eq:keyApproxNonBoot}}
Note that under assumption \textbf{A2}, standard results in random matrix theory \cite{geman80,silverstein85,silverstein89} guarantee that $\opnorm{V_n}=\gO_P(1)$. Furthermore, standard results in classic multivariate analysis \cite{anderson03} show that 
$\lambda_q(T_n)\tendsto \lambda_q(\Sigma_{11})$ in probability. Hence, we have 
$$
\frac{\lambda_{max}(V_n)}{\lambda_q(T_n)-n^{-\alpha}\lambda_{max}(V_n)}=\gO_P(1)\;.
$$
We therefore have 
$$
\sup_{1\leq i \leq q}\sqrt{n} (\lambda_i(S_n)-\lambda_i(T_n))=\gO_P(n^{1/2-\alpha})\;.
$$

$\bullet$ \textbf{Proof of Equation \eqref{eq:keyApproxBoot}} We note that if $D_w$ is the diagonal matrix with the bootstrap weights on the diagonal, we have
$$
S_n^*=\frac{1}{n}X\trsp D_w X\;.
$$

Therefore, we see that $\opnorm{T_n^*}=\gO_{P,w}(1)$, $\lambda_q(T_n^*)\rightarrow_{P,w} \lambda_q(\Sigma_{11})$ by the law of large numbers (provided $\Exp{w_i}=1$; the case of Multinomial($n,1/n$) weights is also easy to deal with by the technique described in the previous subsection for instance) and $\opnorm{V^*_n}=\gO_{P,w}(\polyLog (n))$ provided $\norm{w}_\infty=\polyLog (n)$.

We can then conclude that
$$
\sup_{1\leq i \leq q} \sqrt{n} (\lambda_i(S^*_n)-\lambda_i(T^*_n))=\lo_{P,w}(1)\;.
$$

\end{proof}

\begin{proof}[Proof of Theorem \ref{thm:consistencyBootSimpleCase}]
The results from Theorem \ref{thm:keyApproxResults} imply that 
\begin{equation}\label{eq:approxBetweenBootDistributions}
\sup_{1\leq i \leq q} \left|\left[\sqrt{n}(\lambda_i(S_n^*)-\lambda_i(S_n))\right]-\left[\sqrt{n}(\lambda_i(T_n^*)-\lambda_i(T_n))\right]\right|=
\lo_{P,w}(1)\;.
\end{equation}

The arguments used in the proof of Theorem \ref{thm:keyApproxResults} also apply to $\Sigma$ and show that 
$$
\sup_{1\leq i \leq q} |\lambda_i(\Sigma_n)-\lambda_i(\Sigma_{11})|\leq n^{-\alpha}\frac{\lambda_{max}(\Sigma_{22})\lambda_{max}(\Sigma_{11})}{\lambda_q(\Sigma_{11})-n^{-\alpha}\lambda_{max}(\Sigma_{22})}
$$
Hence, when $\alpha>1/2+\eps$, we have 
$$
\sqrt{n} \sup_{1\leq i \leq q} |\lambda_i(\Sigma_n)-\lambda_i(\Sigma_{11})|=\lo(1)\;.
$$
Therefore, 
$$
\sup_{1\leq i \leq q} \sqrt{n}\left|\lambda_i(S_n)-\lambda_i(\Sigma_{n})-[\lambda_i(T_n)-\lambda_i(\Sigma_{11})]\right|=\lo_P(1)\;.
$$
Hence, the $q$ largest eigenvalues of $S_n$ have the same limiting fluctuation behavior as the $q$ largest eigenvalues of $T_n$ (classical results \cite{anderson03} show that $\sqrt{n}$ is the correct order of fluctuations). The same is true for the bootstrapped version of their distributions, according to Equation \eqref{eq:approxBetweenBootDistributions}. 

Since the results of \cite{BeranSrivastava85,BeranSrivastava87Correction,EatonTyler91} show consistency of the bootstrap distribution of the eigenvalues of $T_n$, this result carries over to the bootstrap distribution of $S_n$. 
\end{proof}

\clearpage
\begin{center}
\textbf{\textsc{Supplementary Figures}}
\end{center}
\vspace{1cm}

\newlength{\doubleFig}
\setlength{\doubleFig}{.45\textwidth}
\begin{figure}[t]
	\centering
	\subfloat[][$\lambda_1^{Alt}=1+3\sqrt{r}$, r=0.01]{\includegraphics[type=pdf,ext=.pdf,read=.pdf,width=\doubleFig]{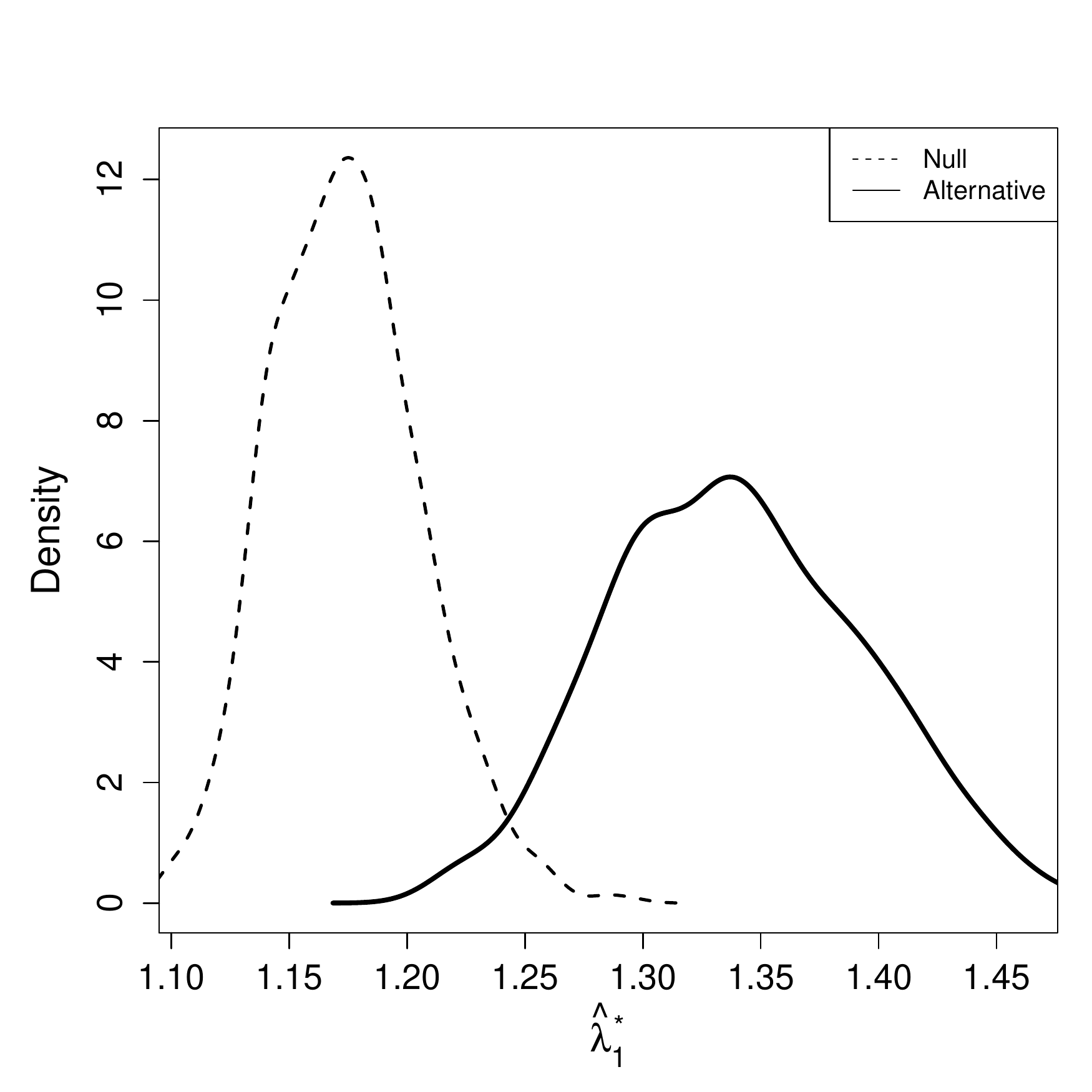} \label{subfig:trueBiasNormZ:0.01_3} }
	\subfloat[][$\lambda_1^{Alt}=1+11\sqrt{r}$, r=0.01]{\includegraphics[type=pdf,ext=.pdf,read=.pdf,width=\doubleFig]{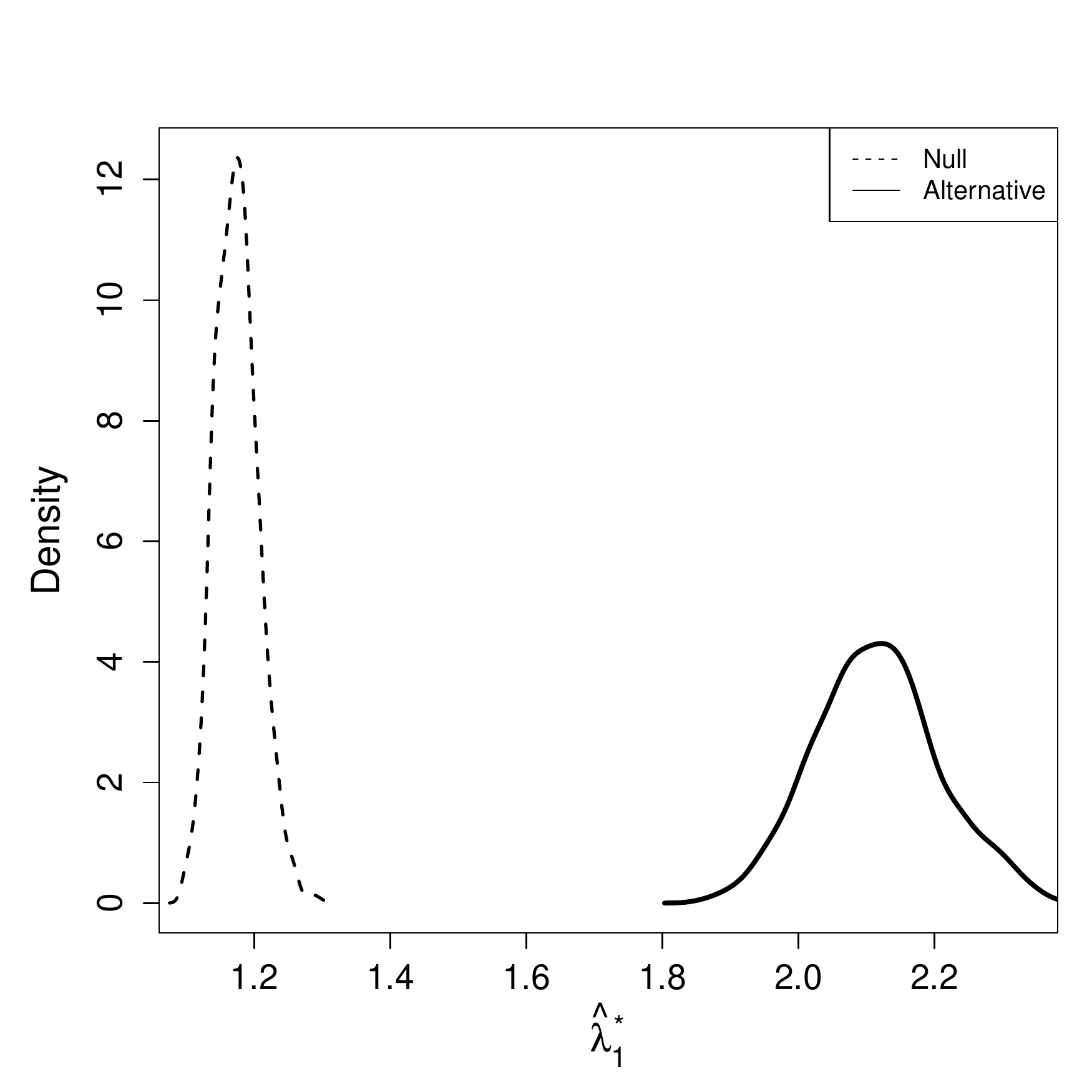}  \label{subfig:trueBiasNormZ:0.01_11} }
	\\
	\subfloat[][$\lambda_1^{Alt}=1+3\sqrt{r}$, r=0.03]{\includegraphics[type=pdf,ext=.pdf,read=.pdf,width=\doubleFig]{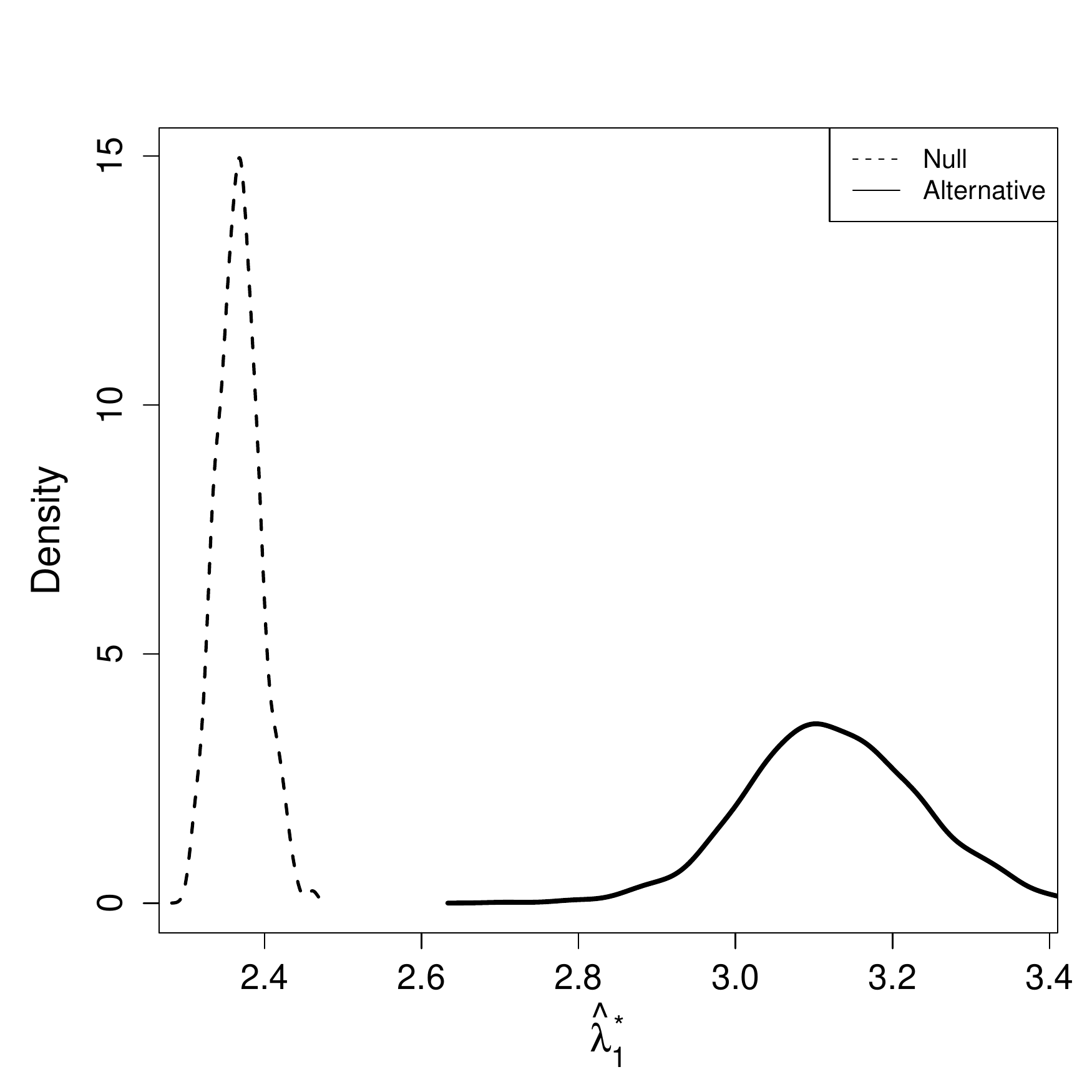}  \label{subfig:trueBiasNormZ:0.3_3} }
	\subfloat[][$\lambda_1^{Alt}=1+11\sqrt{r}$, r=0.3]{\includegraphics[type=pdf,ext=.pdf,read=.pdf,width=\doubleFig]{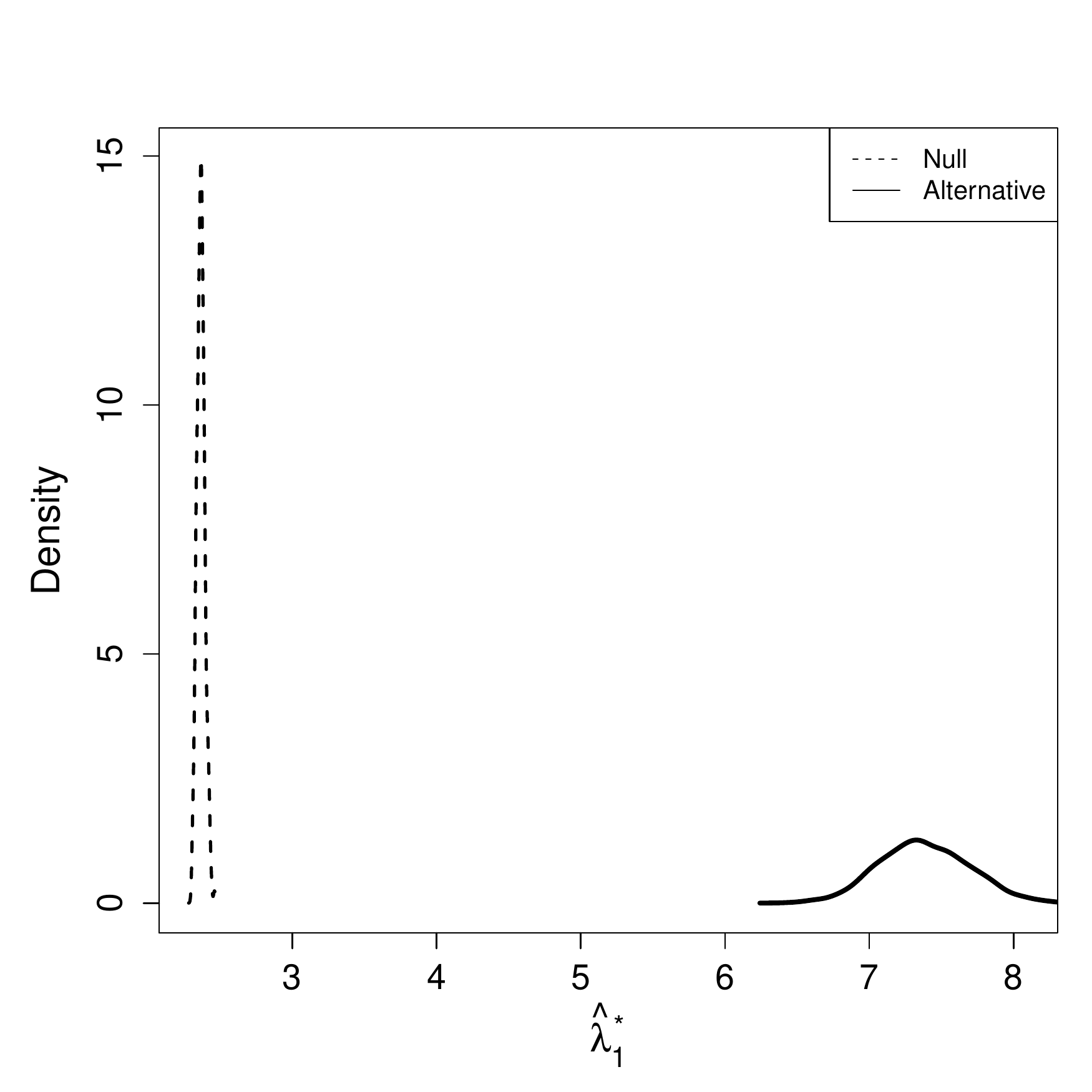} \label{subfig:trueBiasNormZ:0.3_11} }
\caption{
 \textbf{Top Eigenvalue: Distribution of Largest Eigenvalue, Null versus Alternative, $X_i\sim$ Normal } 
}\label{fig:trueBiasNormZ}
\end{figure}

\begin{figure}[t]
	\centering
	\subfloat[][$\lambda_1^{Alt}=1+3\sqrt{r}$, r=0.01]{\includegraphics[type=pdf,ext=.pdf,read=.pdf,width=\doubleFig]{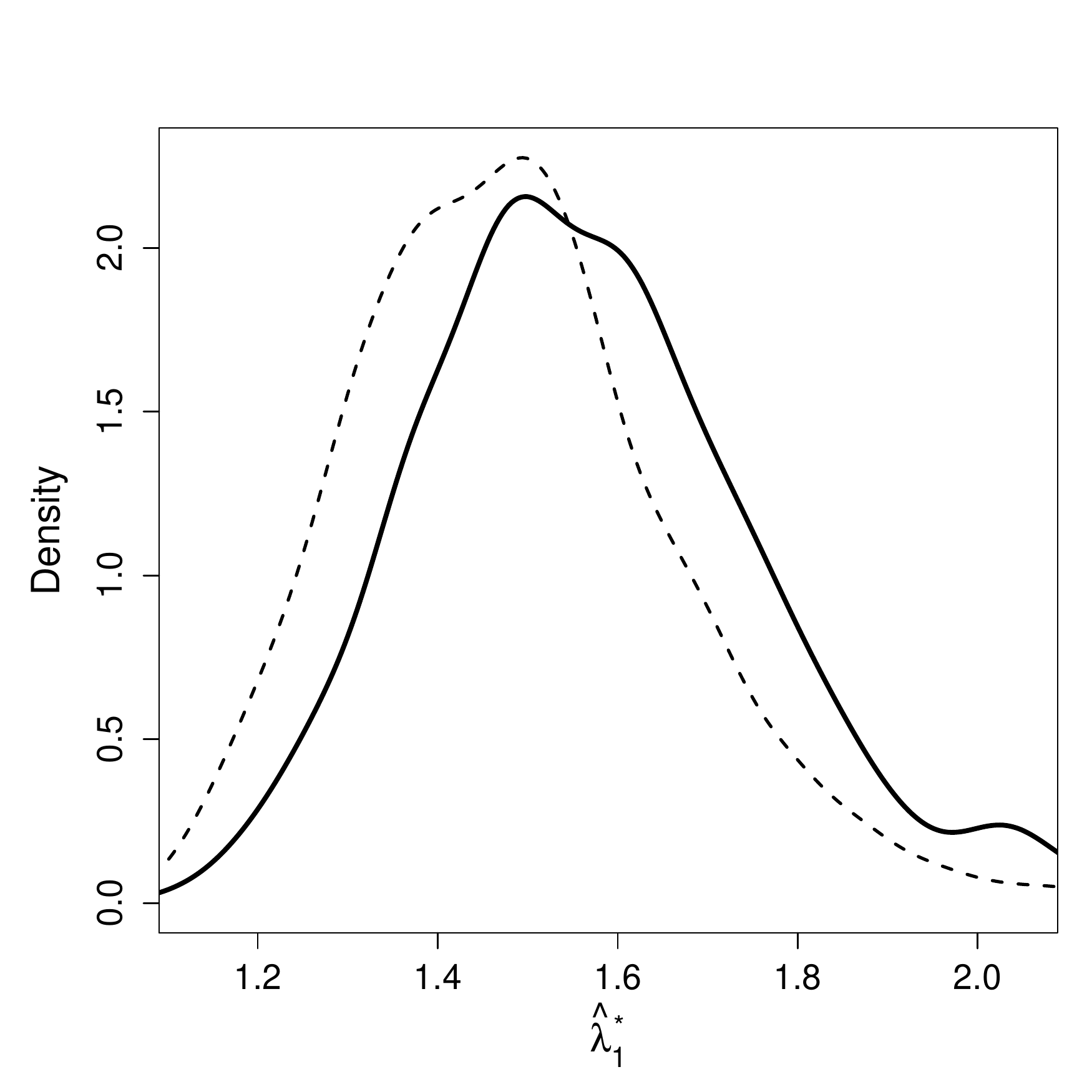} \label{subfig:trueBiasNormZ:0.01_3} }
	\subfloat[][$\lambda_1^{Alt}=1+11\sqrt{r}$, r=0.01]{\includegraphics[type=pdf,ext=.pdf,read=.pdf,width=\doubleFig]{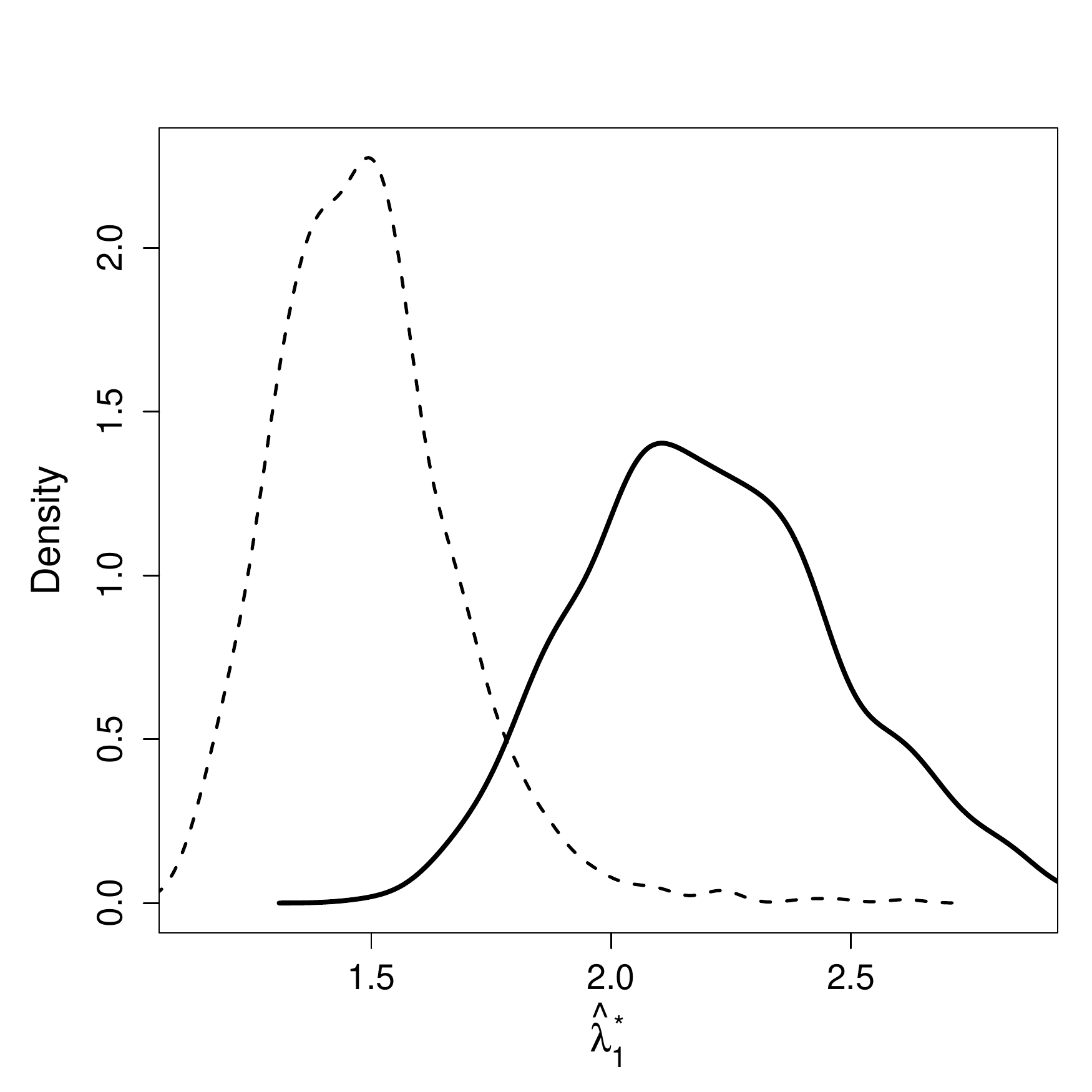}  \label{subfig:trueBiasNormZ:0.01_11} }
	\\
	\subfloat[][$\lambda_1^{Alt}=1+3\sqrt{r}$, r=0.03]{\includegraphics[type=pdf,ext=.pdf,read=.pdf,width=\doubleFig]{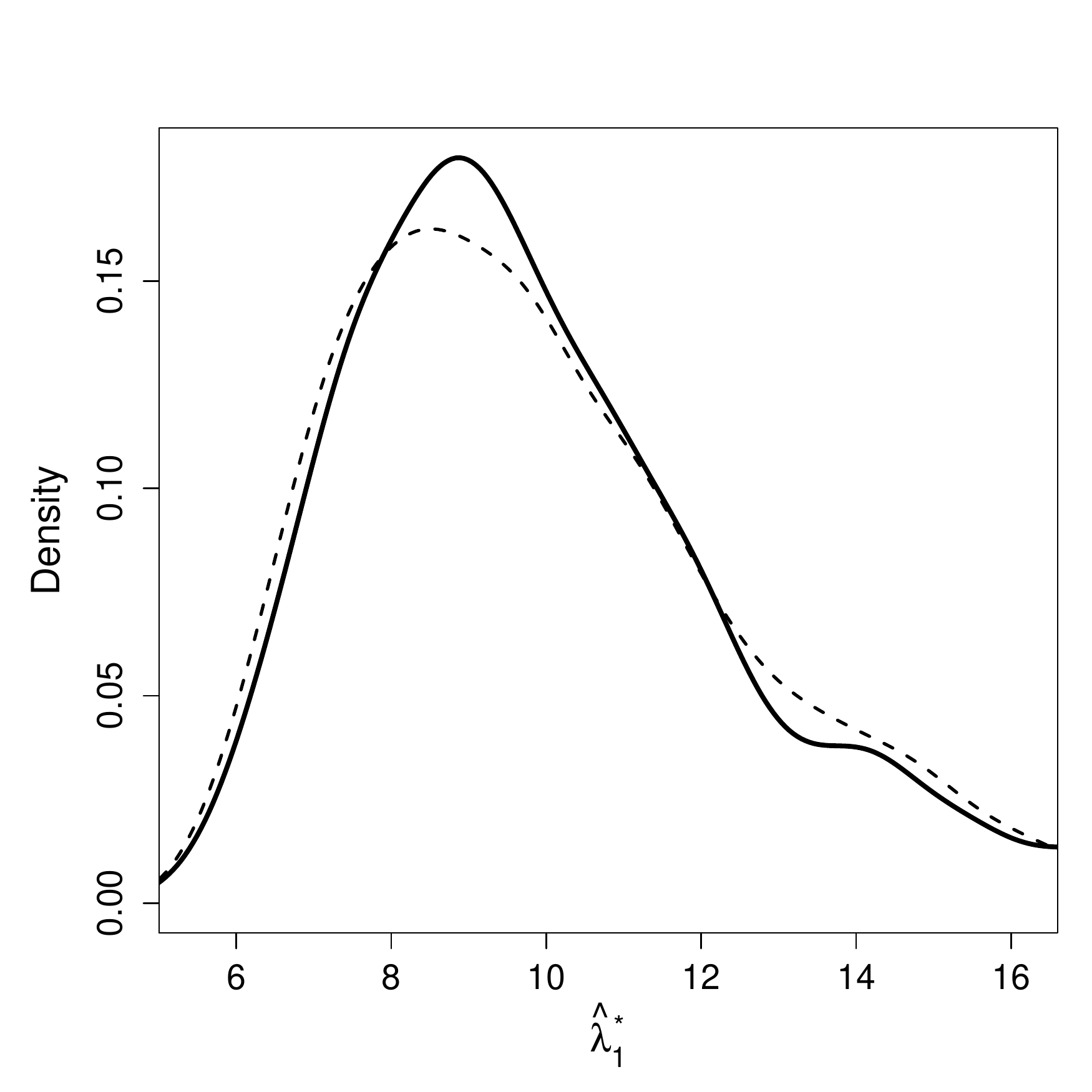}  \label{subfig:trueBiasNormZ:0.3_3} }
	\subfloat[][$\lambda_1^{Alt}=1+11\sqrt{r}$, r=0.3]{\includegraphics[type=pdf,ext=.pdf,read=.pdf,width=\doubleFig]{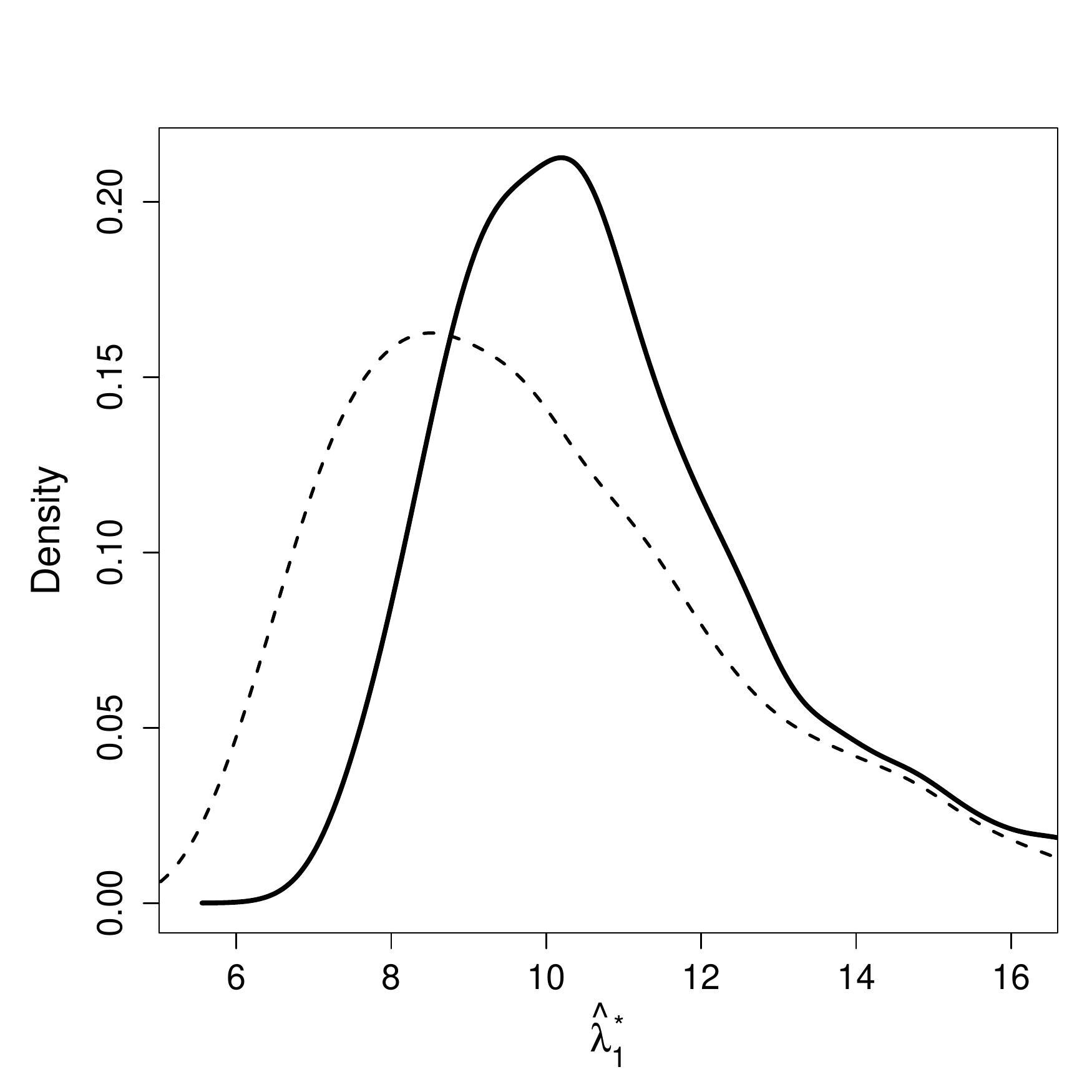} \label{subfig:trueBiasNormZ:0.3_11} }
\caption{
 \textbf{Top Eigenvalue: Distribution of Largest Eigenvalue,  Null versus Alternative, $X_i\sim$ Ellip. Exp } 
}\label{fig:trueBiasNormZ}
\end{figure}	


\begin{figure}[t]
	\centering
	\subfloat[][$X_i\sim$ Normal]{\includegraphics[width=\doubleFig]{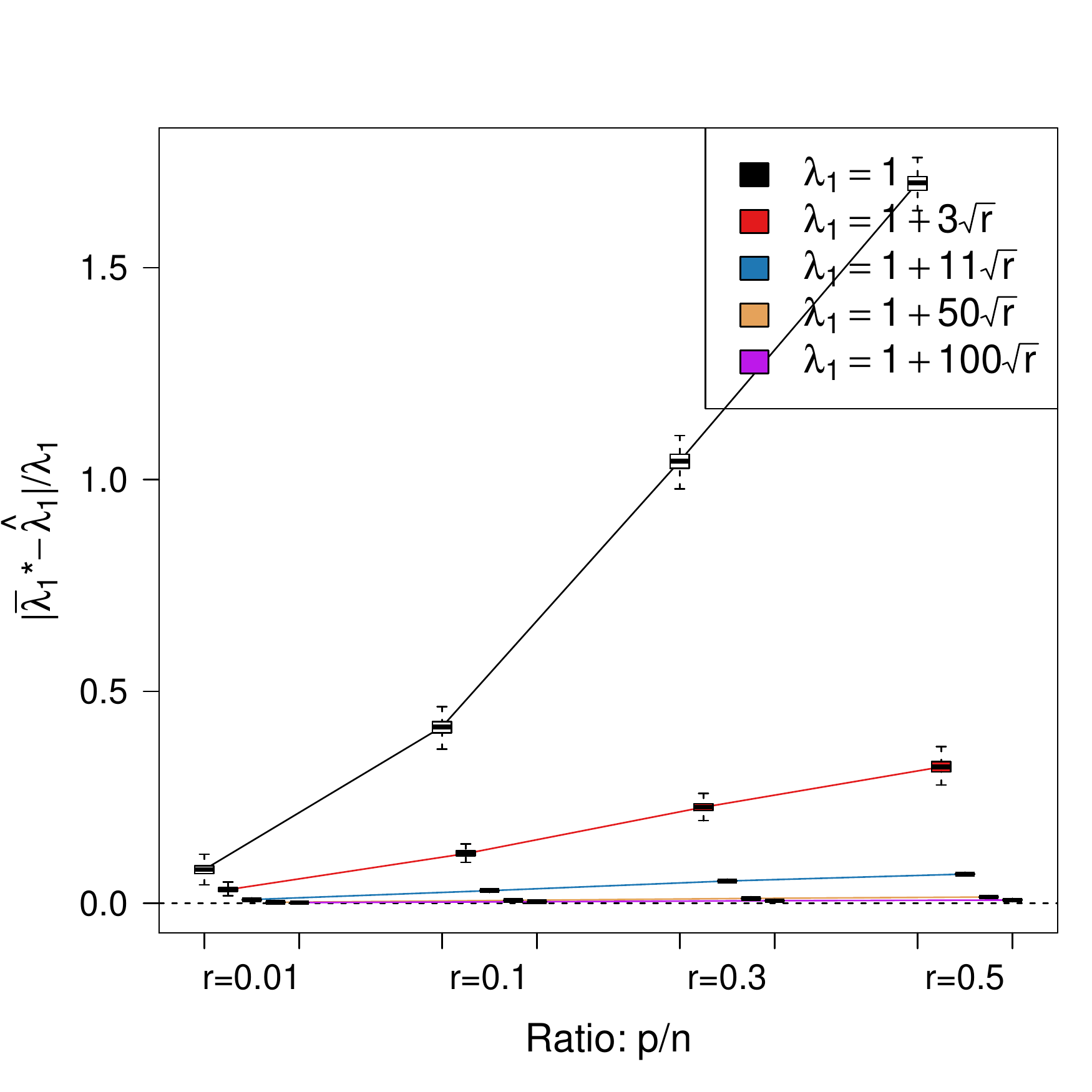} \label{subfig:bootBigBias:Norm} }
	\subfloat[][$X_i\sim$ Ellip. Exp]{\includegraphics[width=\doubleFig]{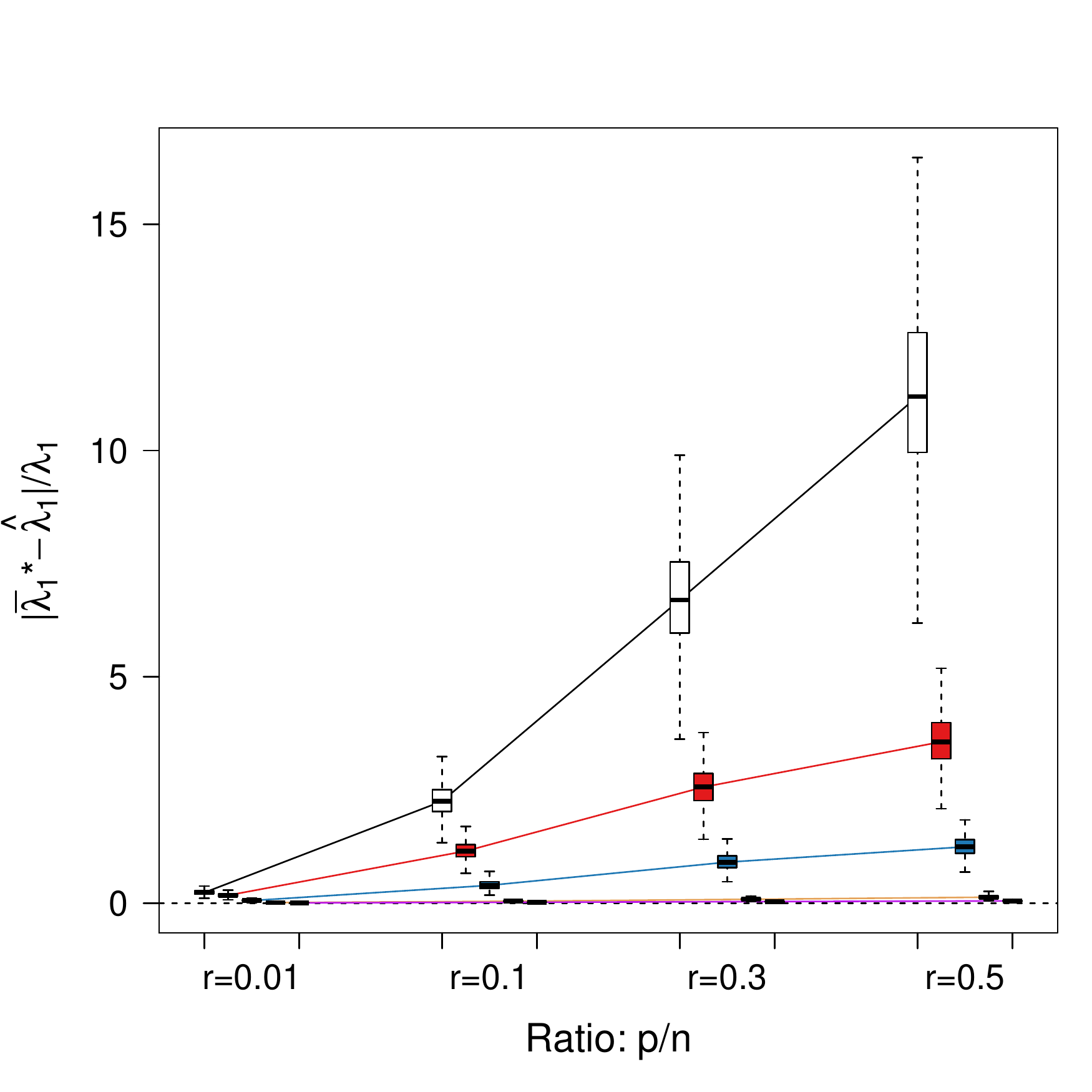} 
	\label{subfig:bootBigBias:EllipExp} }\\
	\subfloat[][$X_i\sim$ Ellip. Normal]{\includegraphics[width=\doubleFig]{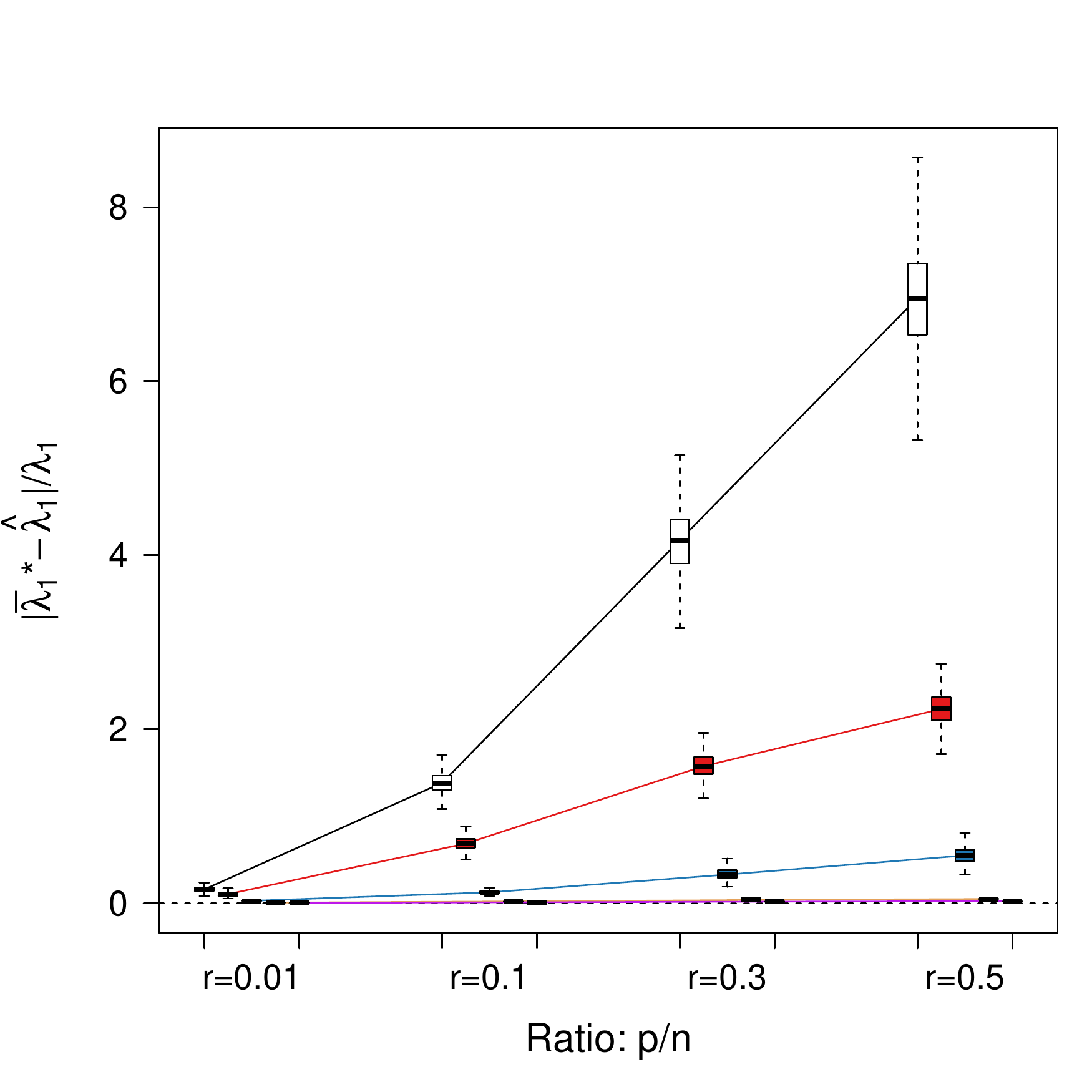} 
	\label{subfig:bootBigBias:EllipNorm} }
	\subfloat[][$X_i\sim$ Ellip. Uniform]{\includegraphics[width=\doubleFig]{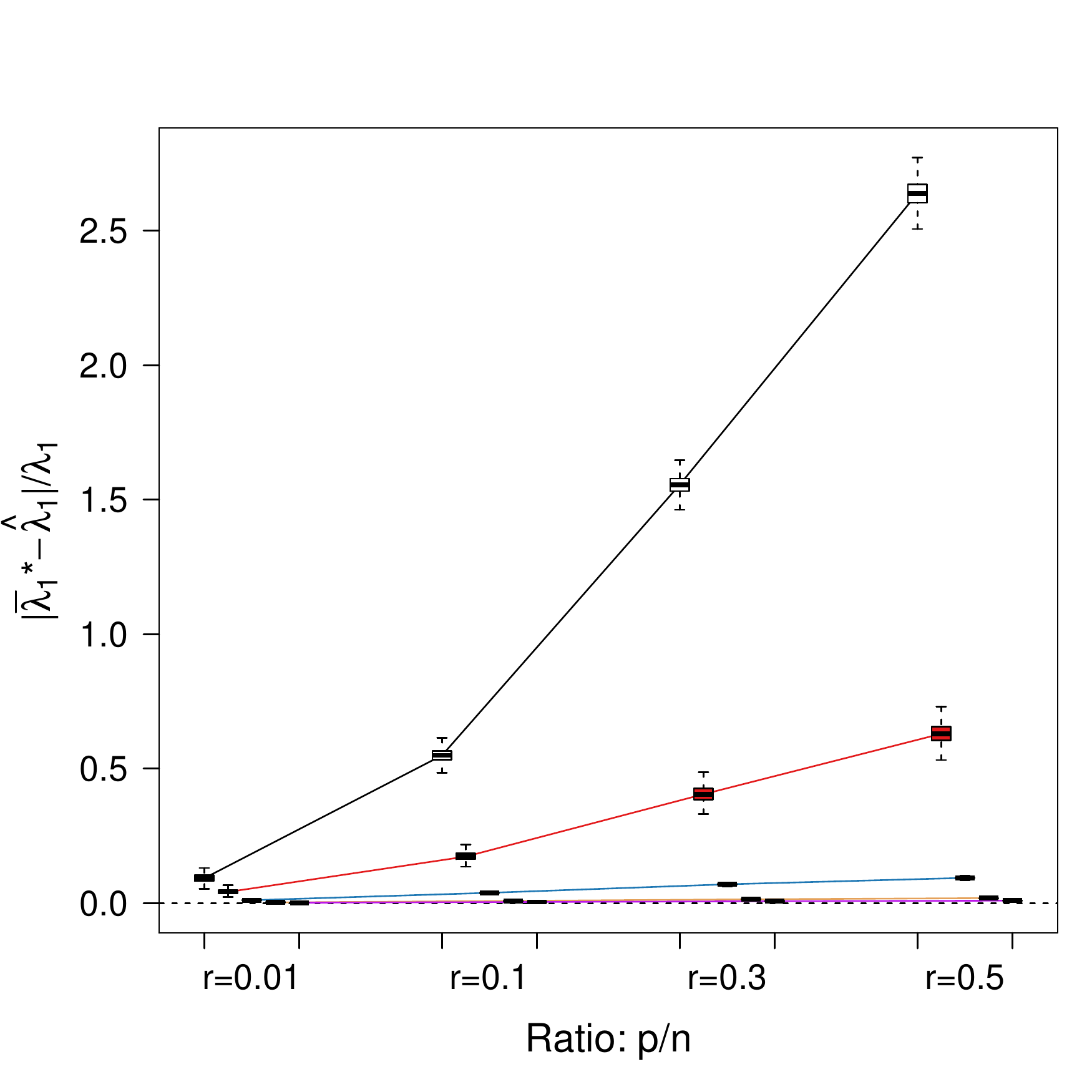} 
	\label{subfig:bootBigBias:EllipUnif} }
\caption{
 \textbf{Top Eigenvalue: Bias of Largest Bootstrap Eigenvalue, n=1,000: } Plotted are boxplots of the bootstrap estimate of bias covering larger values of $\lambda_1$ than shown in the main text. Unlike the main text, here we scale the bias by the true $\lambda_1$ so as to make the comparisons more comparable ($|\bar{\lambda}_1^*-\hat{\lambda}_1|/\lambda_1$). See the legend of Figure \ref{fig:bootBias} in the main text for more details.
}\label{fig:bootBigBias}
\end{figure}

\begin{figure}[t]
	\centering
	\subfloat[][$X_i\sim$ Normal]{\includegraphics[width=\doubleFig]{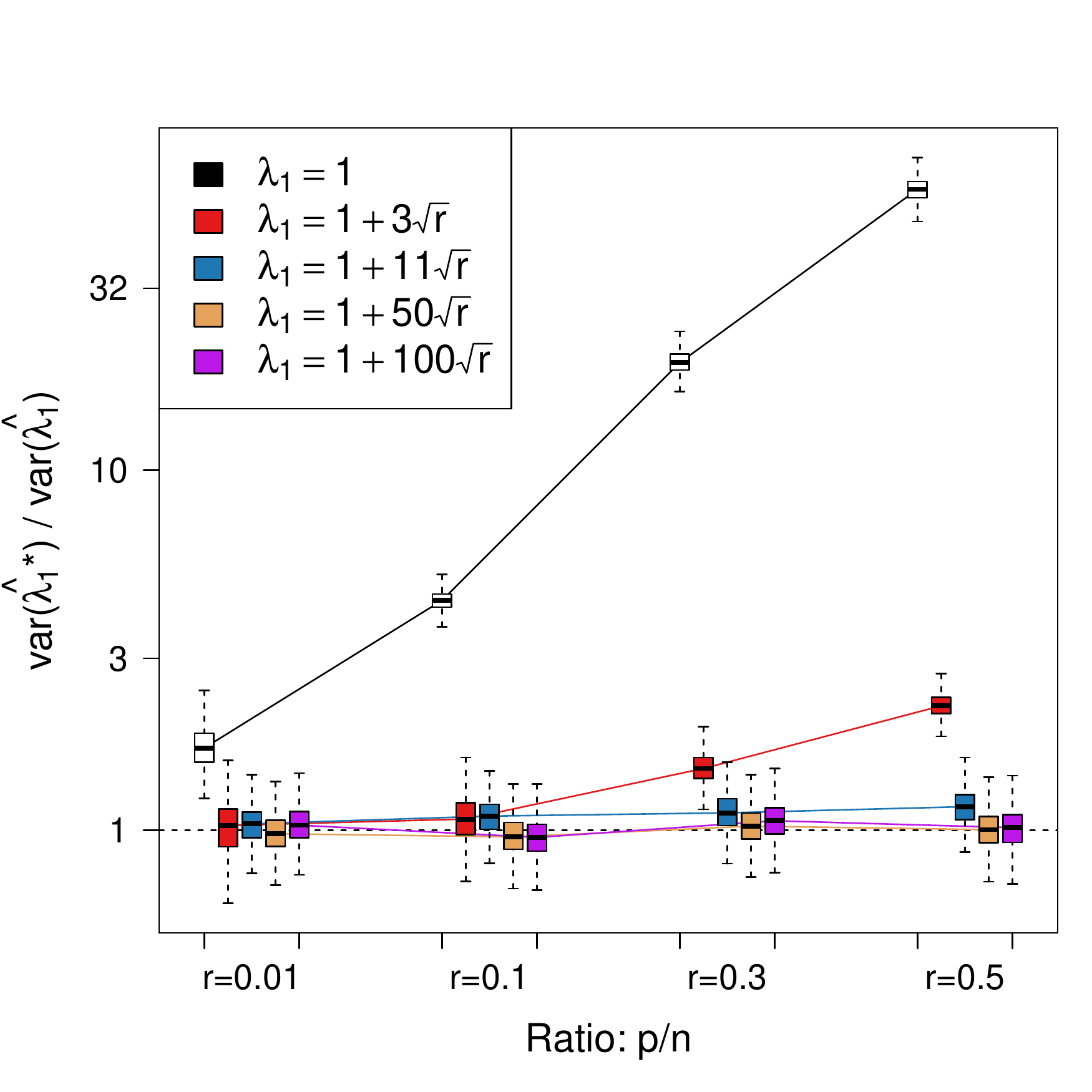} \label{subfig:bootBigVar:Norm} }
	\subfloat[][$X_i\sim$ Ellip. Exp]{\includegraphics[width=\doubleFig]{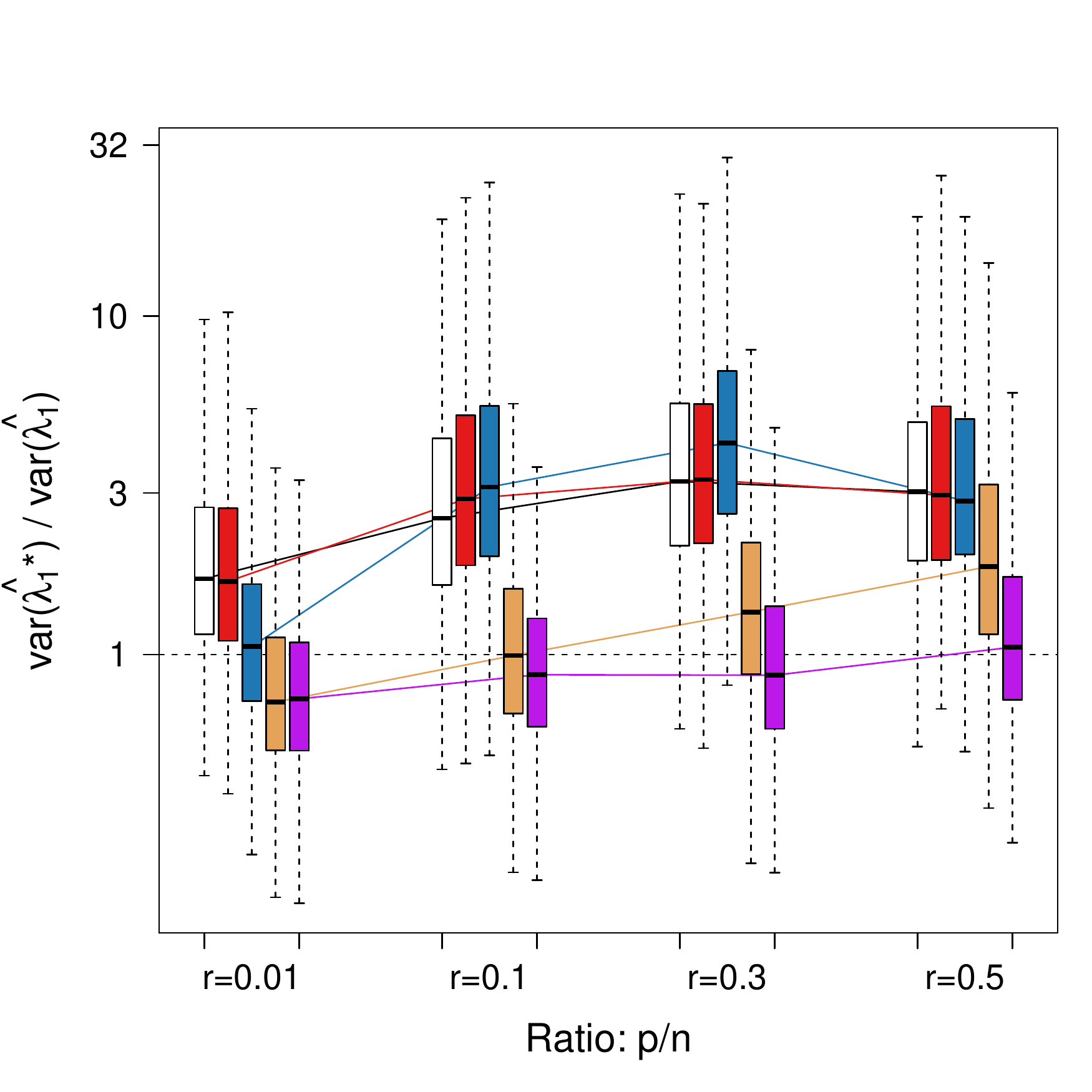} \label{subfig:bootBigVar:EllipExp} }\\
	\subfloat[][$X_i\sim$ Ellip Normal]{\includegraphics[width=\doubleFig]{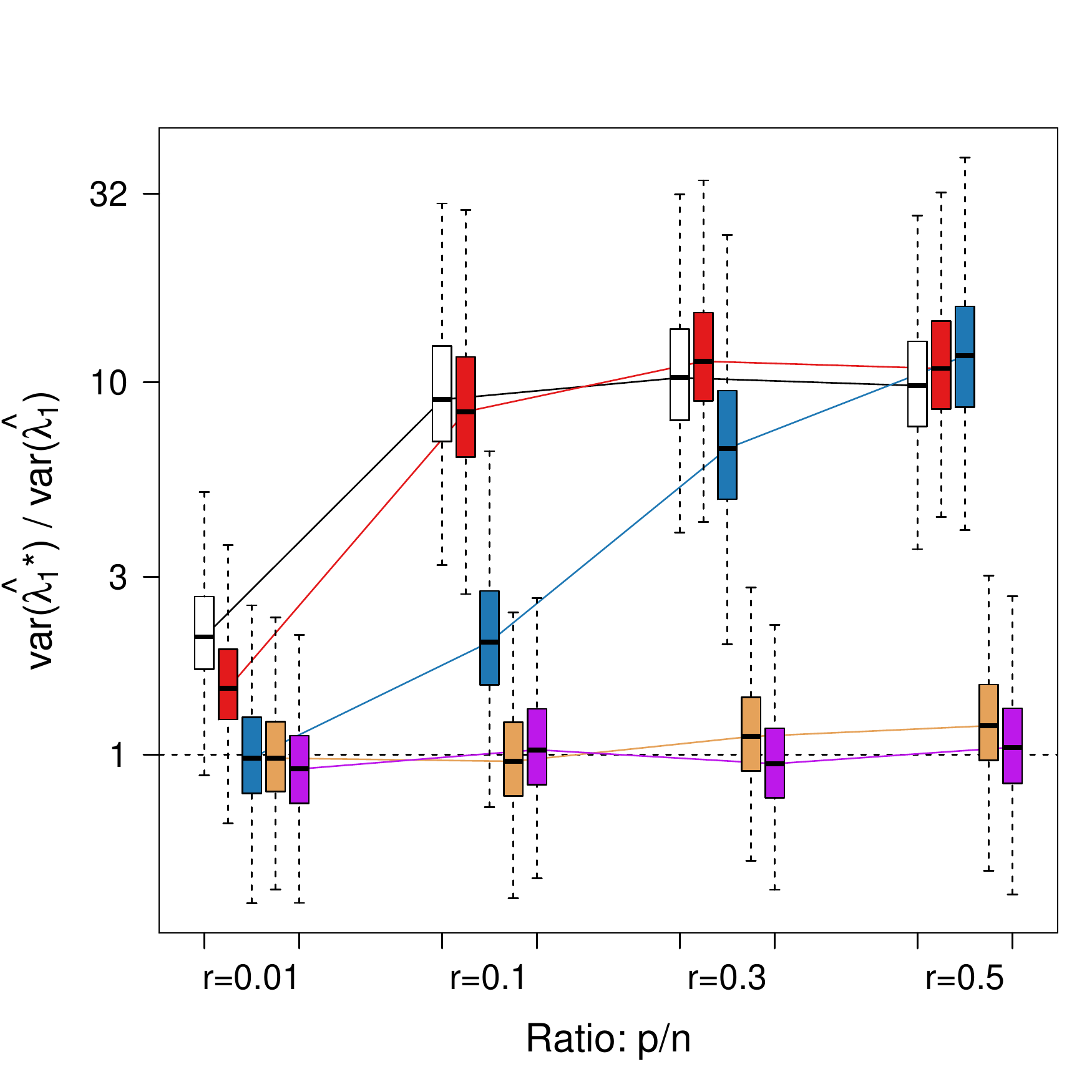} \label{subfig:bootBigVar:Norm} }
	\subfloat[][$X_i\sim$ Ellip. Uniform]{\includegraphics[width=\doubleFig]{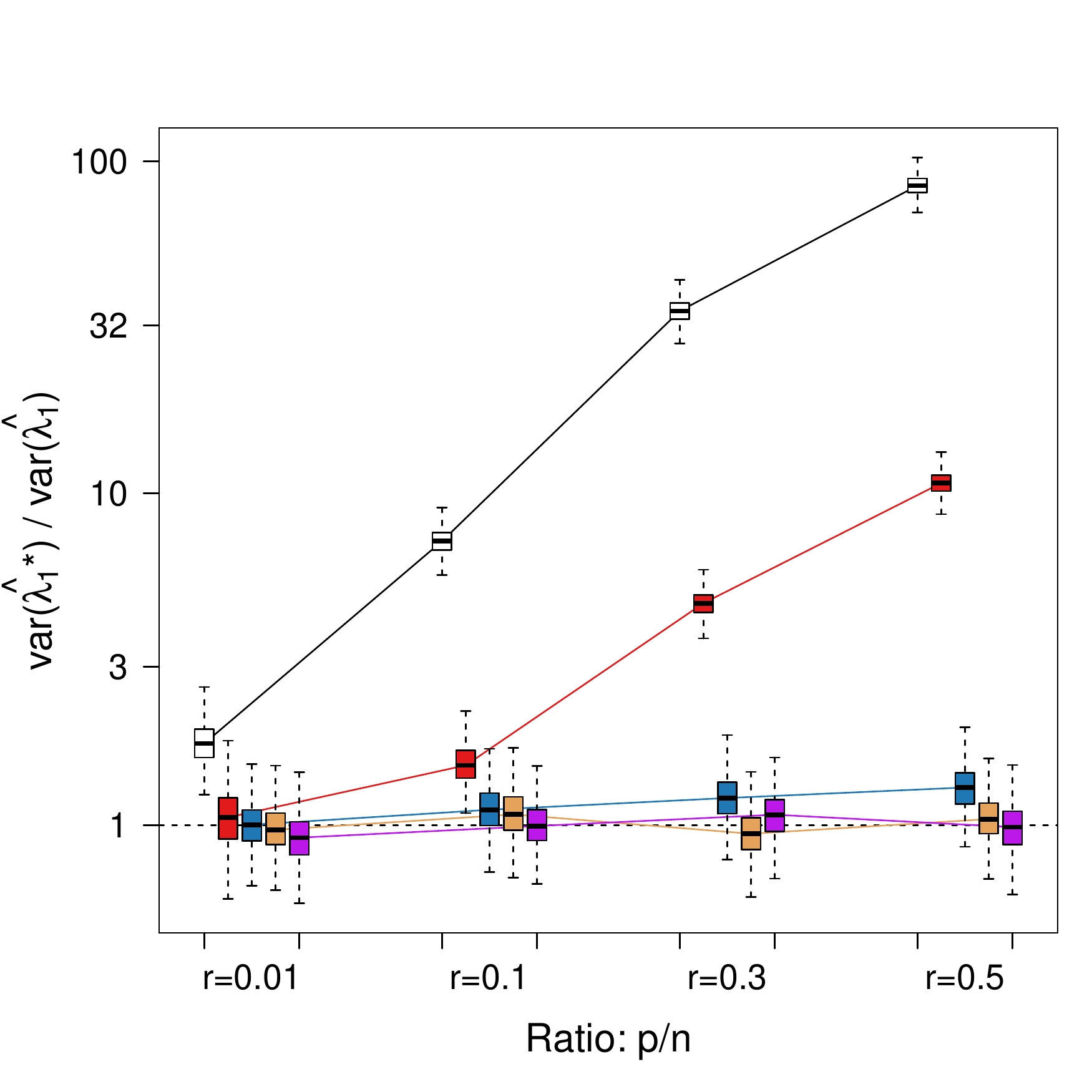} \label{subfig:bootBigVar:EllipExp} }\\
\caption{
 \textbf{Top Eigenvalue: Ratio of Bootstrap Estimate of Variance to True Variance for Largest Eigenvalue, n=1,000:  } Plotted are boxplots of the bootstrap estimate of variance, showing larger values of $\lambda_1$ than shown in the main text. See the legend of Figure \ref{fig:bootVar} in the main text for more details.
}\label{fig:bootBigVar}
\end{figure}

\begin{figure}[t]
	\centering
	\subfloat[][\centering$X_i\sim$ Ellip Normal\par Percentile Intervals]{\includegraphics[type=pdf,ext=.pdf,read=.pdf,width=\doubleFig]{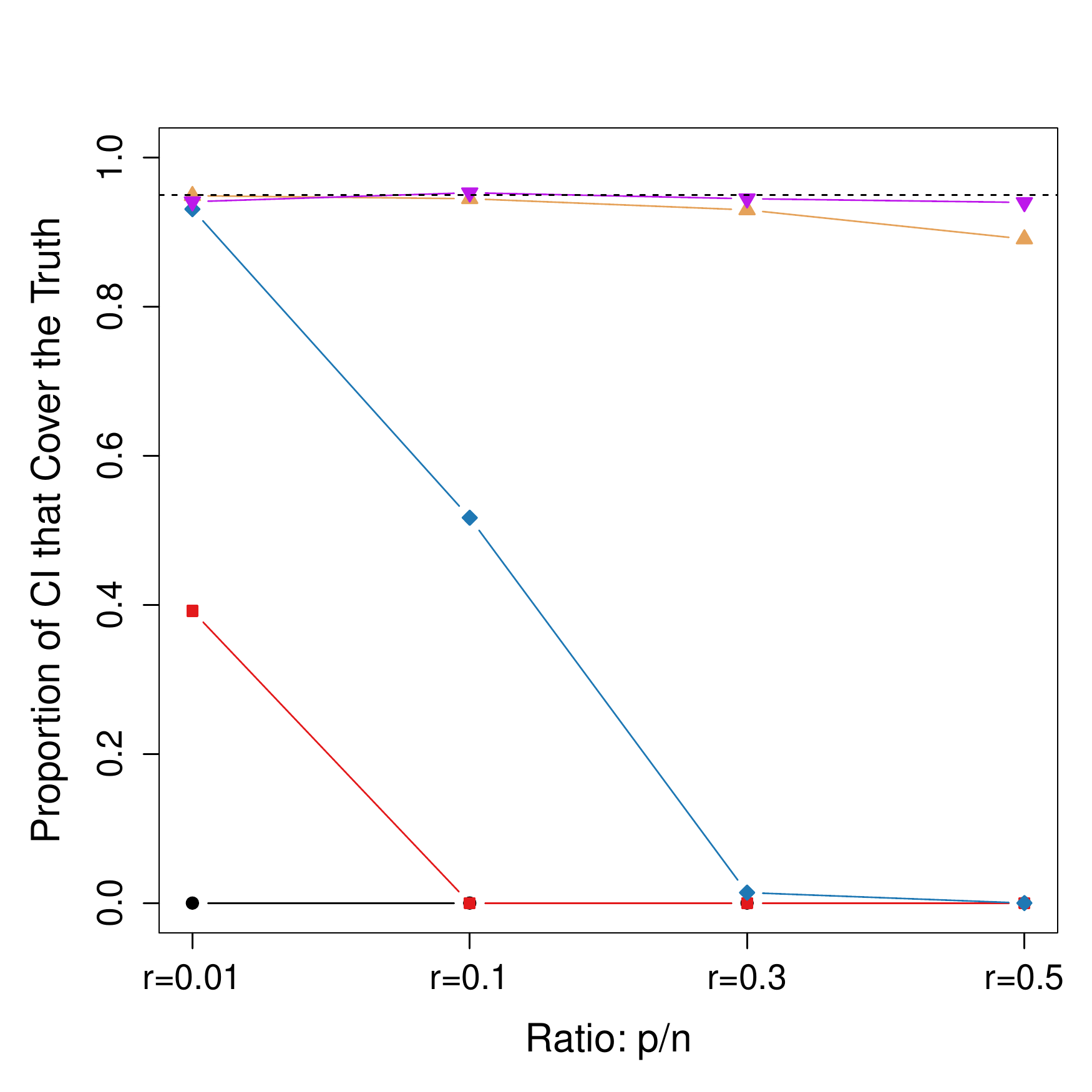} \label{subfig:bootCI:NormPerc} }
	\subfloat[][\centering$X_i\sim$ Ellip Normal\par Normal-based Intervals]{\includegraphics[type=pdf,ext=.pdf,read=.pdf,width=\doubleFig]{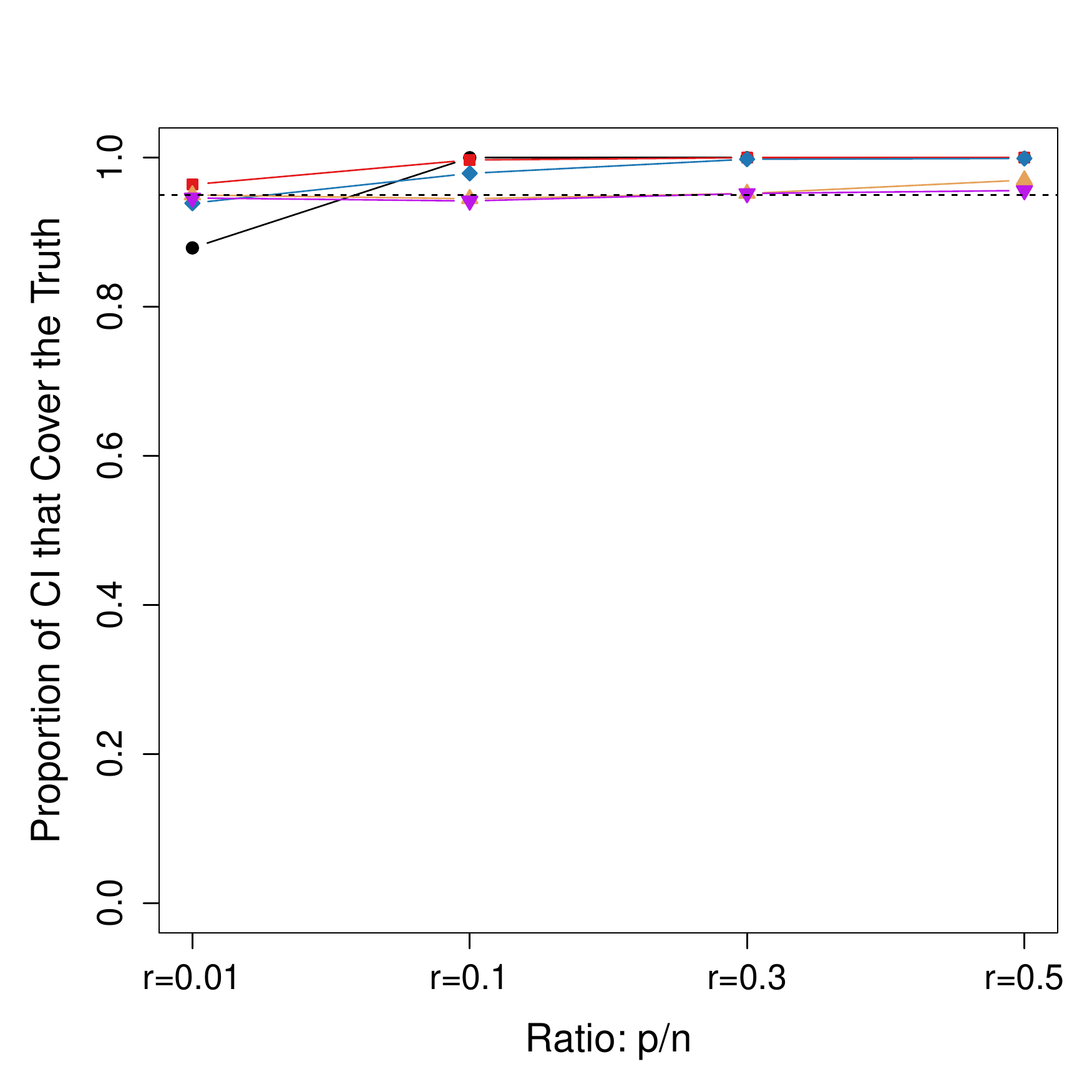} \label{subfig:bootCI:NormNormal} }
	\\
	\subfloat[][\centering$X_i\sim$ Ellip Uniform\par Percentile Intervals]{\includegraphics[type=pdf,ext=.pdf,read=.pdf,width=\doubleFig]{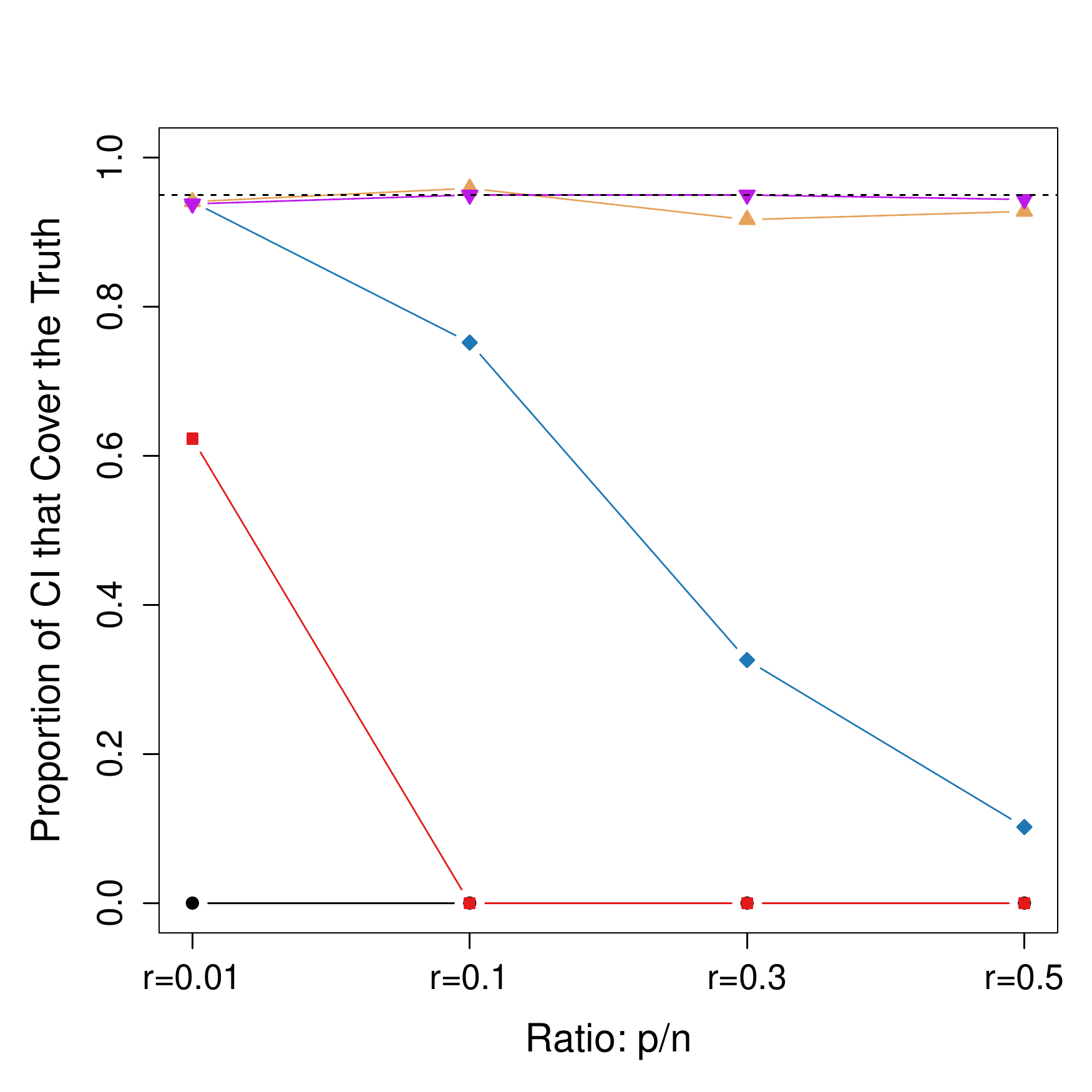} \label{subfig:bootCI:EllipExpPerc} }
	\subfloat[][\centering $X_i\sim$ Ellip Uniform \par Normal-based Intervals]{\includegraphics[type=pdf,ext=.pdf,read=.pdf,width=\doubleFig]{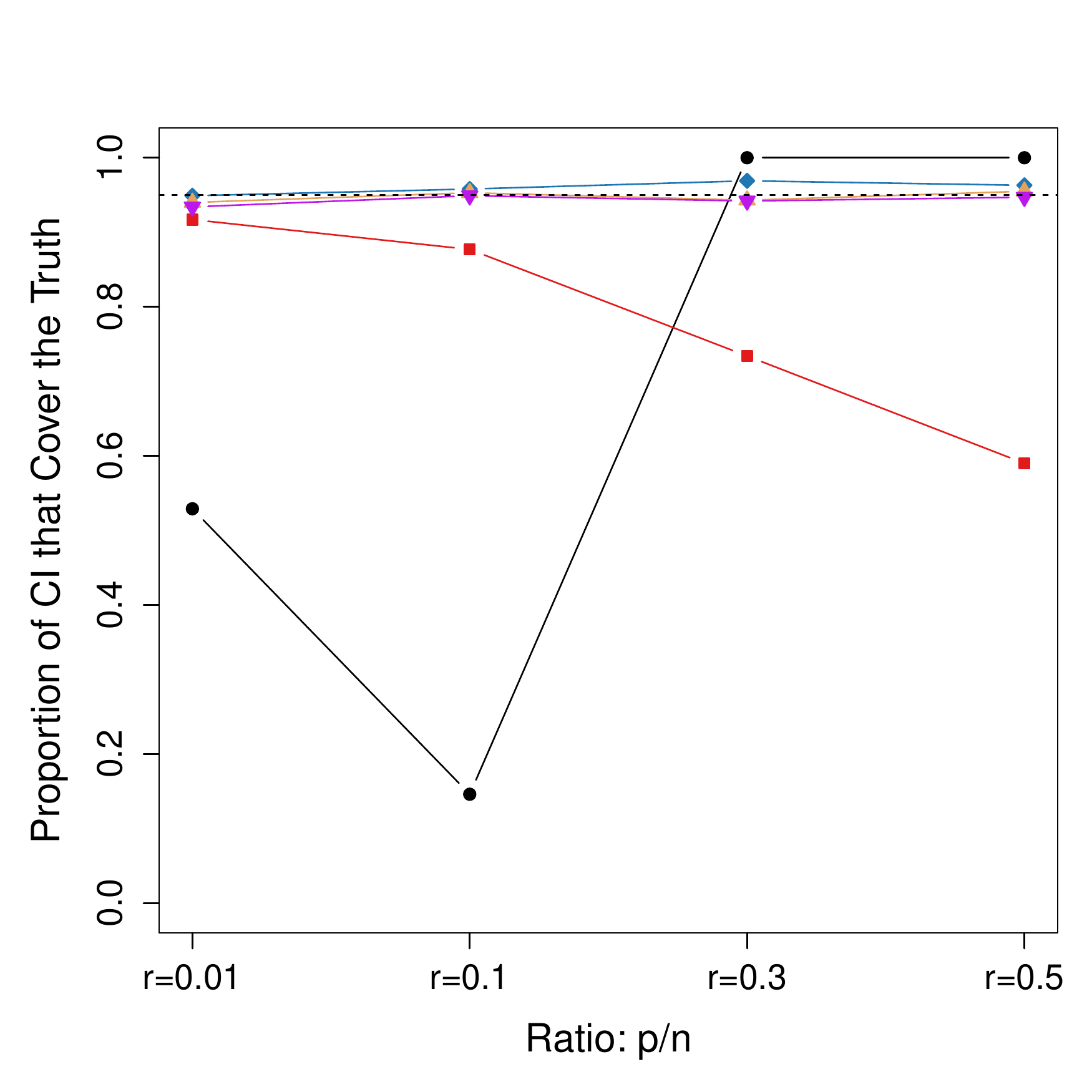} \label{subfig:bootCI:EllipExpNorm} }
\caption{
 \textbf{Top Eigenvalue: 95\% CI Coverage, $n=1,000$ for additional distributions:  } Plotted are the corresponding CI Coverage plots for when $X_i$ follows an elliptical distribution with Normal and Uniform weights. See Figure \ref{fig:bootCI} for more details. }\label{fig:bootCIMoreEllip}
\end{figure}

\begin{figure}[t]
	\centering
	\subfloat[][\centering $X_i\sim$ Normal \par r=0.01]{\includegraphics[type=pdf,ext=.pdf,read=.pdf,width=\doubleFig]{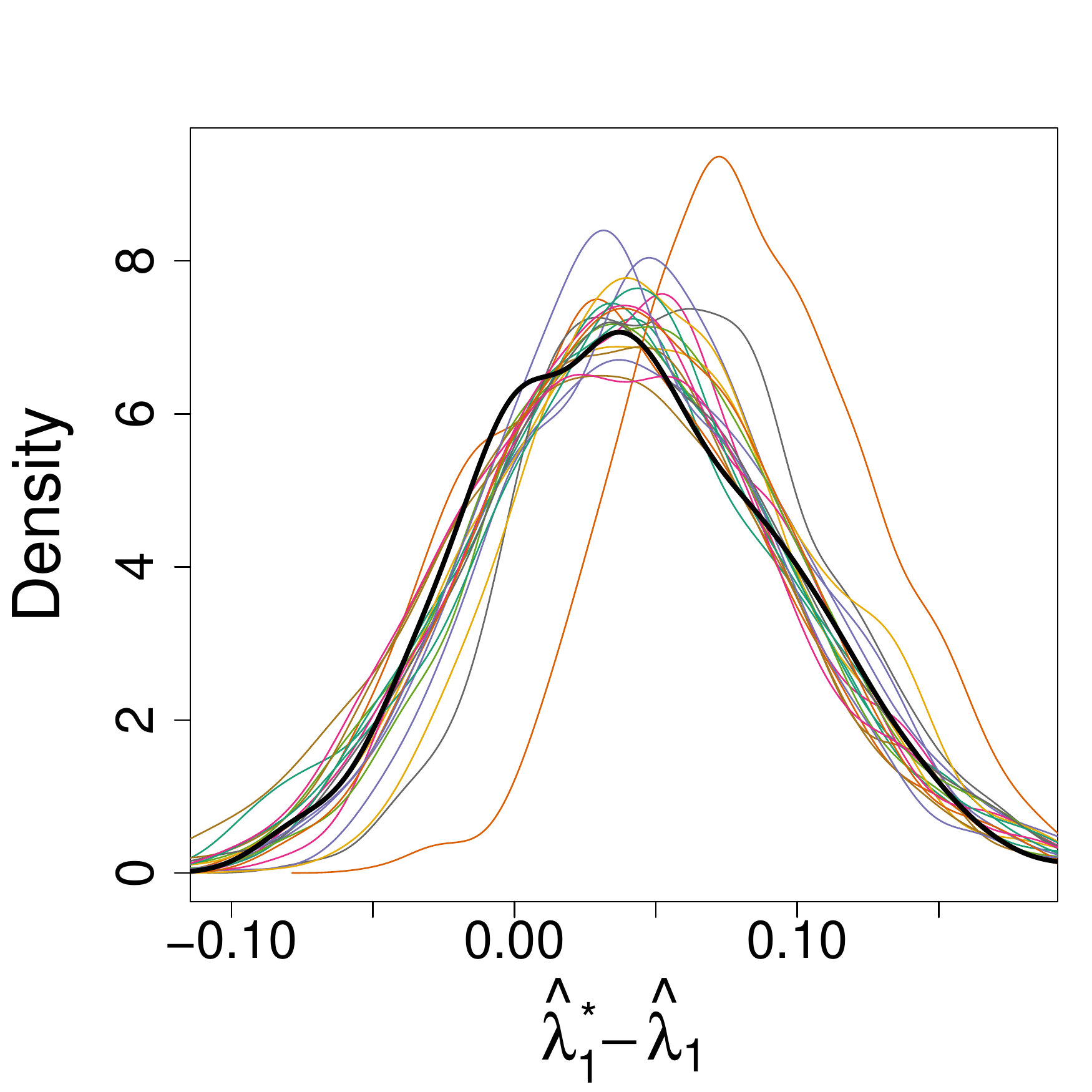} \label{subfig:bootDensityAlt:Norm0.01} }
	\subfloat[][\centering$X_i\sim$ Normal \par r=0.3]{\includegraphics[type=pdf,ext=.pdf,read=.pdf,width=\doubleFig]{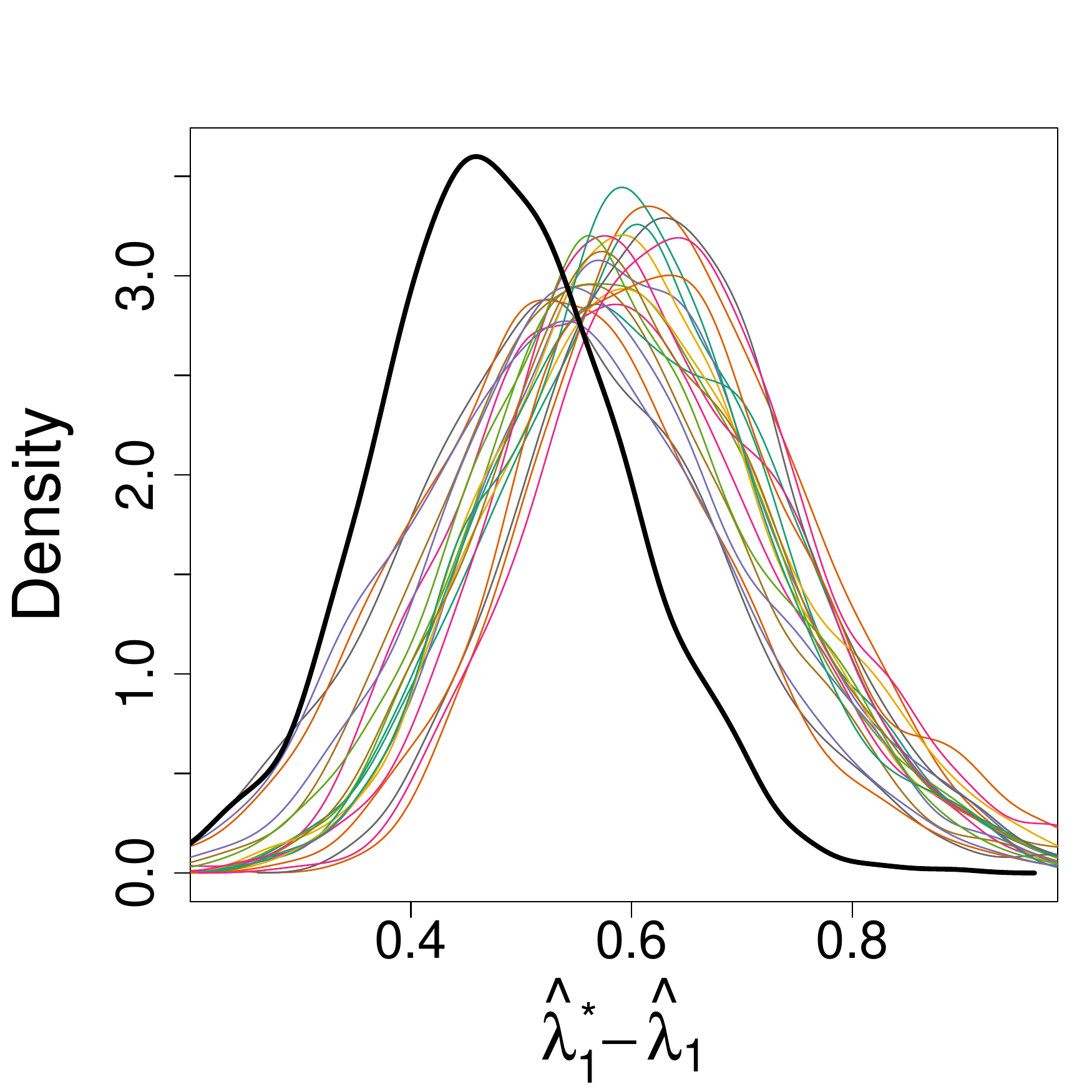} \label{subfig:bootDensityAlt:Norm0.3} }
	\\
	\subfloat[][\centering$X_i\sim$ Ellip. Exp \par r=0.01]{\includegraphics[type=pdf,ext=.pdf,read=.pdf,width=\doubleFig]{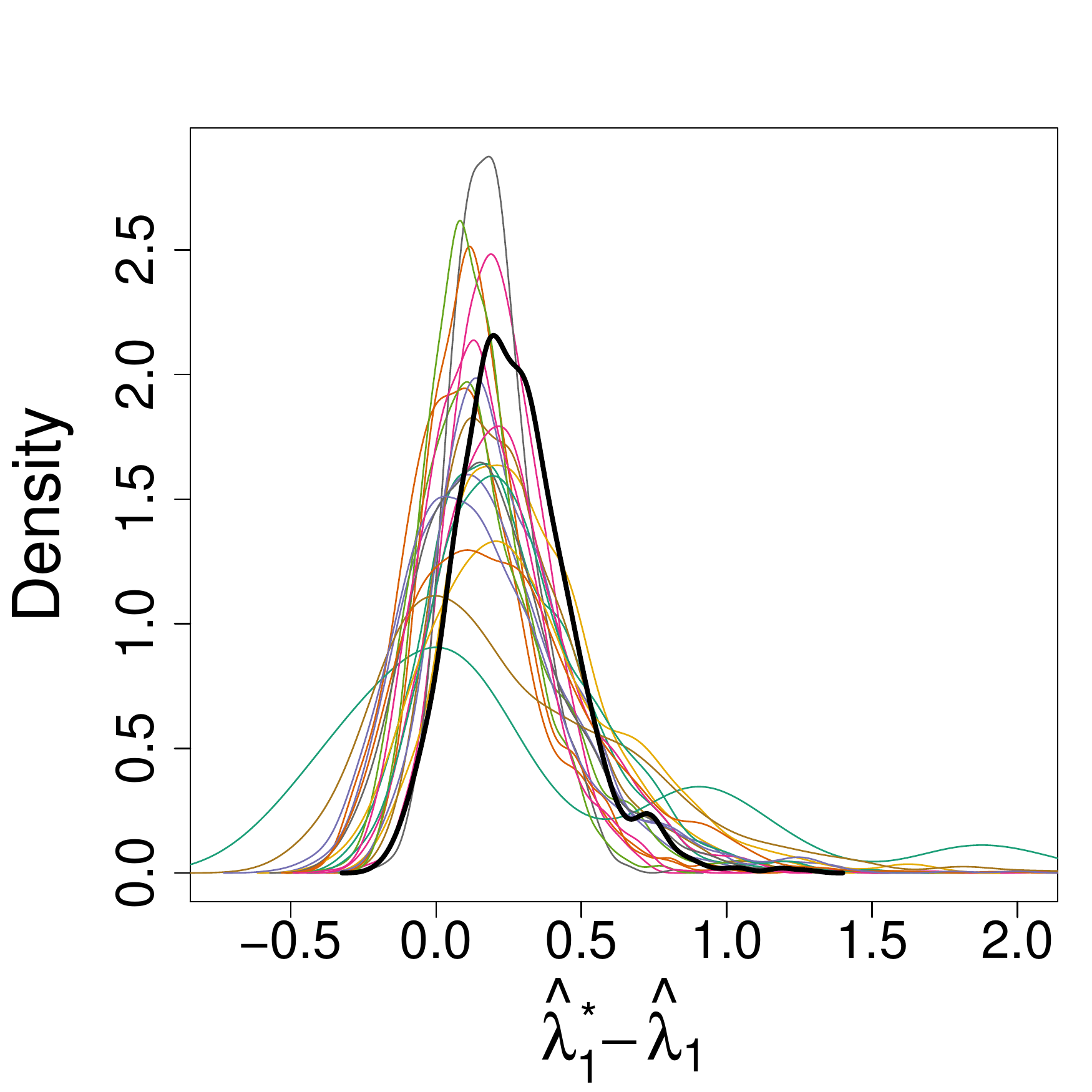} \label{subfig:bootDensityAlt:EllipExp0.01} }
	\subfloat[][\centering$X_i\sim$ Ellip. Exp \par r=0.3]{\includegraphics[type=pdf,ext=.pdf,read=.pdf,width=\doubleFig]{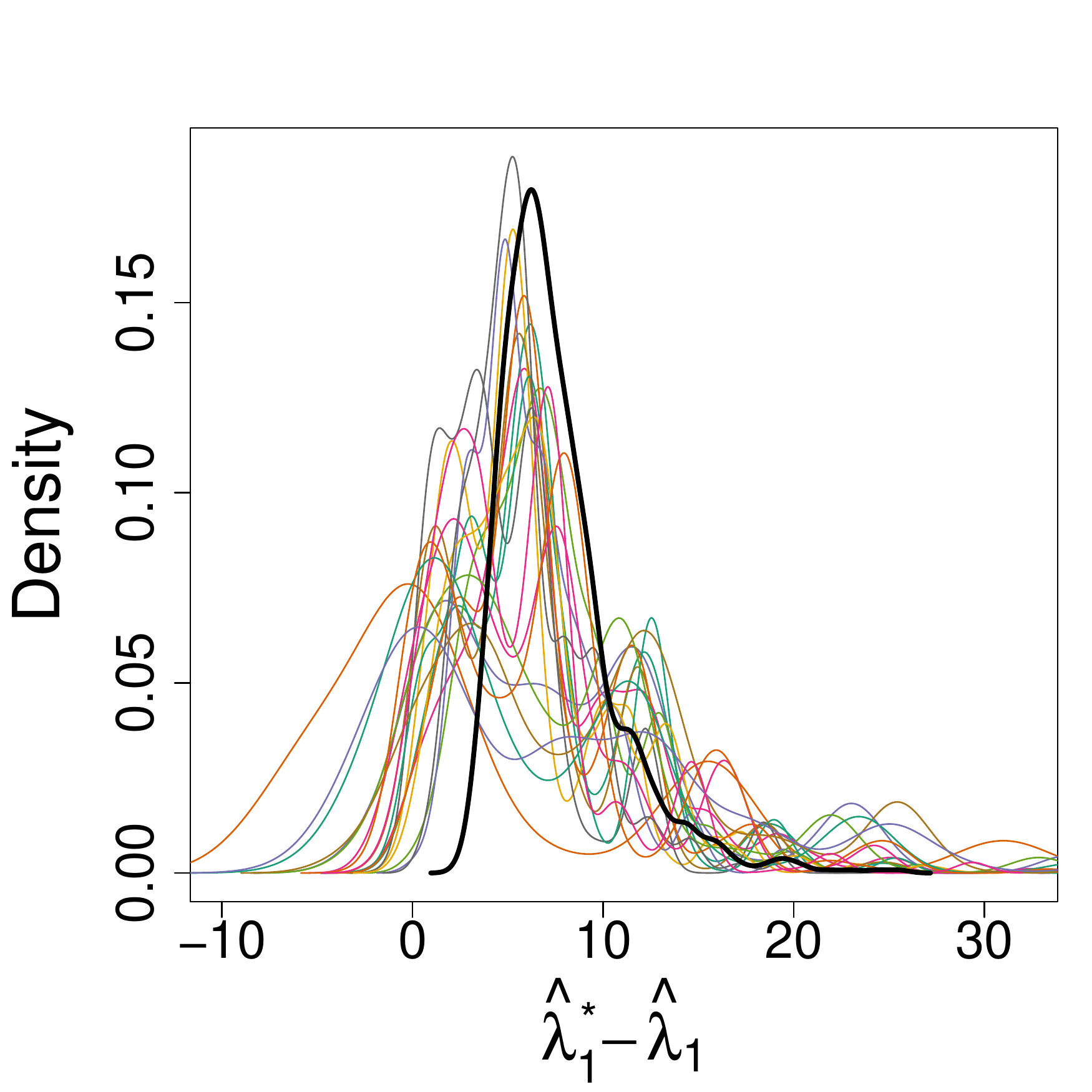} \label{subfig:bootDensityAlt:EllipExp0.3} }
\caption{
 \textbf{Top Eigenvalue: Bootstrap distribution of $\hat{\lambda}^*_1$ when $\lambda_1=1+3\sqrt{r}$, n=1,000:  } Plotted are the estimated density of twenty simulations of the bootstrap distribution of $\hat{\lambda}^{*b}_1-\hat{\lambda}_1$, with $b=1,\ldots,999$. The solid black line line represents the distribution of $\hat{\lambda}_1-\lambda_1$ over 1,000 simulations. 
}\label{fig:bootDensityAlt}
\end{figure}	


\begin{figure}[t]
	\centering
	\subfloat[][$X_i\sim$ Normal]{\includegraphics[width=\doubleFig]{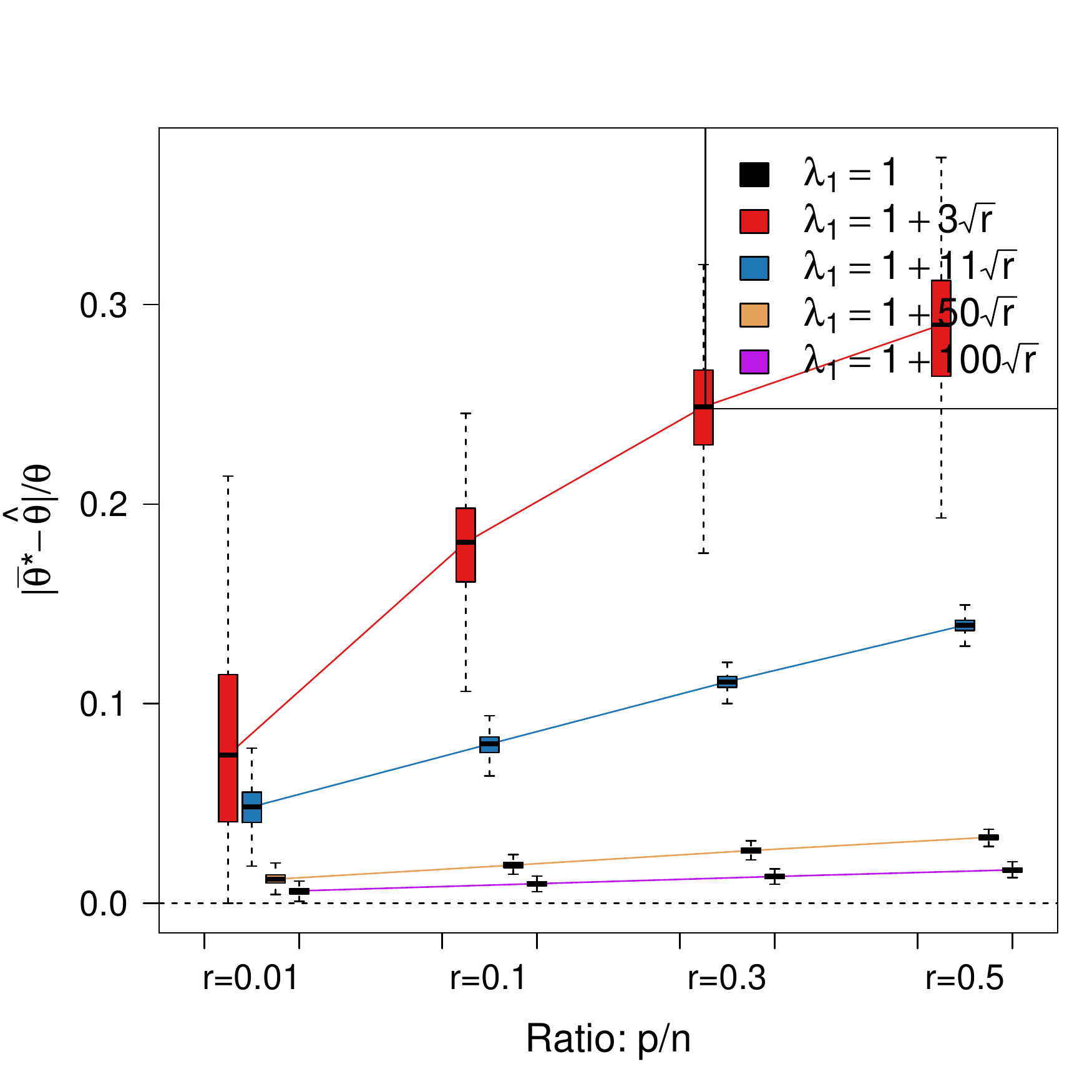} \label{subfig:bootBiasGap:Norm} }
	\subfloat[][$X_i\sim$ Ellip. Exp]{\includegraphics[width=\doubleFig]{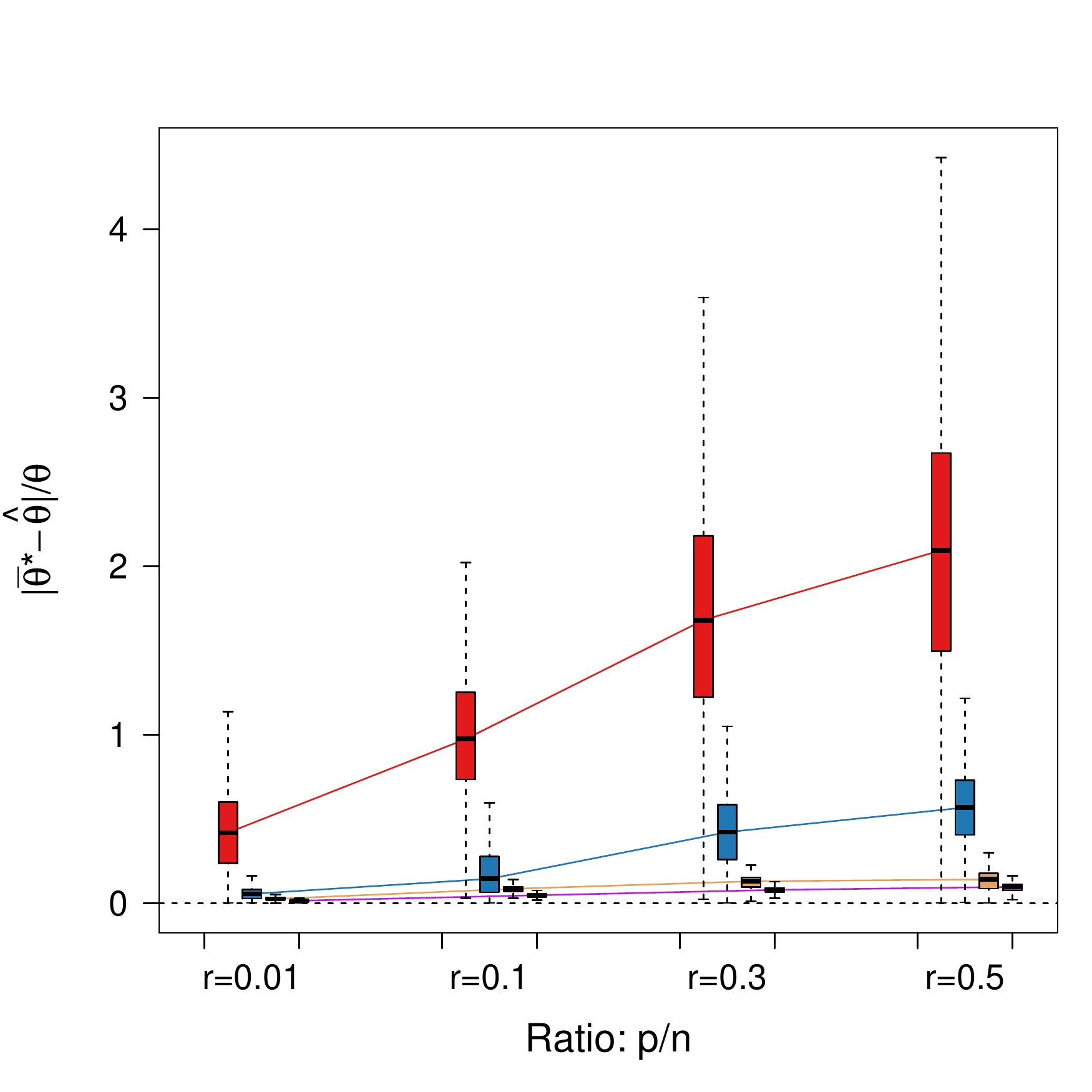} \label{subfig:bootBiasGap:EllipExp} }\\
	\subfloat[][$X_i\sim$ Ellip. Norm]{\includegraphics[width=\doubleFig]{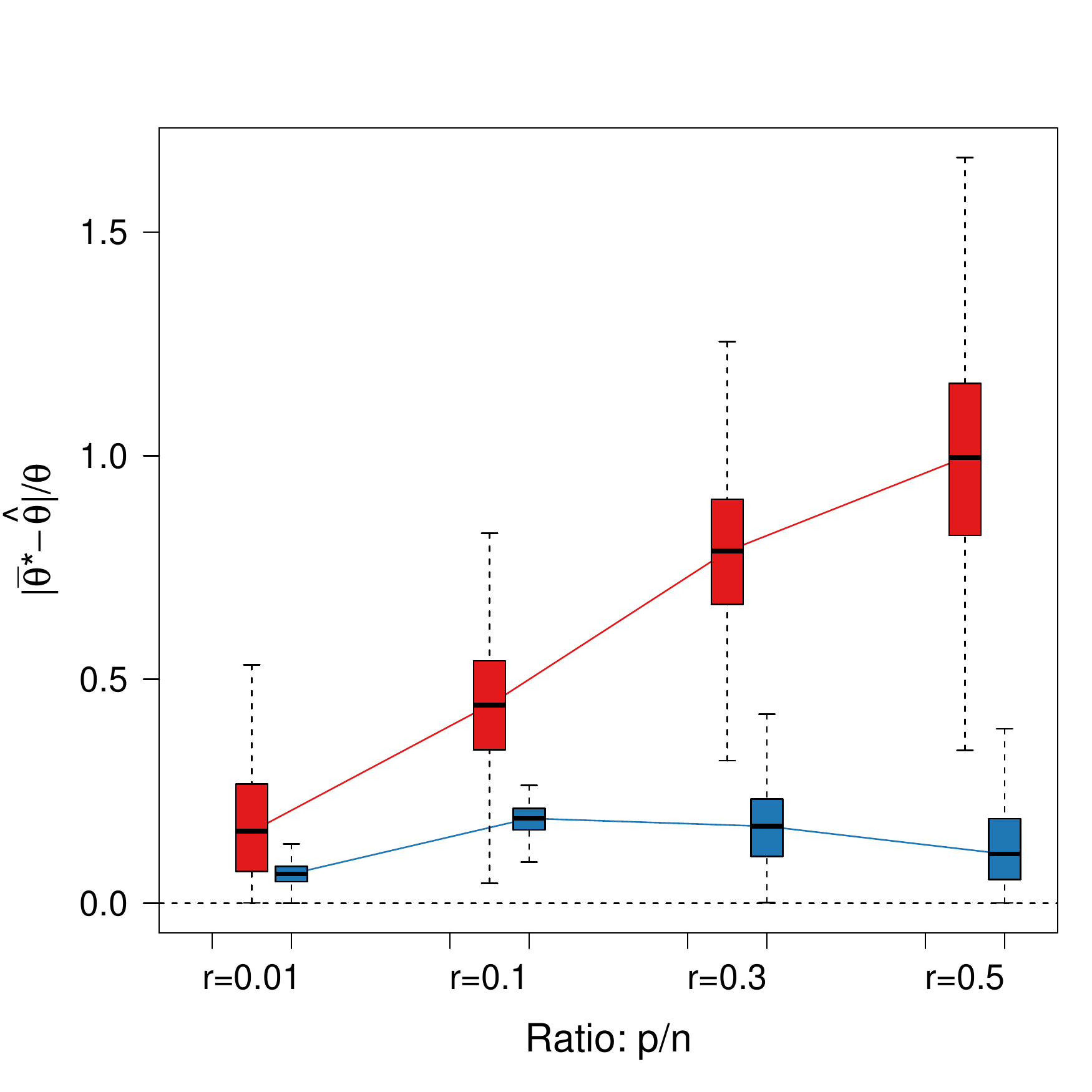} \label{subfig:bootBiasGap:EllipNorm} }
	\subfloat[][$X_i\sim$ Ellip. Unif]{\includegraphics[width=\doubleFig]{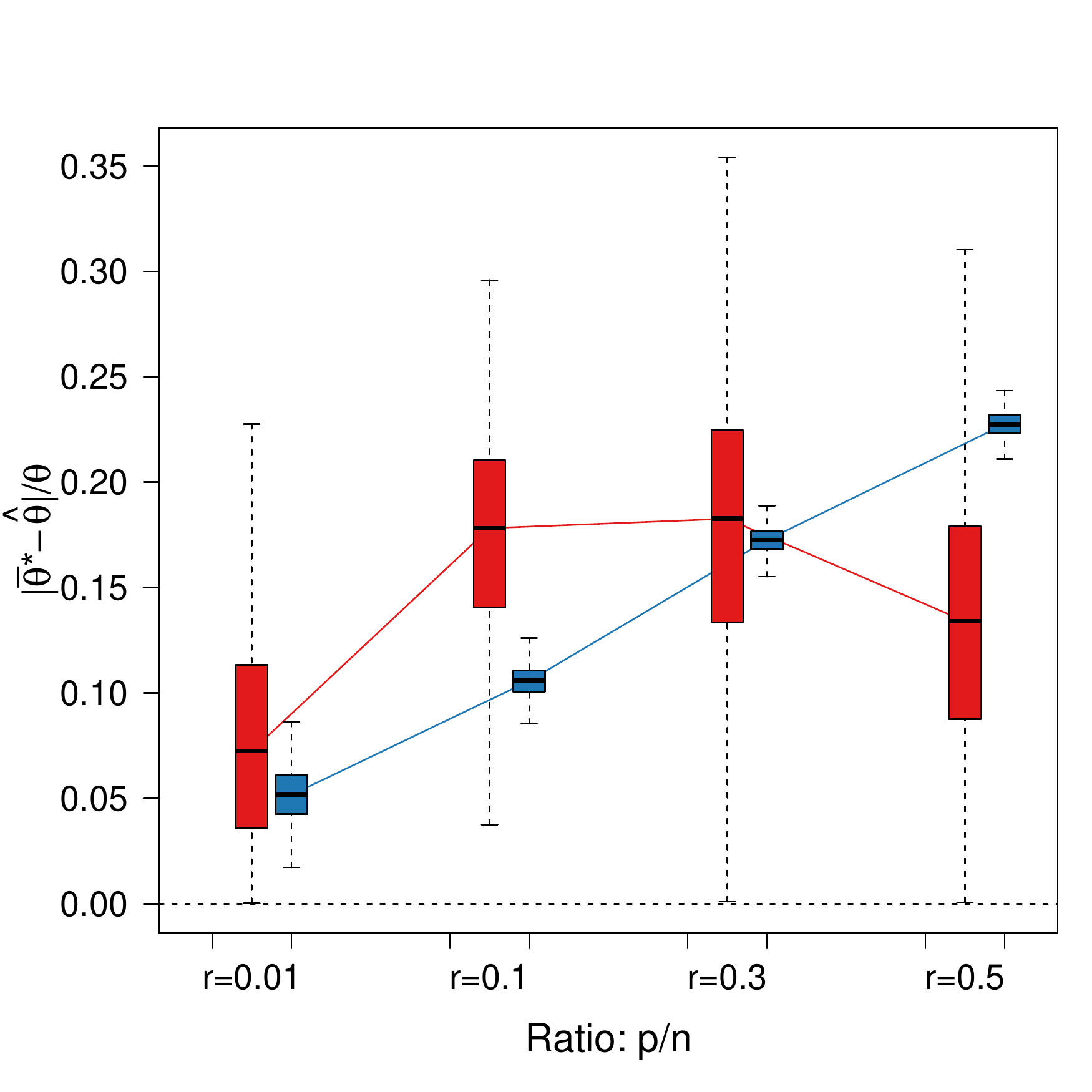} \label{subfig:bootBiasGap:EllipUnif} }
	
\caption{
 \textbf{Gap Statistic: Ratio of Bias of Bootstrap to true Gap Statistic. } Note that the true Gap Statistic when $\lambda_1=1$ is zero, so that the ratio is not well defined and hence not plotted . See also Supplementary Figures \ref{tab:bootBiasNormGap}-\ref{tab:bootBiasEllipUnifGap} for the median bias values and the legend of Figure \ref{fig:bootBias} in the main text for more information about this plot. Note that the y-axis is different for each of the distributions, and differs from that of Supplementary Figure \ref{fig:bootBigBias}. 
}\label{fig:bootBiasGap}
\end{figure}	
\begin{figure}[t]
	\centering
	\subfloat[][$X_i\sim$ Normal]{\includegraphics[width=\doubleFig]{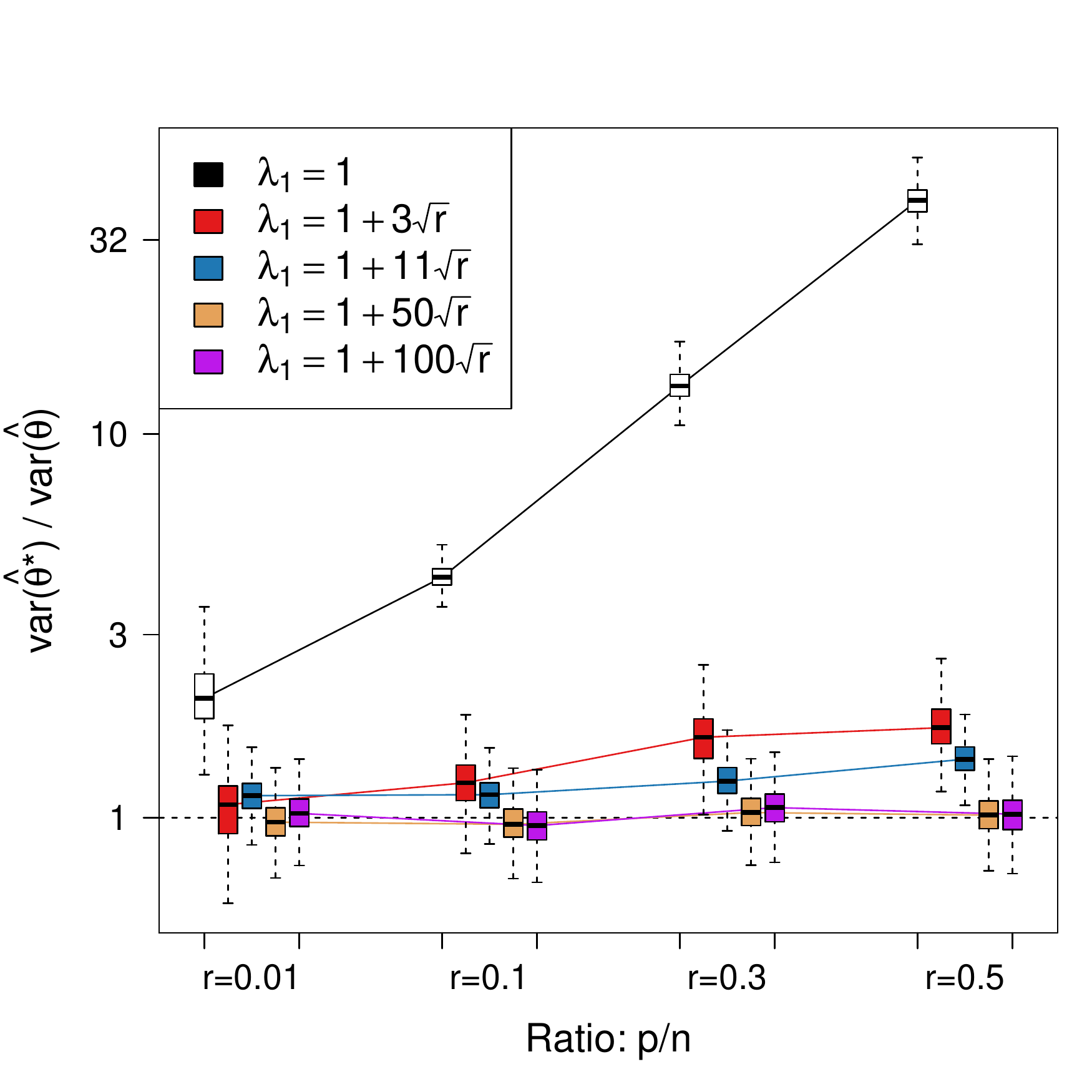} \label{subfig:bootVarGap:Norm} }
	\subfloat[][$X_i\sim$ Ellip. Exp]{\includegraphics[width=\doubleFig]{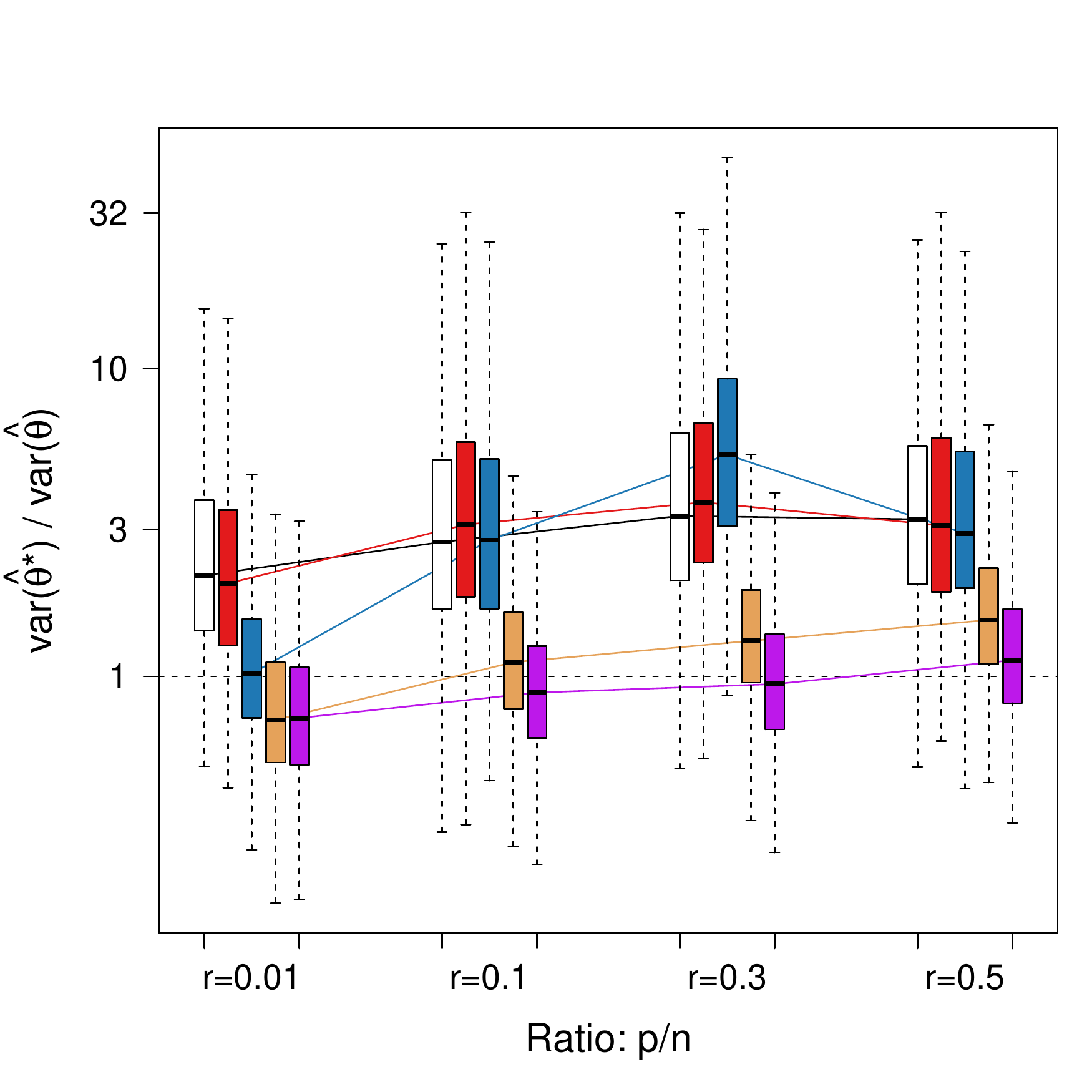} \label{subfig:bootVarGap:EllipExp} }\\
	\subfloat[][$X_i\sim$ Ellip. Normal]{\includegraphics[width=\doubleFig]{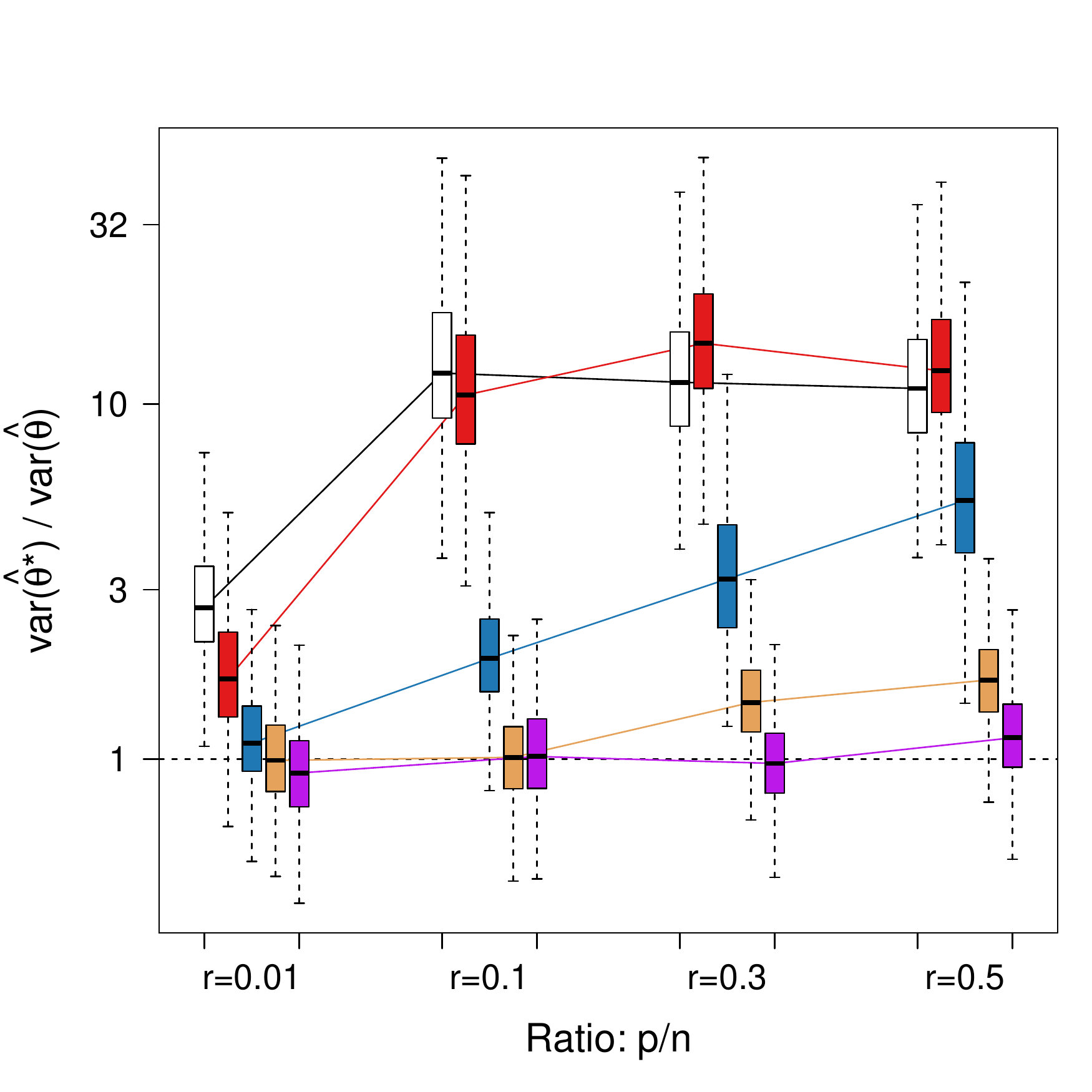} \label{subfig:bootVarGap:EllipNorm} }
	\subfloat[][$X_i\sim$ Ellip. Uniform]{\includegraphics[width=\doubleFig]{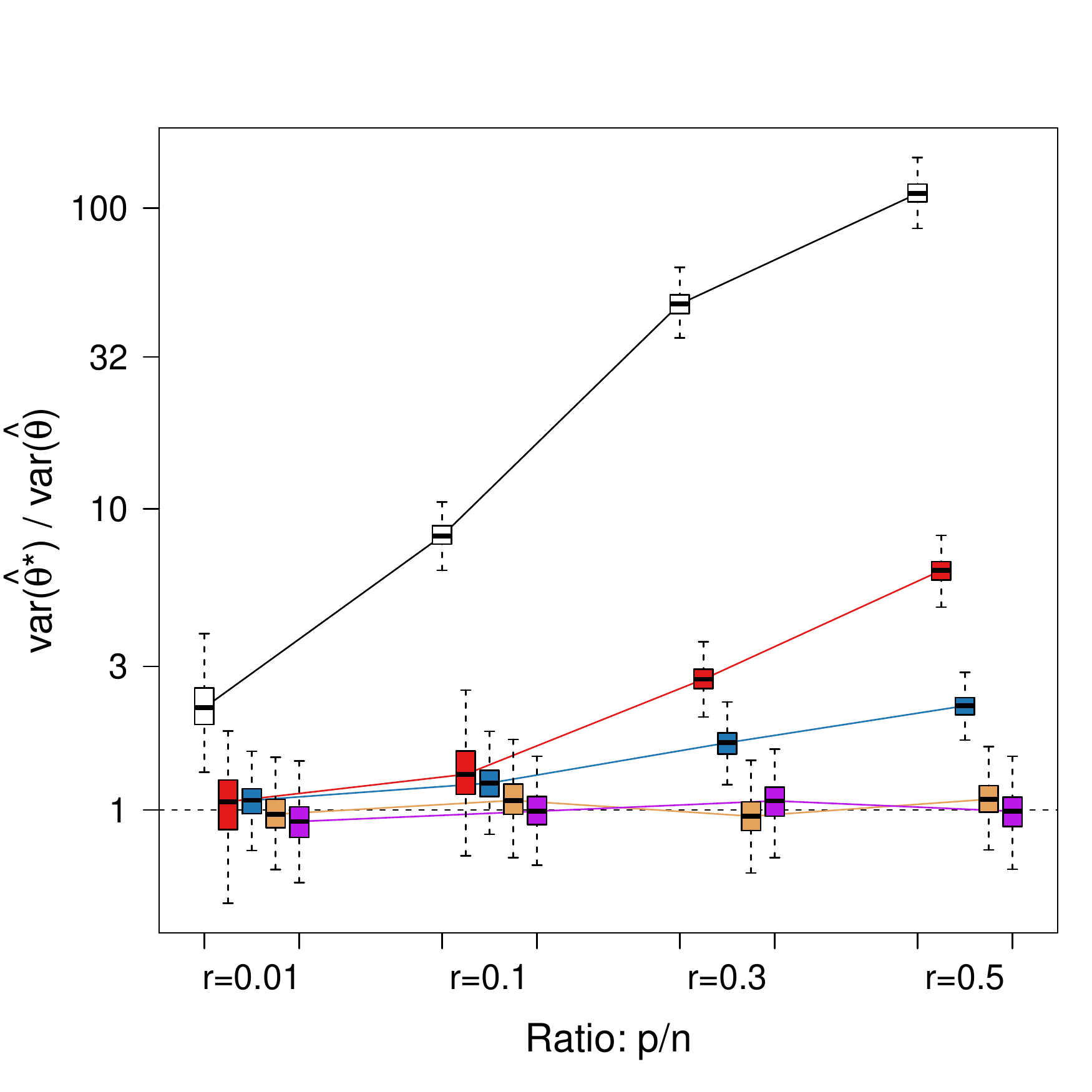} \label{subfig:bootVarGap:EllipUnif} }
	
\caption{
 \textbf{Gap Statistic: Ratio of Bootstrap Estimate of Variance to True Variance.}
}\label{fig:bootVarGap}
\end{figure}	
\newpage

\begin{figure}[t]
	\centering
	\subfloat[][\centering$X_i\sim$ Normal\par Percentile Intervals]{\includegraphics[type=pdf,ext=.pdf,read=.pdf,width=\doubleFig]{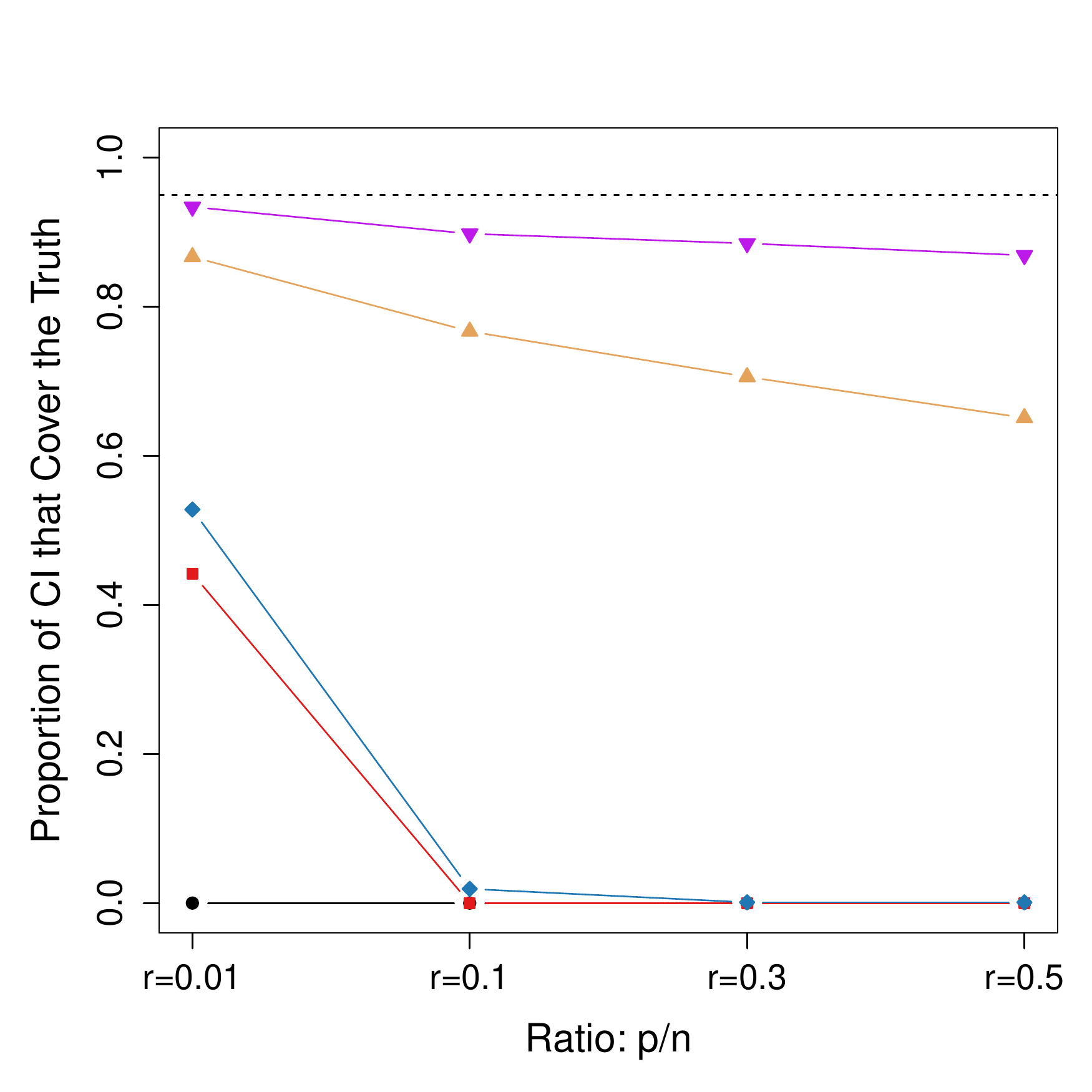} \label{subfig:bootCI:NormPerc} }
	\subfloat[][\centering$X_i\sim$ Normal\par Normal-based Intervals]{\includegraphics[type=pdf,ext=.pdf,read=.pdf,width=\doubleFig]{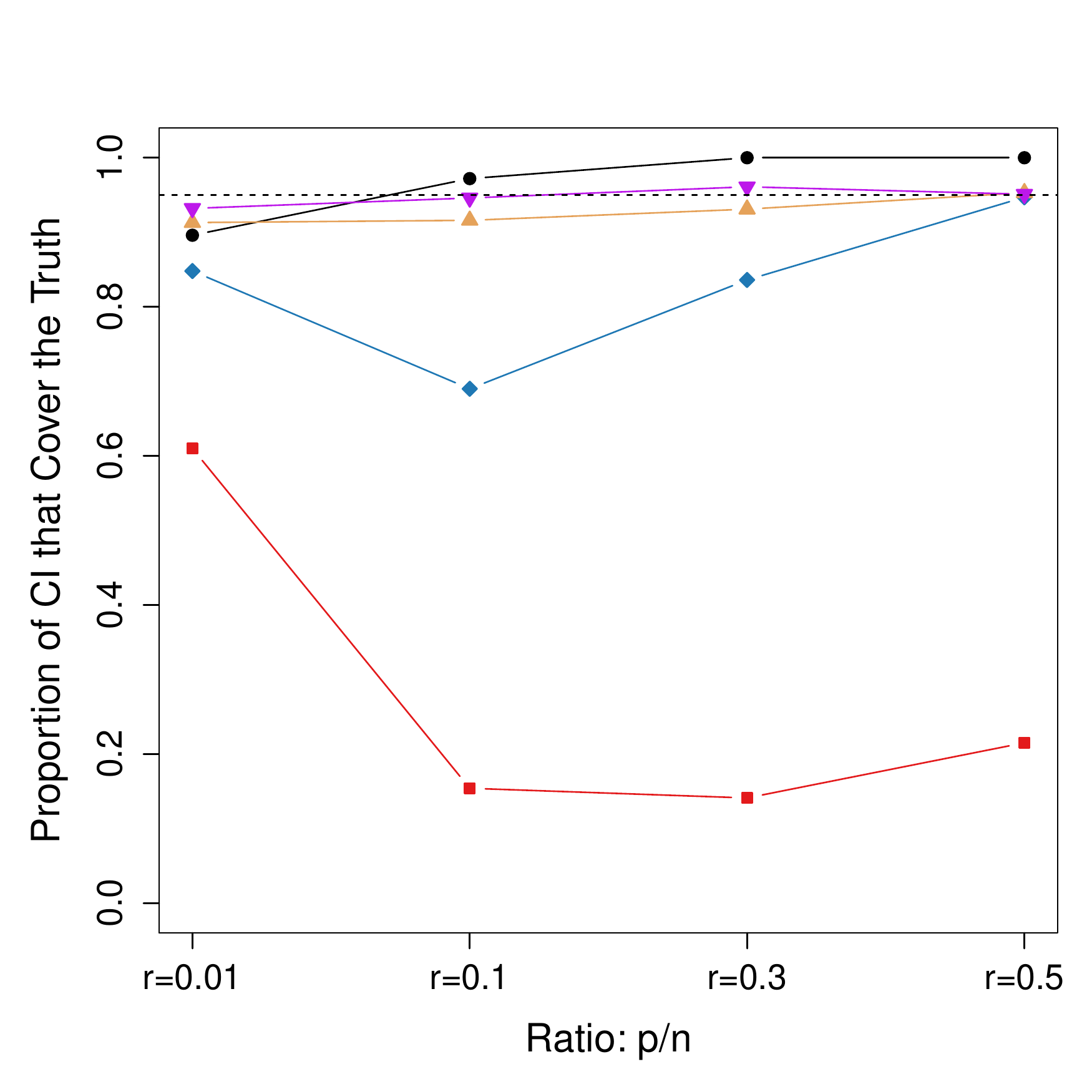} \label{subfig:bootCI:NormNormal} }
	\\
	\subfloat[][\centering$X_i\sim$ Ellip Exp\par Percentile Intervals]{\includegraphics[type=pdf,ext=.pdf,read=.pdf,width=\doubleFig]{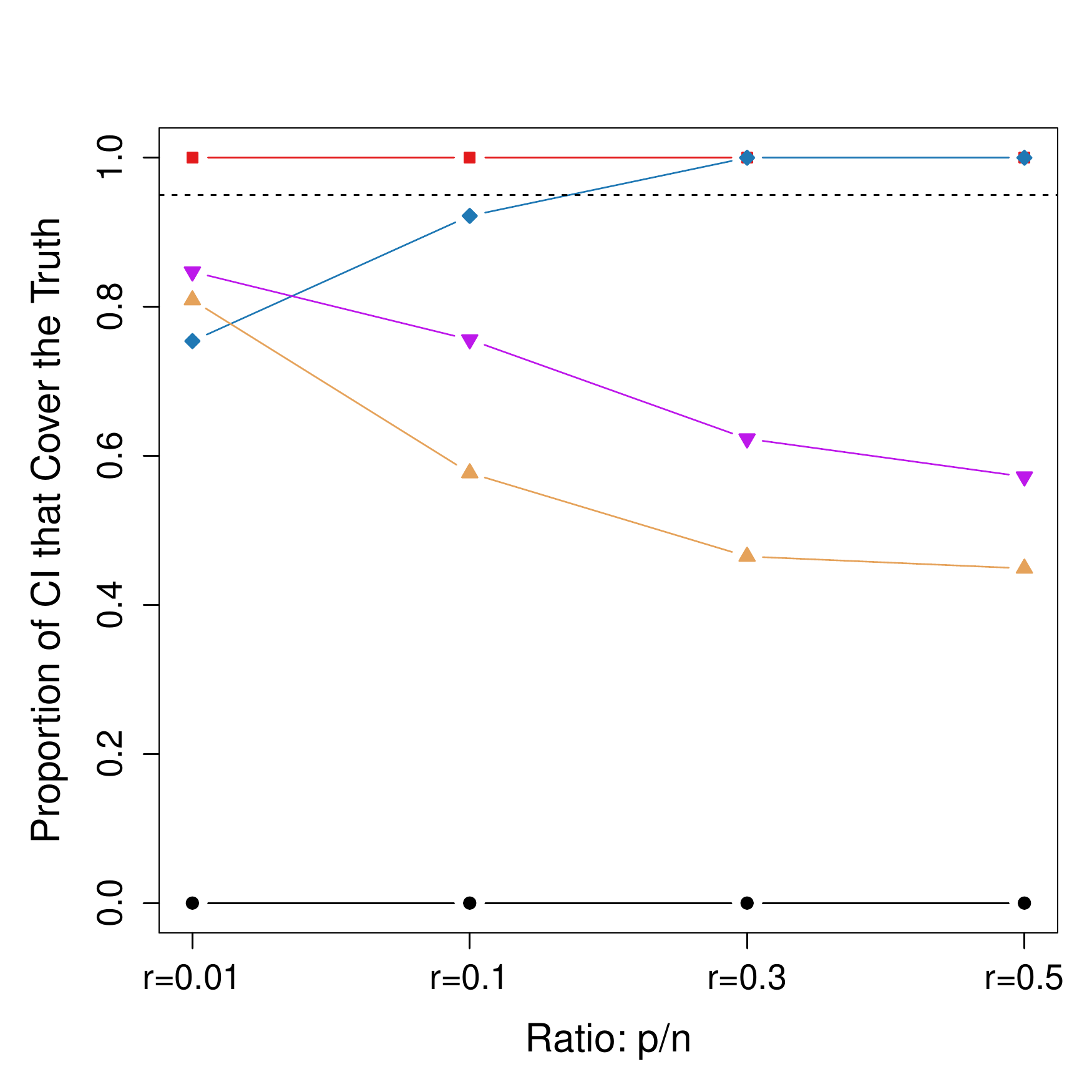} \label{subfig:bootCI:EllipExpPerc} }
	\subfloat[][\centering $X_i\sim$ Ellip Exp \par Normal-based Intervals]{\includegraphics[type=pdf,ext=.pdf,read=.pdf,width=\doubleFig]{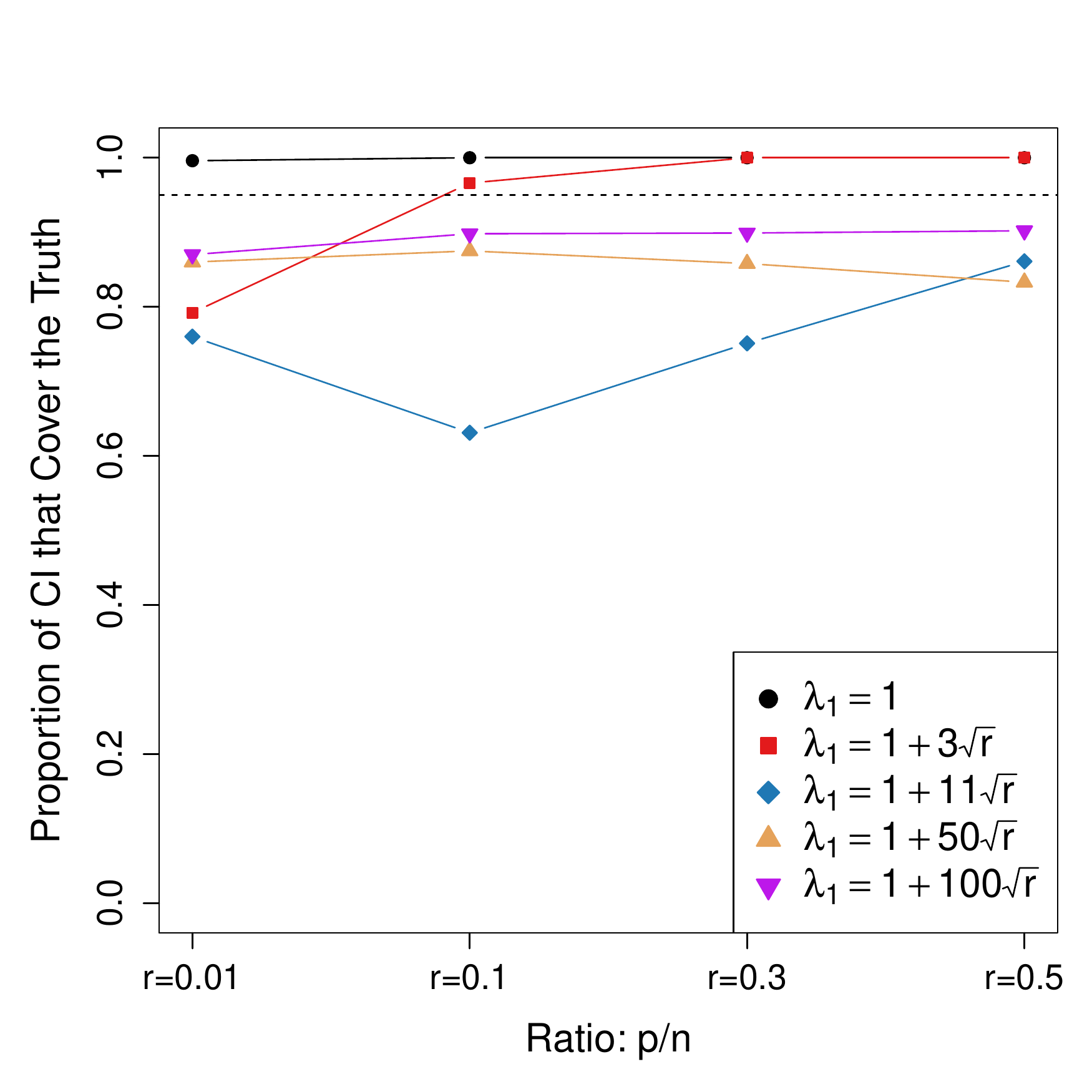} \label{subfig:bootCI:EllipExpNorm} }
\end{figure}
\begin{figure}
	\ContinuedFloat
	\subfloat[][\centering$X_i\sim$ Ellip Normal\par Percentile Intervals]{\includegraphics[type=pdf,ext=.pdf,read=.pdf,width=\doubleFig]{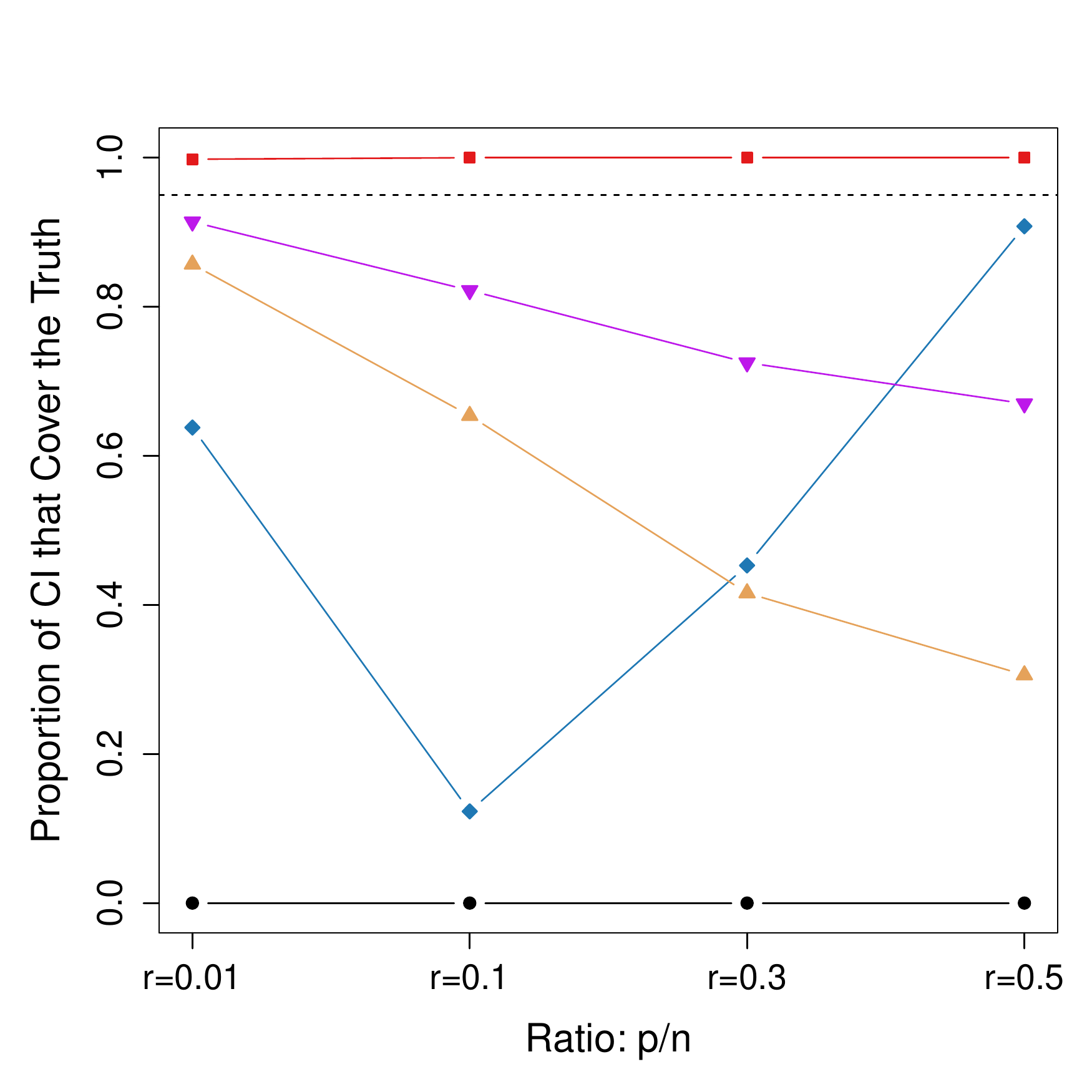} \label{subfig:bootCI:NormPerc} }
	\subfloat[][\centering$X_i\sim$ Ellip Normal\par Normal-based Intervals]{\includegraphics[type=pdf,ext=.pdf,read=.pdf,width=\doubleFig]{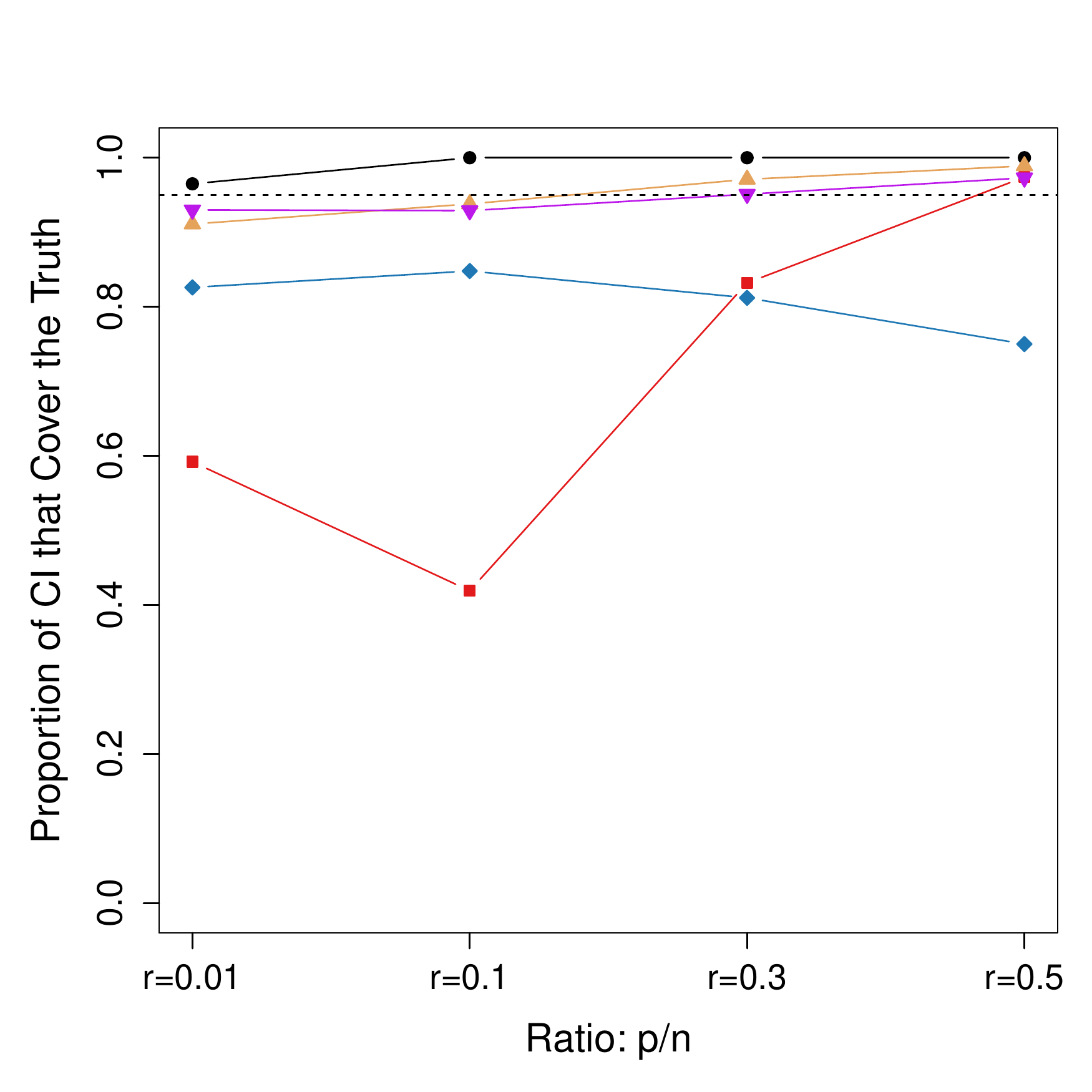} \label{subfig:bootCI:NormNormal} }
	\\
	\subfloat[][\centering$X_i\sim$ Ellip Uniform\par Percentile Intervals]{\includegraphics[type=pdf,ext=.pdf,read=.pdf,width=\doubleFig]{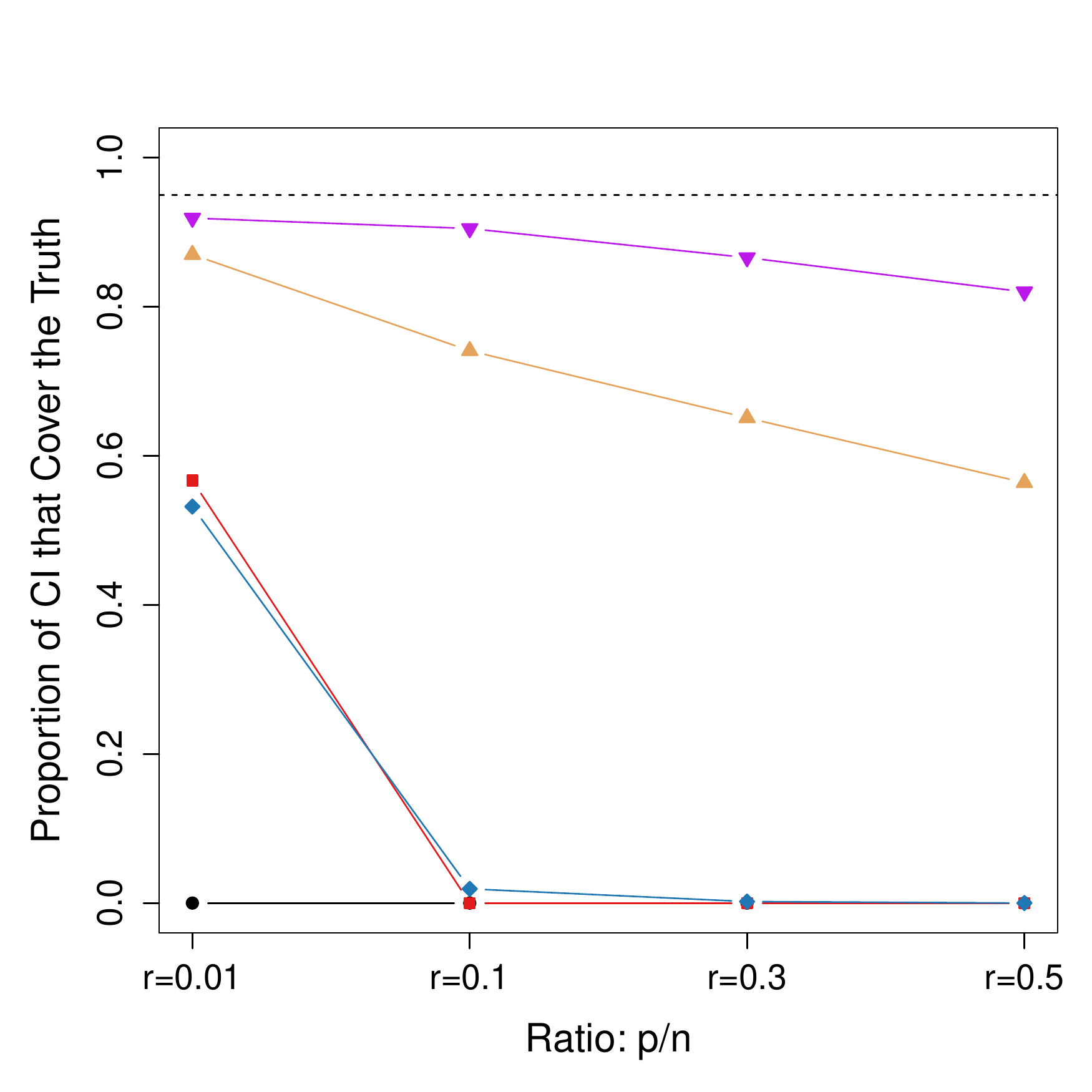} \label{subfig:bootCI:EllipExpPerc} }
	\subfloat[][\centering $X_i\sim$ Ellip Uniform \par Normal-based Intervals]{\includegraphics[type=pdf,ext=.pdf,read=.pdf,width=\doubleFig]{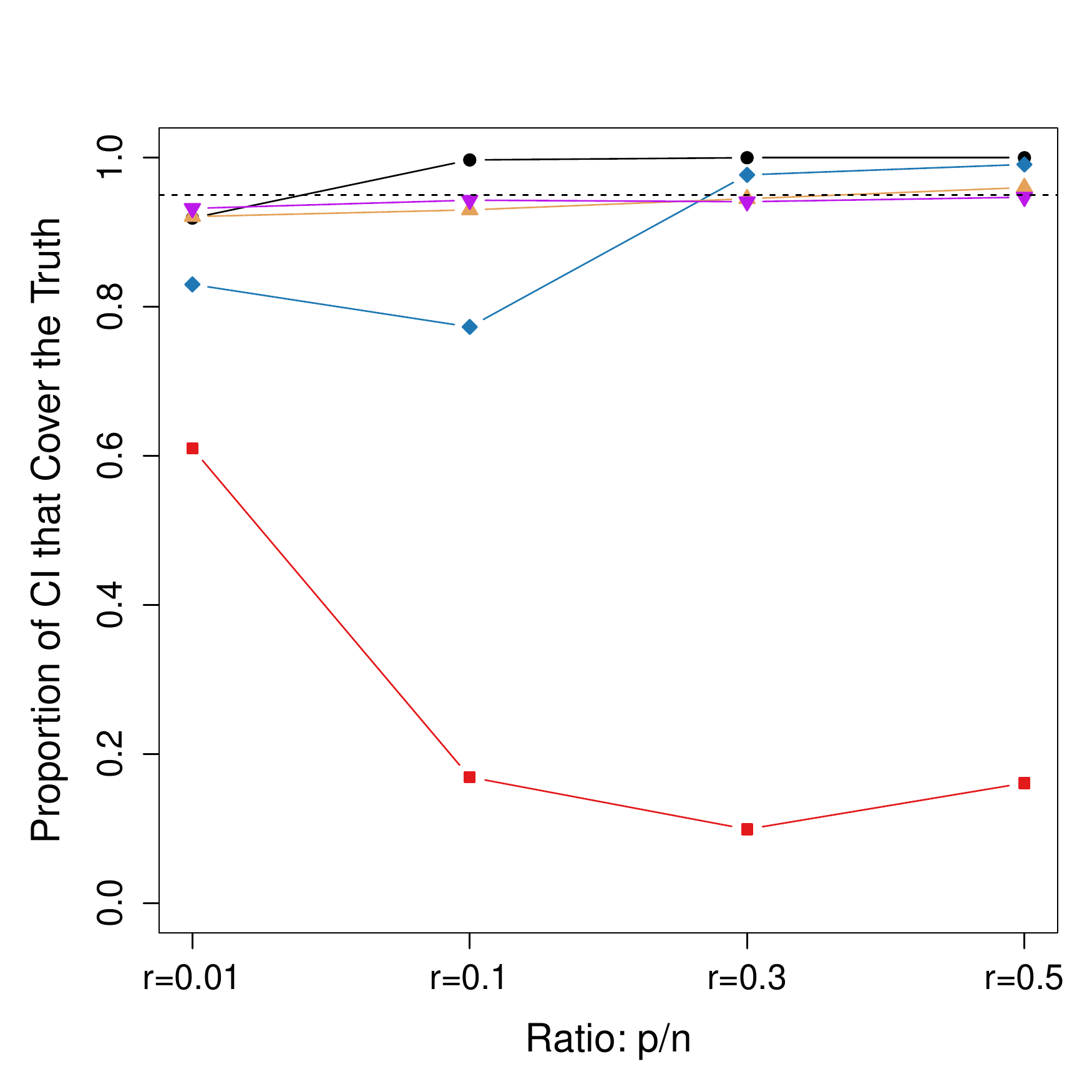} \label{subfig:bootCI:EllipExpNorm} }
\caption{
 \textbf{Gap Statistic: 95\% CI Coverage, $n=1,000$:  } \label{fig:bootCIGap}}
\end{figure}

\begin{figure}[t]
	\centering
	\subfloat[][\centering $\lambda_1=1$\par r=0.01 ]{\includegraphics[type=pdf,ext=.pdf,read=.pdf,width=\doubleFig]{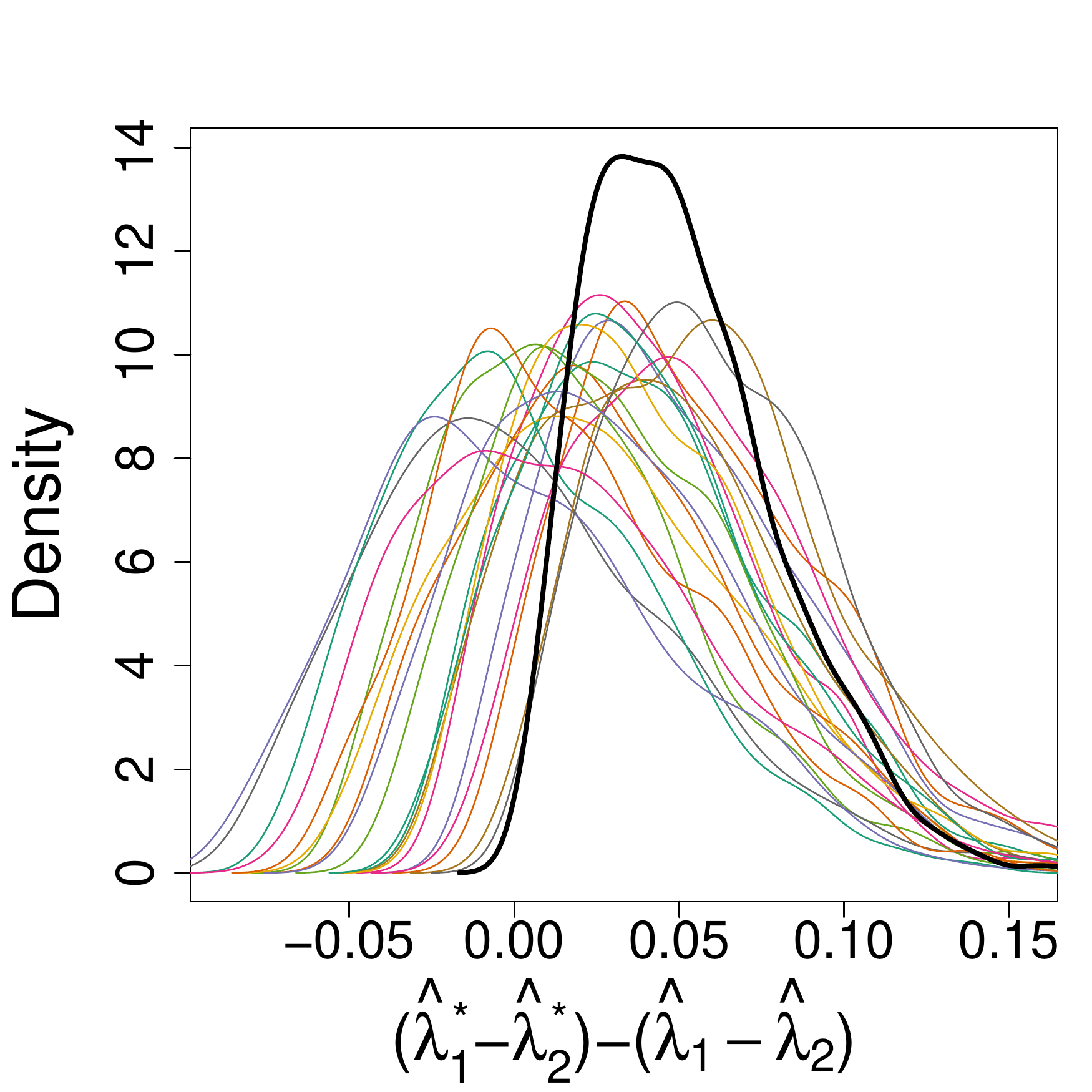} \label{subfig:bootDensityGapNorm:0.01_0} }
	\subfloat[][\centering  $\lambda_1=1$ \par r=0.3 ]{\includegraphics[type=pdf,ext=.pdf,read=.pdf,width=\doubleFig]{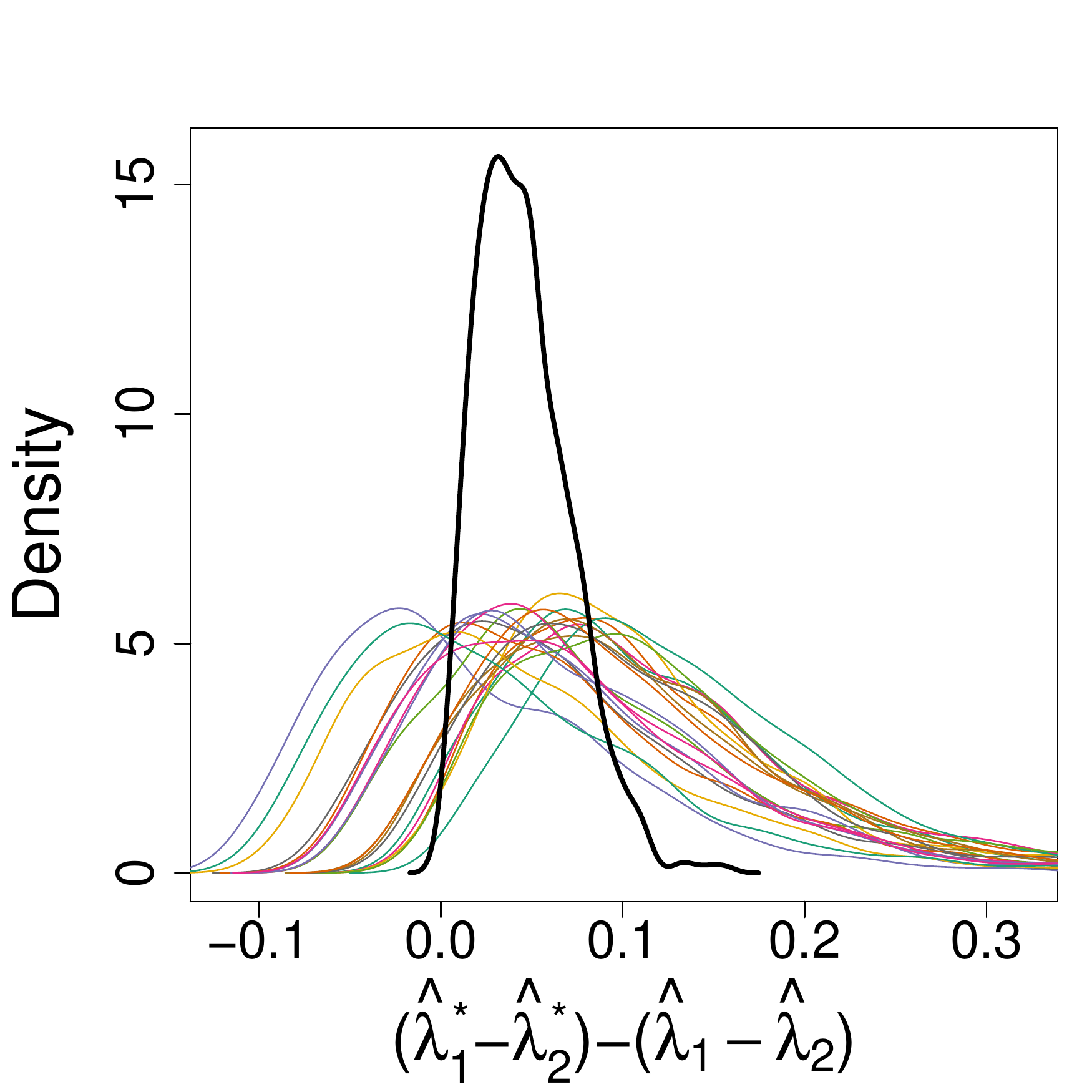} \label{subfig:bootDensityGapNorm:0.03_0} }
	\\
	\subfloat[][\centering  $\lambda_1=1+3\sqrt{r}$ \par r=0.01 ]{\includegraphics[type=pdf,ext=.pdf,read=.pdf,width=\doubleFig]{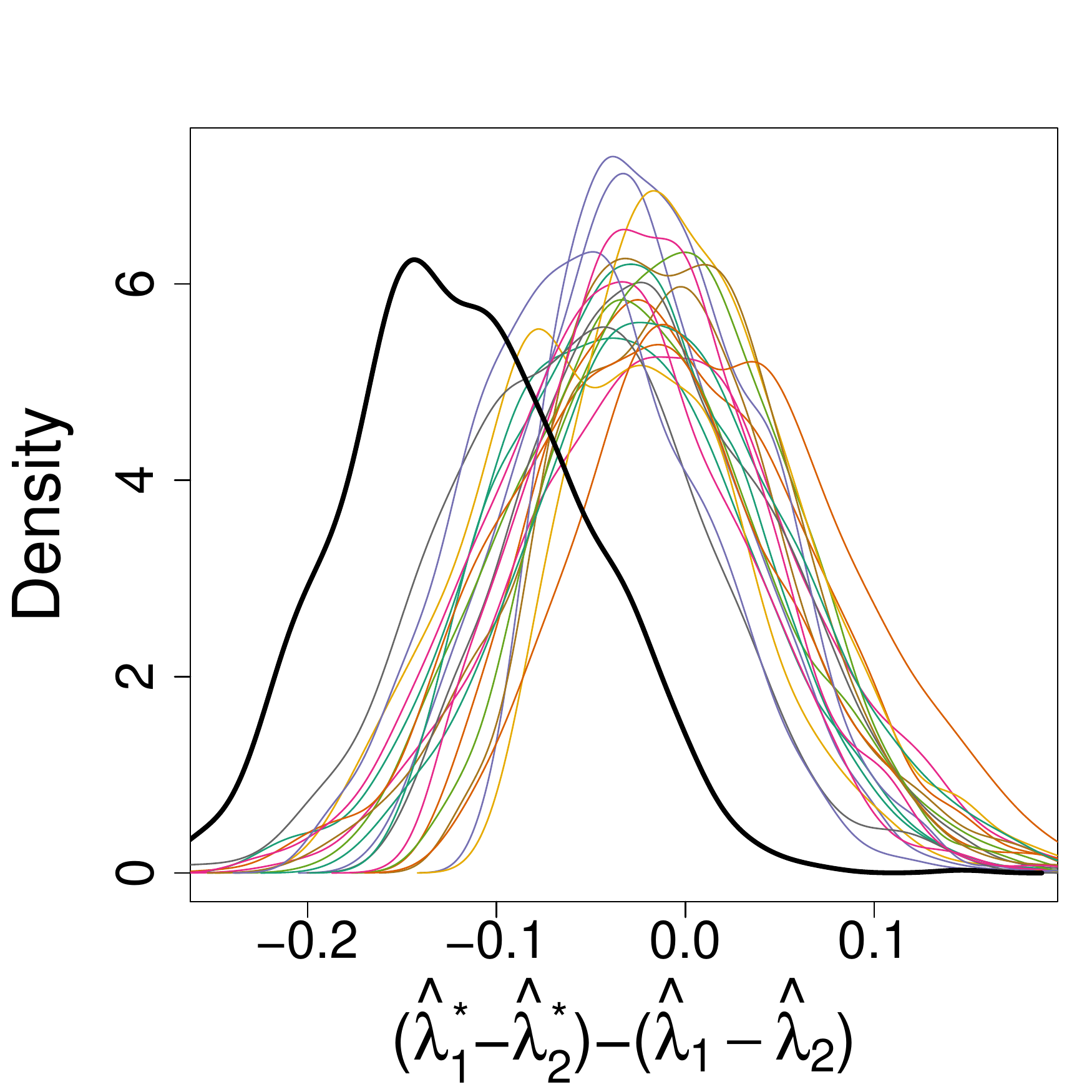} \label{subfig:bootDensityGapNorm:0.01_3} }
	\subfloat[][\centering  $\lambda_1=1+3\sqrt{r}$ \par r=0.3 ]{\includegraphics[type=pdf,ext=.pdf,read=.pdf,width=\doubleFig]{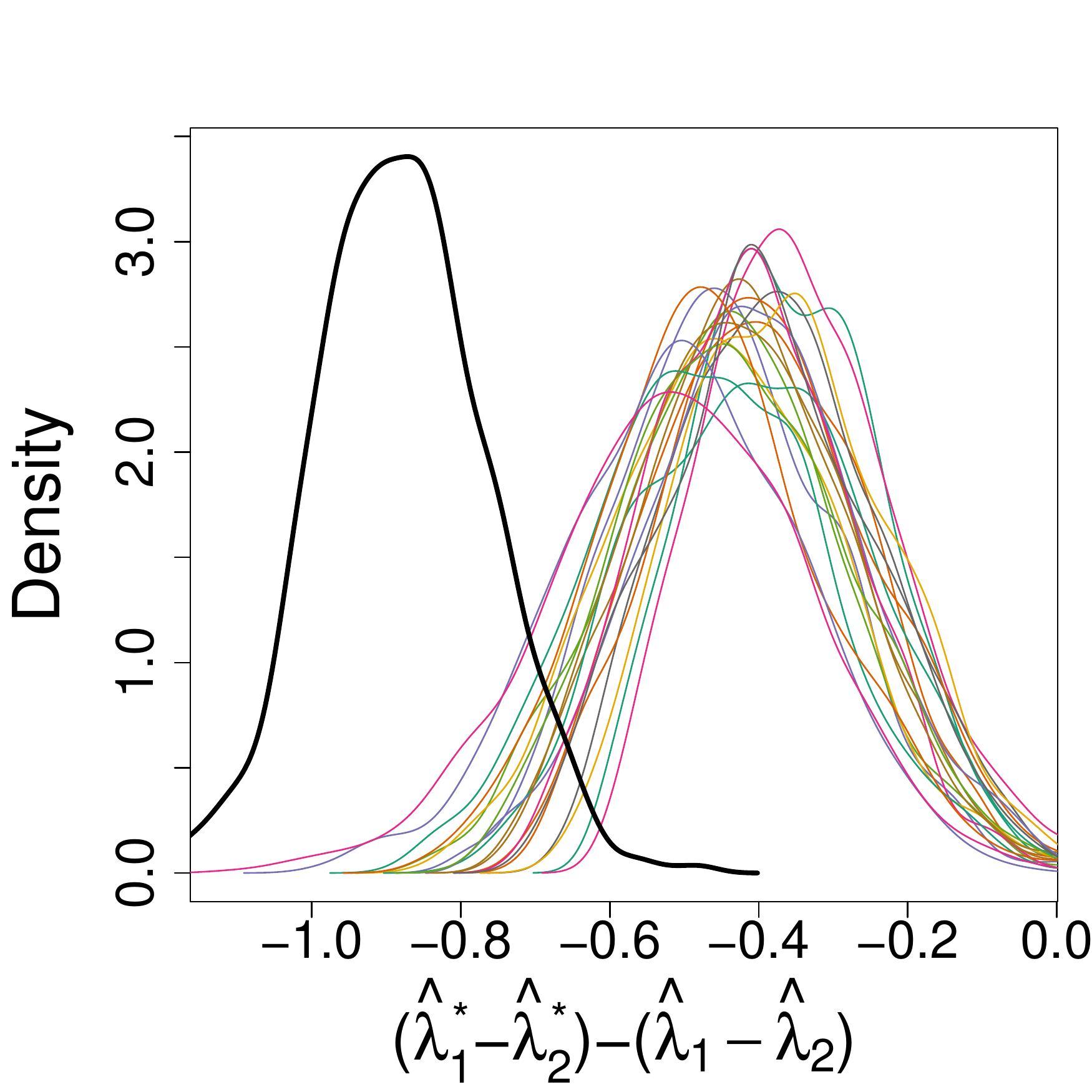} \label{subfig:bootDensityGapNorm:0.03_3} }
\caption{
 \textbf{Gap Statistic: Bootstrap distribution, $X_i\sim$ Normal, n=1,000:  } Plotted are the estimated density of twenty simulations of the bootstrap distribution of $(\hat{\lambda}^{*b}_1-\hat{\lambda}^{*b}_2)-(\hat{\lambda}_1-\hat{\lambda}_2)$, with $b=1,\ldots,999$. The solid black line line represents the distribution of $(\hat{\lambda}_1-\hat{\lambda}_2)-(\lambda_1-\lambda_2)$ over 1,000 simulations. 
}\label{fig:bootDensityGapNorm}
\end{figure}

\begin{figure}[t]
	\centering
	\subfloat[][\centering $\lambda_1=1$\par r=0.01]{\includegraphics[type=pdf,ext=.pdf,read=.pdf,width=\doubleFig]{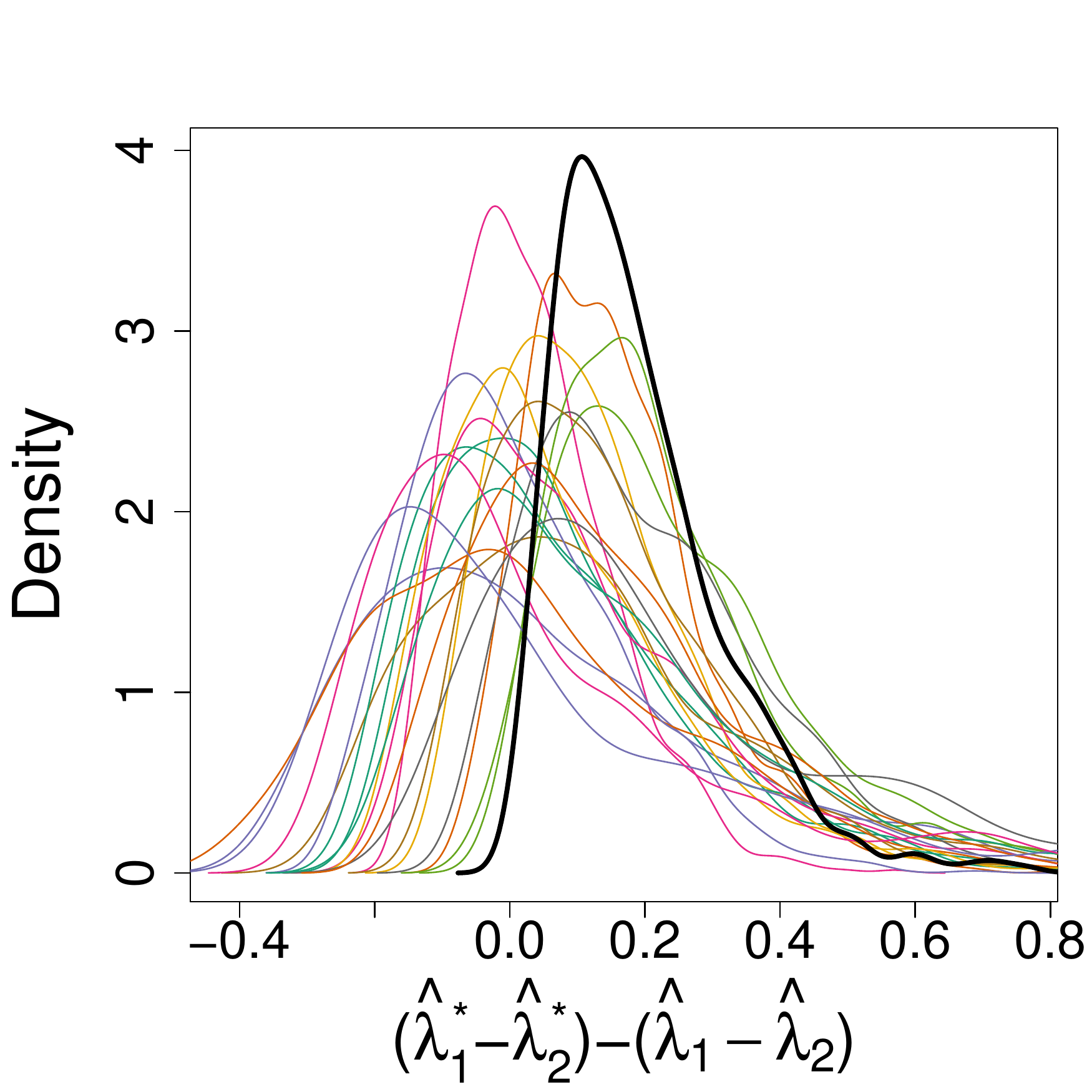} \label{subfig:bootDensityGapEllipExp:0.01_0} }
	\subfloat[][\centering  $\lambda_1=1$ \par r=0.3]{\includegraphics[type=pdf,ext=.pdf,read=.pdf,width=\doubleFig]{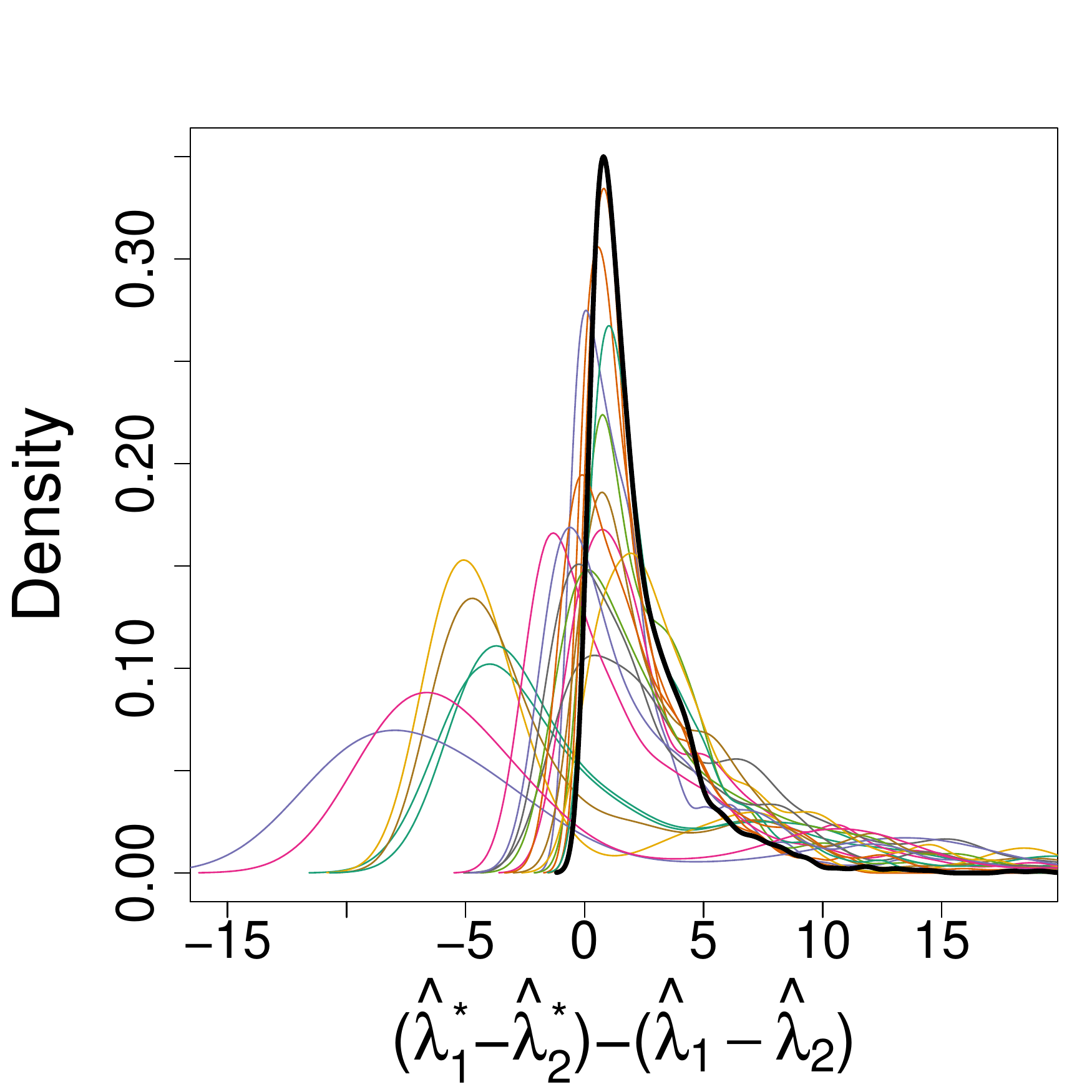} \label{subfig:bootDensityGapEllipExp:0.03_0} }
	\\
	\subfloat[][\centering  $\lambda_1=1+3\sqrt{r}$ \par r=0.01]{\includegraphics[type=pdf,ext=.pdf,read=.pdf,width=\doubleFig]{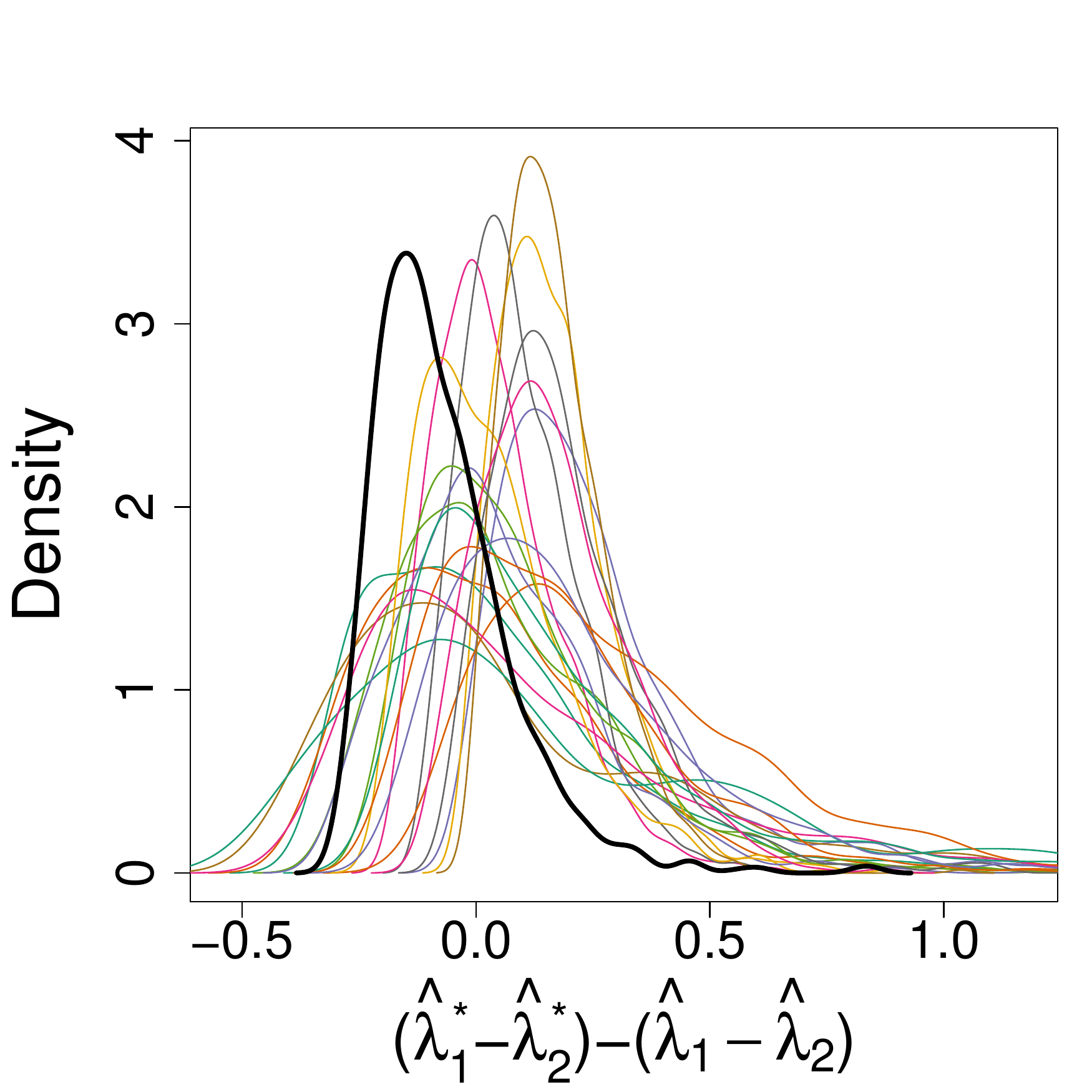} \label{subfig:bootDensityGapEllipExp:0.01_3} }
	\subfloat[][\centering  $\lambda_1=1+3\sqrt{r}$ \par r=0.3]{\includegraphics[type=pdf,ext=.pdf,read=.pdf,width=\doubleFig]{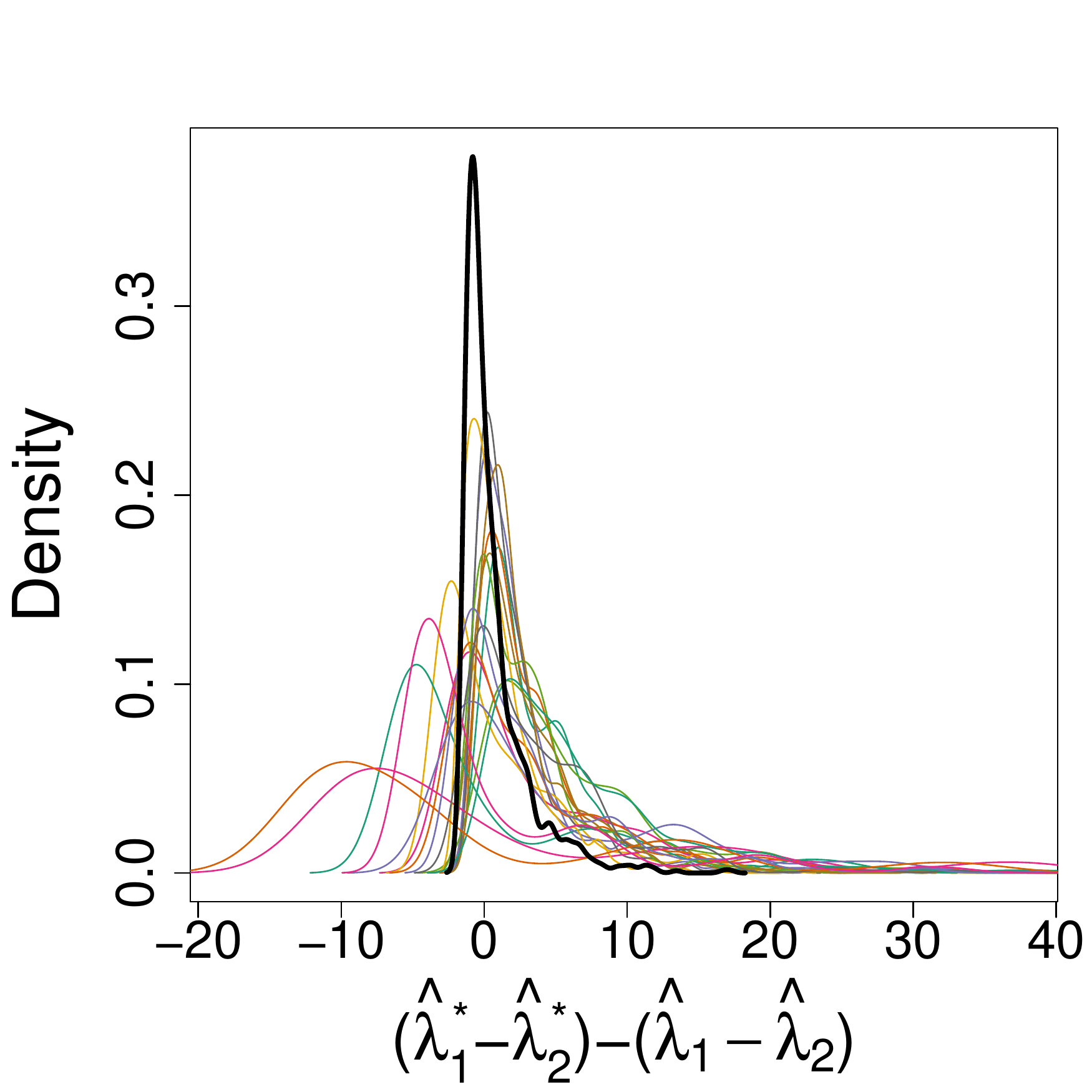} \label{subfig:bootDensityGapEllipExp:0.03_3} }
\caption{
 \textbf{Gap Statistic: Bootstrap distribution, $X_i\sim$ Ellip Exp, n=1,000:  } Plotted are the estimated density of twenty simulations of the bootstrap distribution of $(\hat{\lambda}^{*b}_1-\hat{\lambda}^{*b}_2)-(\hat{\lambda}_1-\hat{\lambda}_2)$, with $b=1,\ldots,999$. The solid black line line represents the distribution of $(\hat{\lambda}_1-\hat{\lambda}_2)-(\lambda_1-\lambda_2)$ over 1,000 simulations. 
}\label{fig:bootDensityGapEllipExp}
\end{figure}


\begin{figure}[t]
	\centering
	\subfloat[][$X_i\sim$ Normal]{\includegraphics[width=\doubleFig]{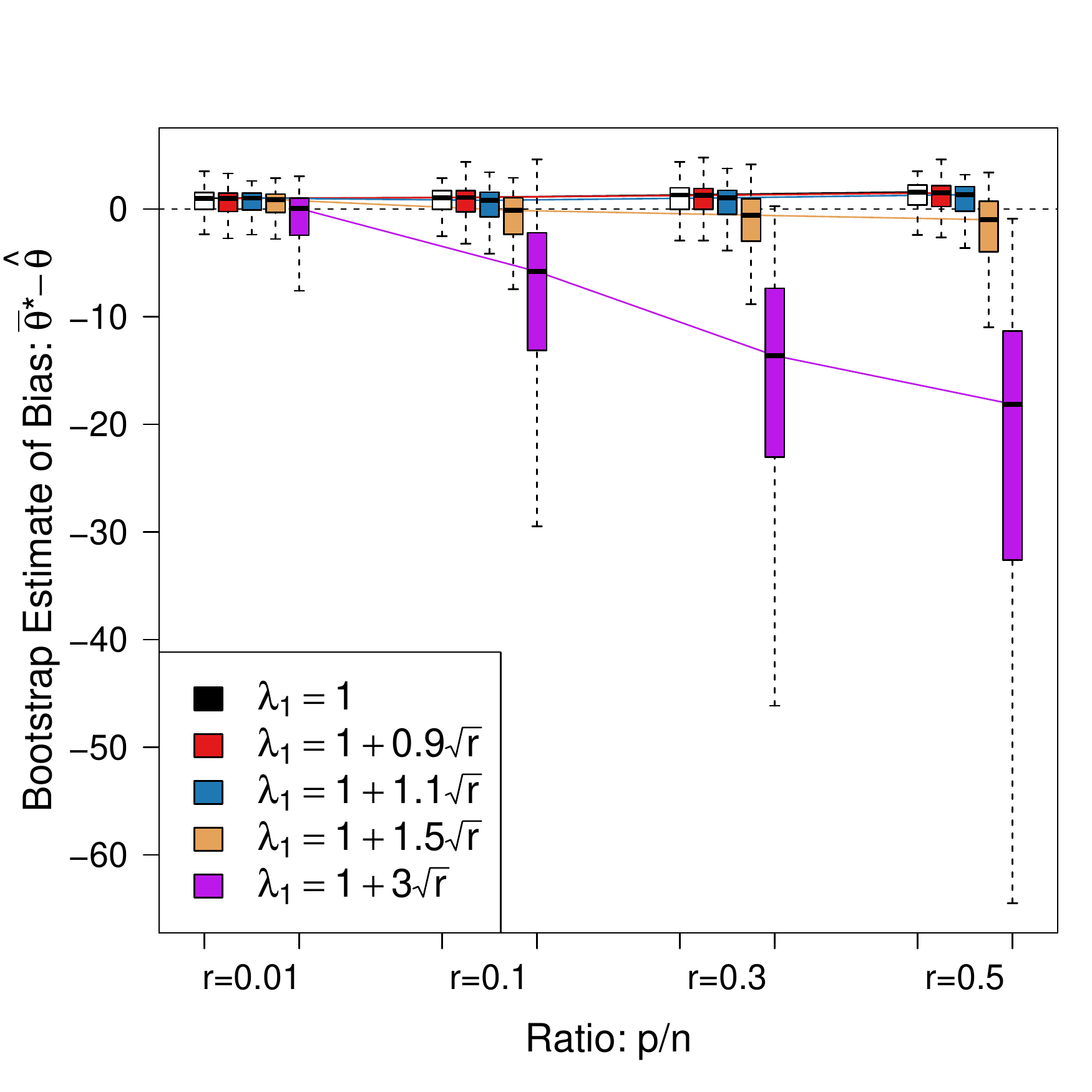} \label{subfig:bootBiasGapRatio:Norm} }
	\subfloat[][$X_i\sim$ Ellip. Exp]{\includegraphics[width=\doubleFig]{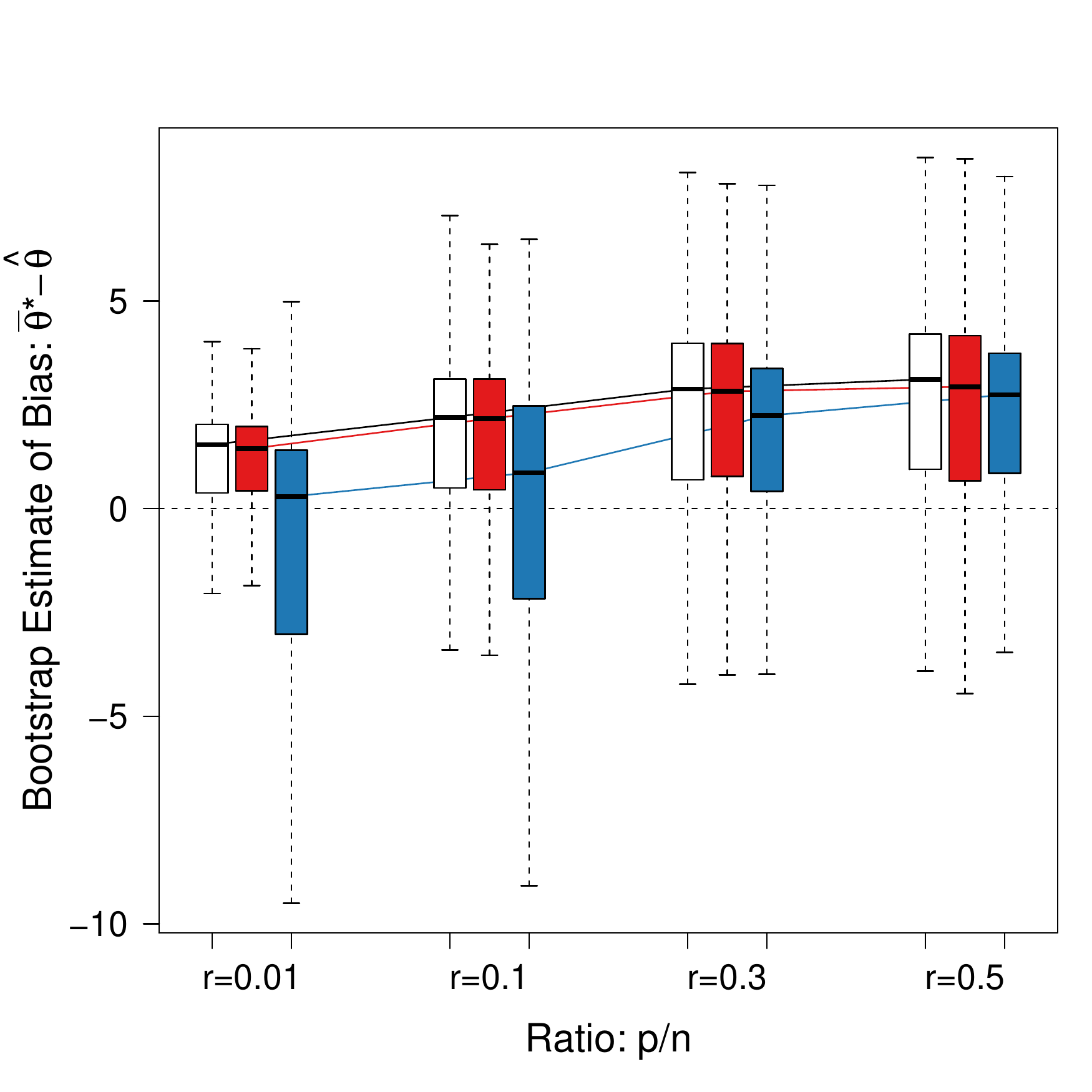} \label{subfig:bootBiasGapRatio:EllipExp} }\\
	\subfloat[][$X_i\sim$ Ellip. Norm]{\includegraphics[width=\doubleFig]{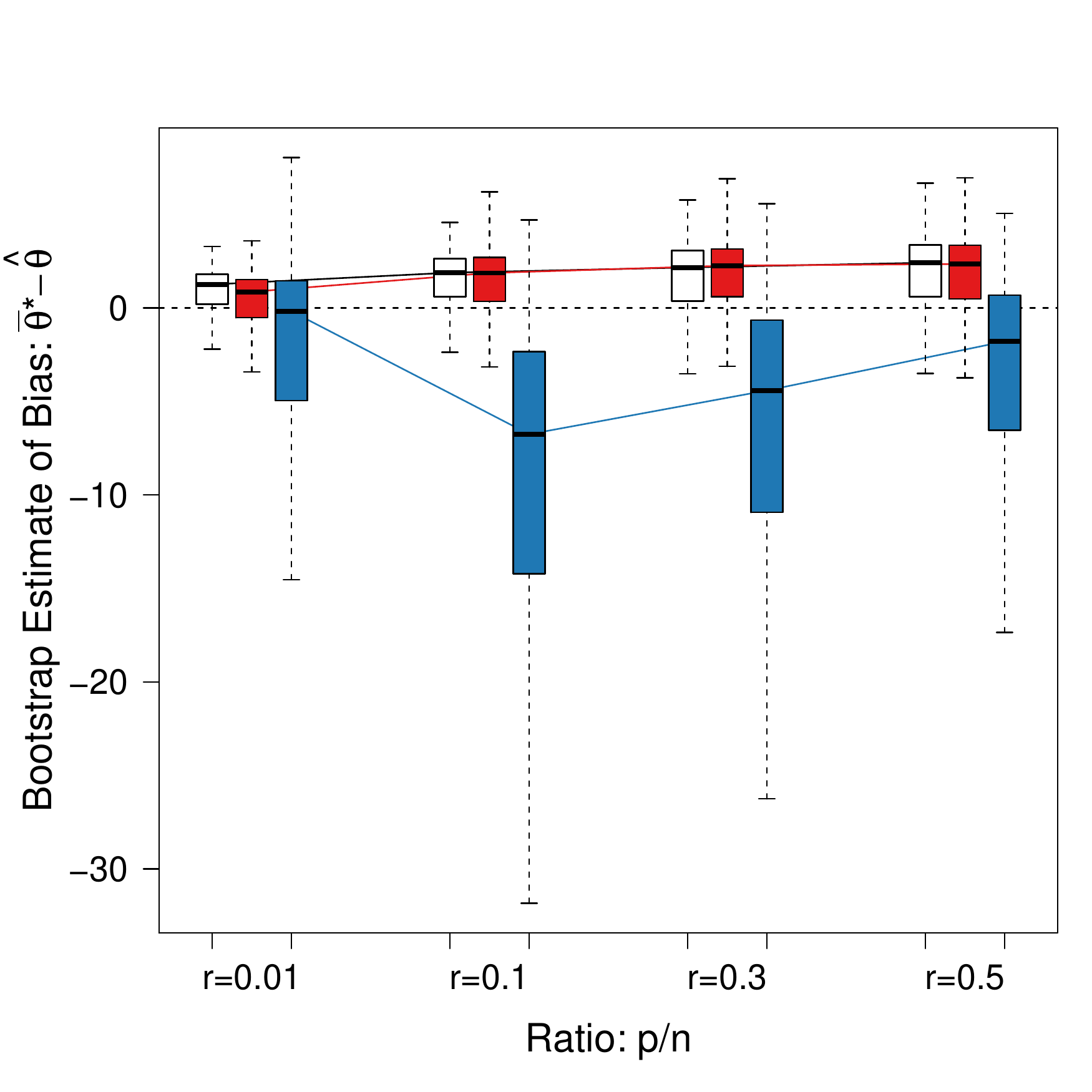} \label{subfig:bootBiasGapRatio:EllipNorm} }
	\subfloat[][$X_i\sim$ Ellip. Unif]{\includegraphics[width=\doubleFig]{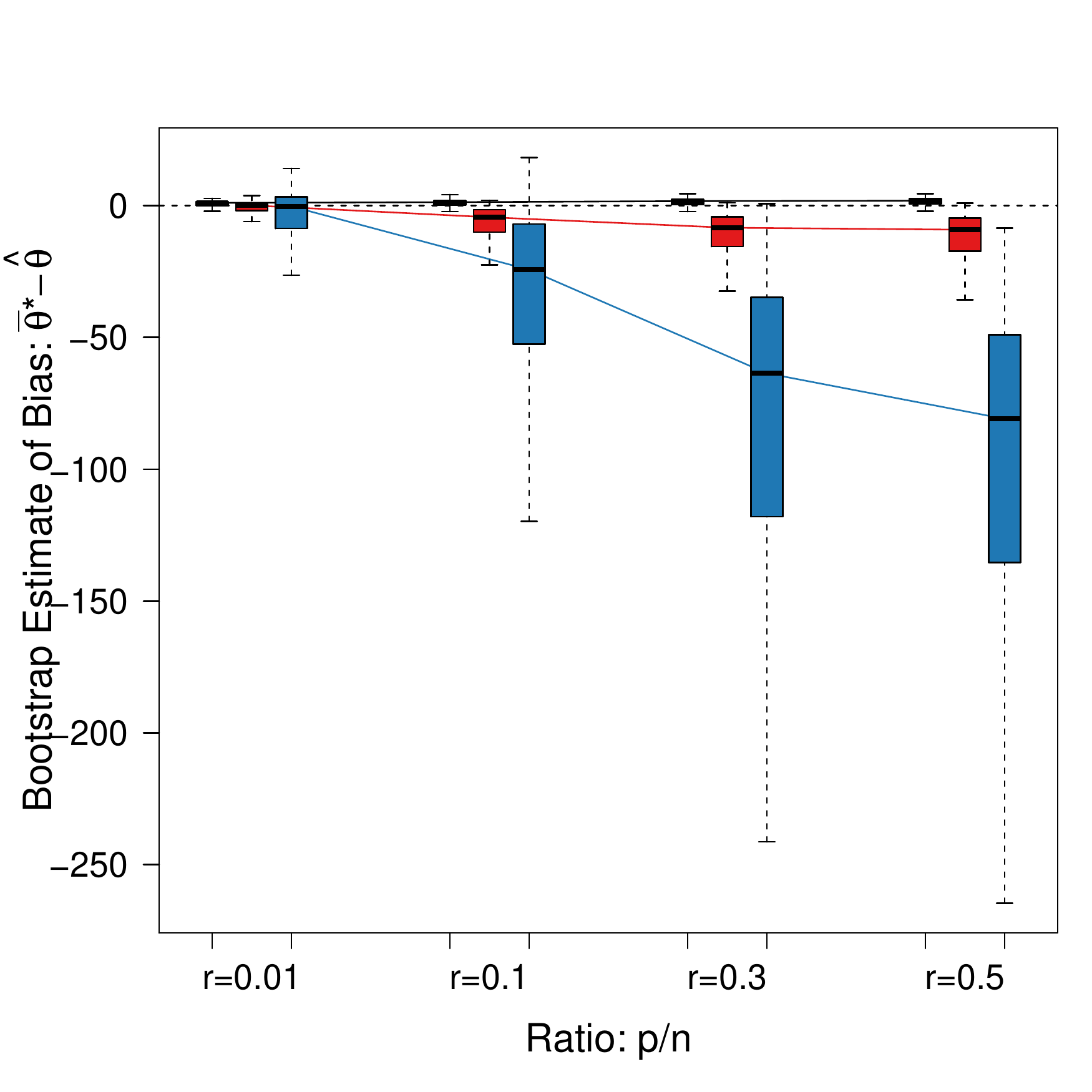} \label{subfig:bootBiasGapRatio:EllipUnif} }
	
\caption{
 \textbf{Gap Ratio Statistic: Bias of Bootstrap. }  Note that the Gap Ratio is not well defined in the population for this simulation (since $\lambda_2=\lambda_3$) so we can not scale the Gap Ratio by the true value of the Gap Ratio as was done for other plots in the Supplementary Figures. Instead we plot the actual bias, as in Figure \ref{fig:bootBias}.  For this reason, we only show the smaller values of $\lambda_1$ (otherwise, without scaling, the plot is dominated by the bias of large values of $\lambda_1$, even though the relative value of the bias is small). Similarly, the true bias of the estimate $\hat{\lambda}_1-\hat{\lambda}_2$ is not a well-defined quantity and hence is not plotted. 
}\label{fig:bootBiasGapRatio}
\end{figure}

\begin{figure}[t]
	\centering
	\subfloat[][$X_i\sim$ Normal]{\includegraphics[width=\doubleFig]{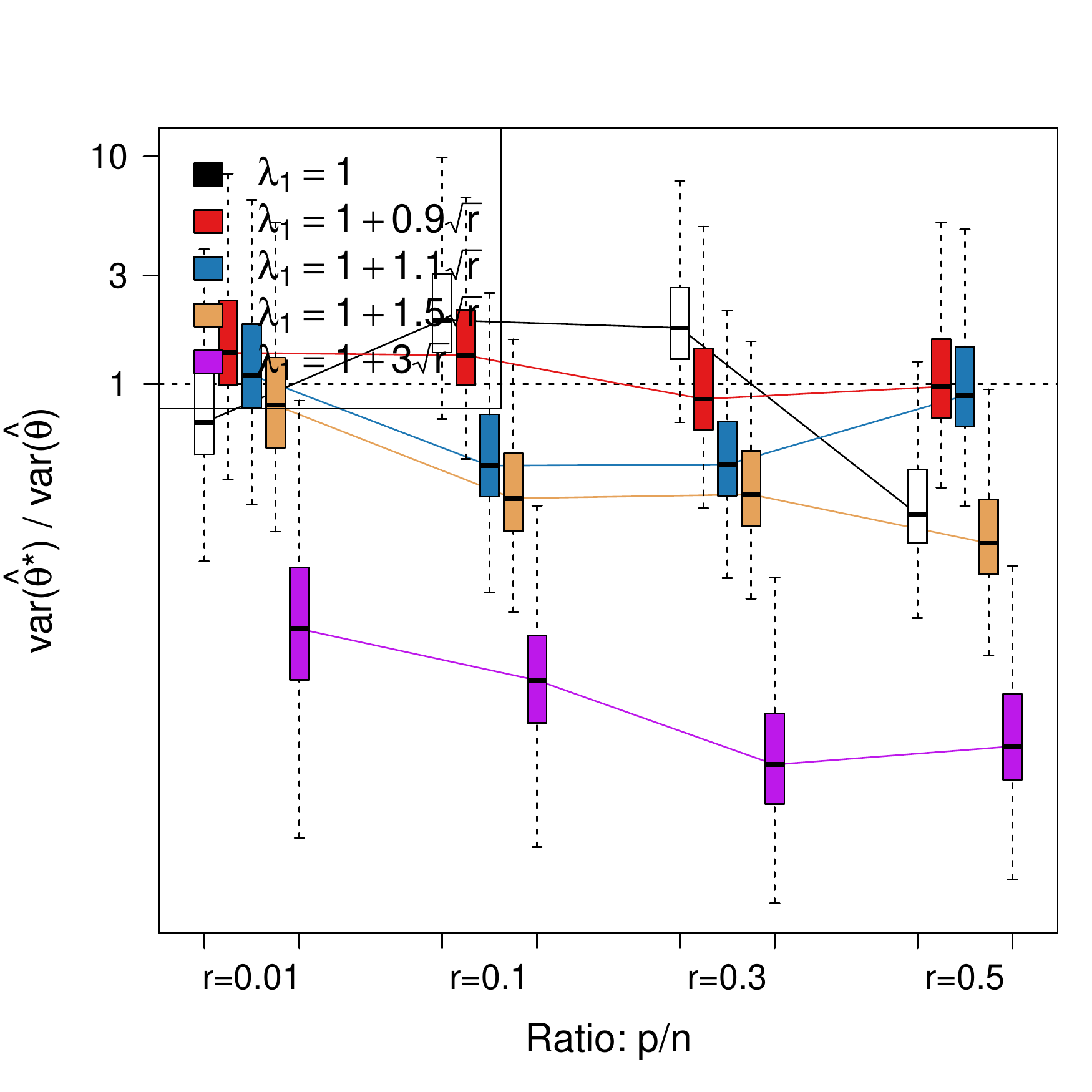} \label{subfig:bootVarGapRatio:Norm} }
	\subfloat[][$X_i\sim$ Ellip. Exp]{\includegraphics[width=\doubleFig]{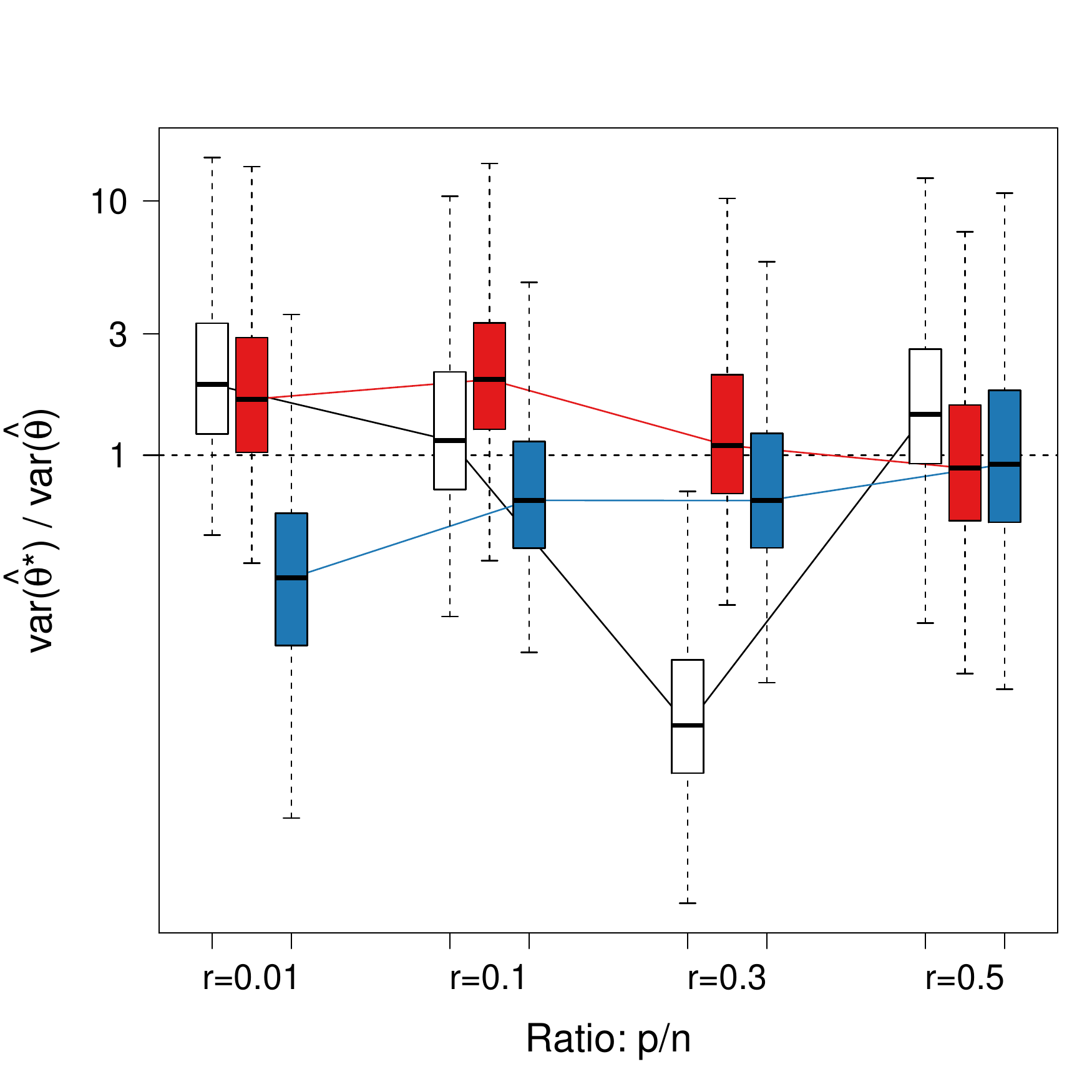} \label{subfig:bootVarGapRatio:EllipExp} }\\
	\subfloat[][$X_i\sim$ Ellip. Normal]{\includegraphics[width=\doubleFig]{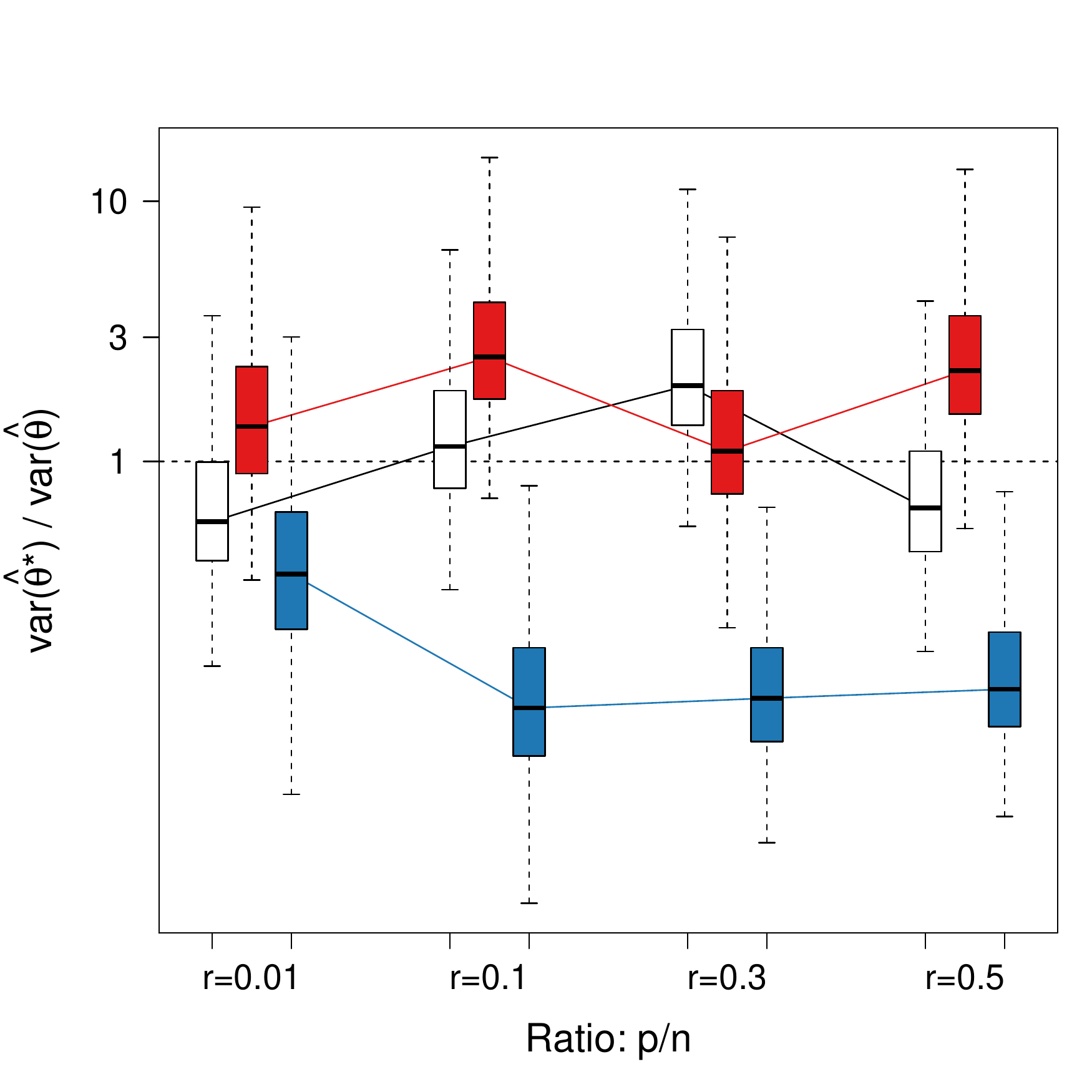} \label{subfig:bootVarGapRatio:EllipNorm} }
	\subfloat[][$X_i\sim$ Ellip. Uniform]{\includegraphics[width=\doubleFig]{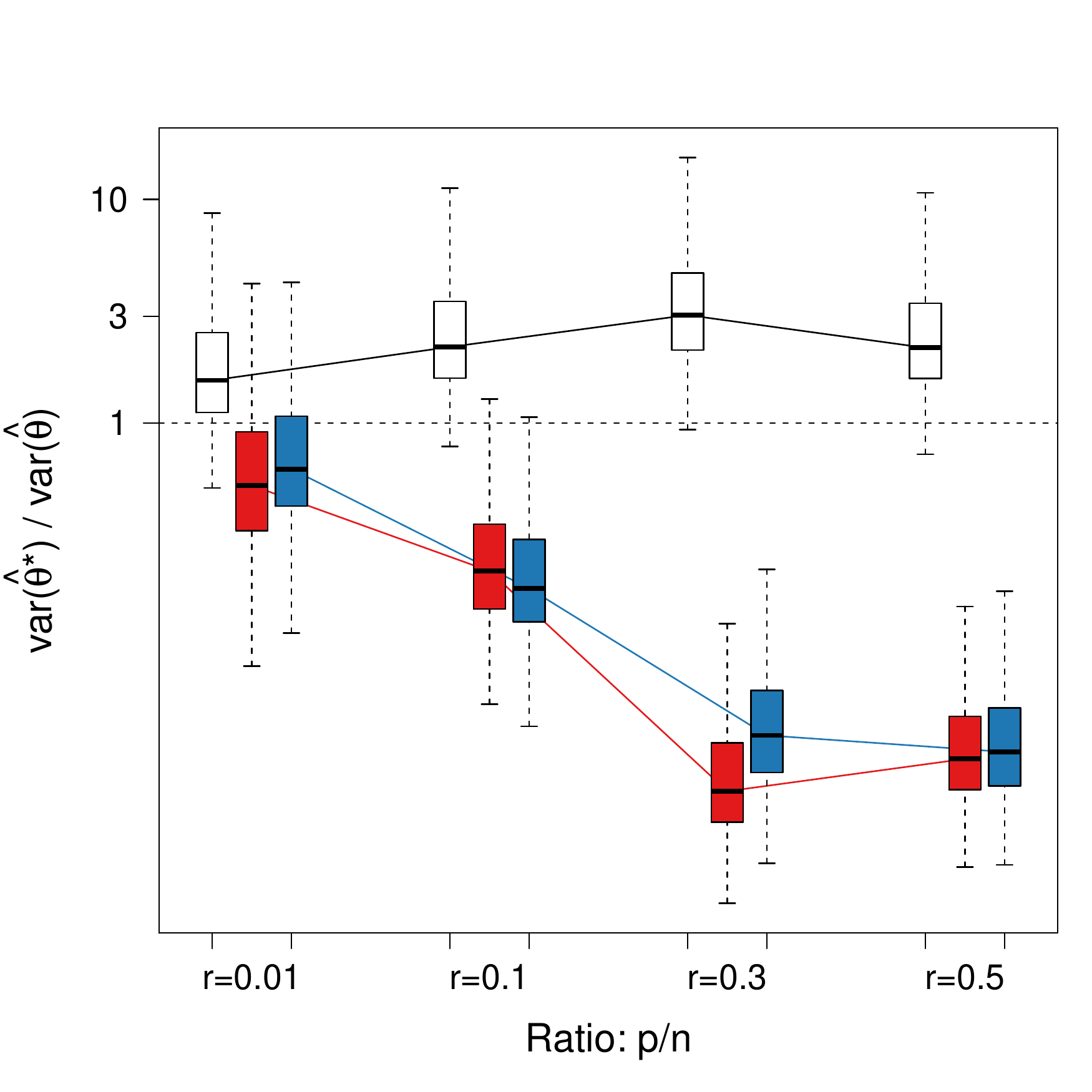} \label{subfig:bootVarGapRatio:EllipUnif} }
	
\caption{
 \textbf{Gap Ratio Statistic: Ratio of Bootstrap Estimate of Variance to True Variance.}
}\label{fig:bootVarGapRatio}
\end{figure}

\begin{figure}[t]
	\centering
	\subfloat[][\centering $\lambda_1=1$\par r=0.01]{\includegraphics[type=pdf,ext=.pdf,read=.pdf,width=.3\textwidth]{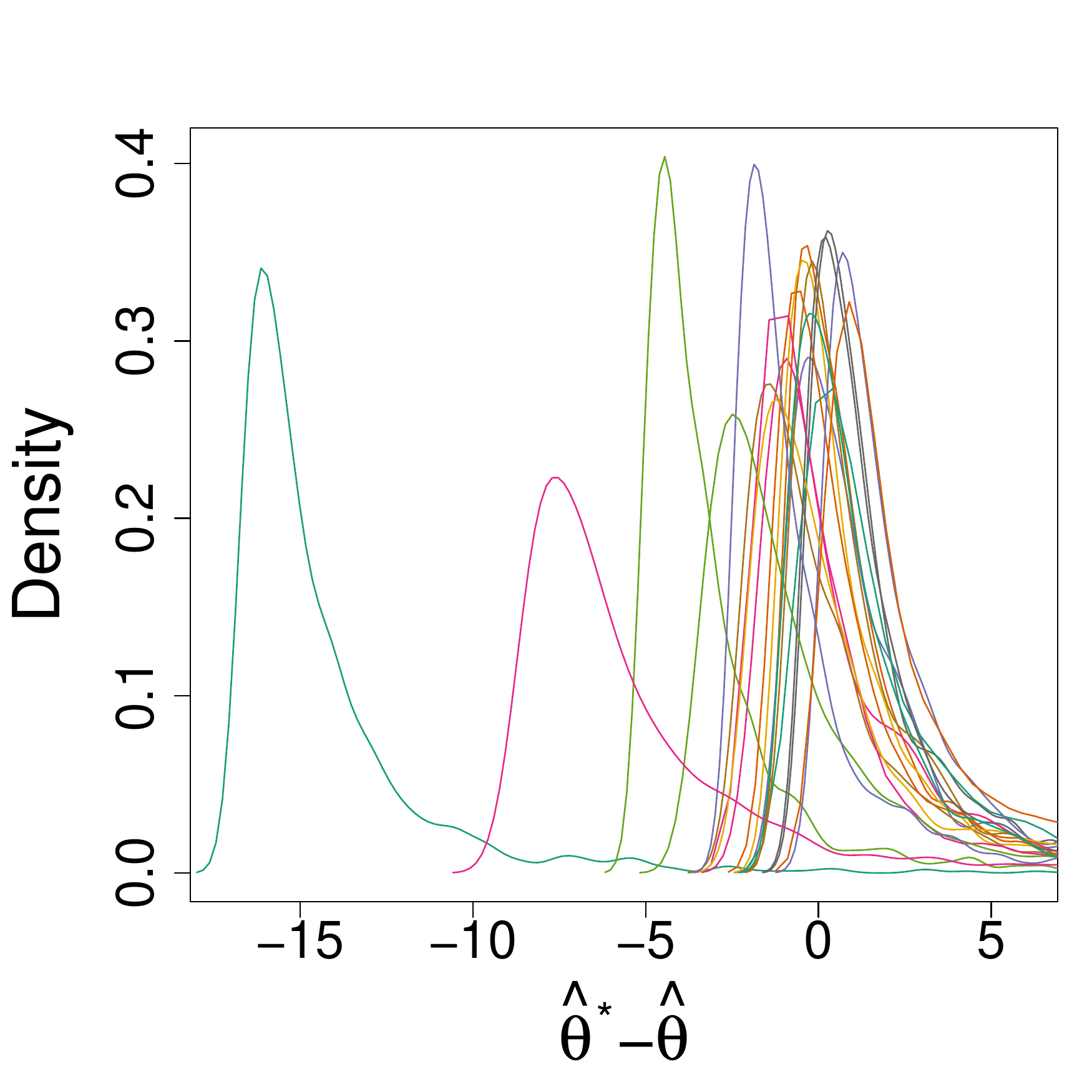} \label{subfig:bootDensityGapRatioEllipExp:0.01_0} }
	\subfloat[][\centering  $\lambda_1=1$ \par r=0.3]{\includegraphics[type=pdf,ext=.pdf,read=.pdf,width=.3\textwidth]{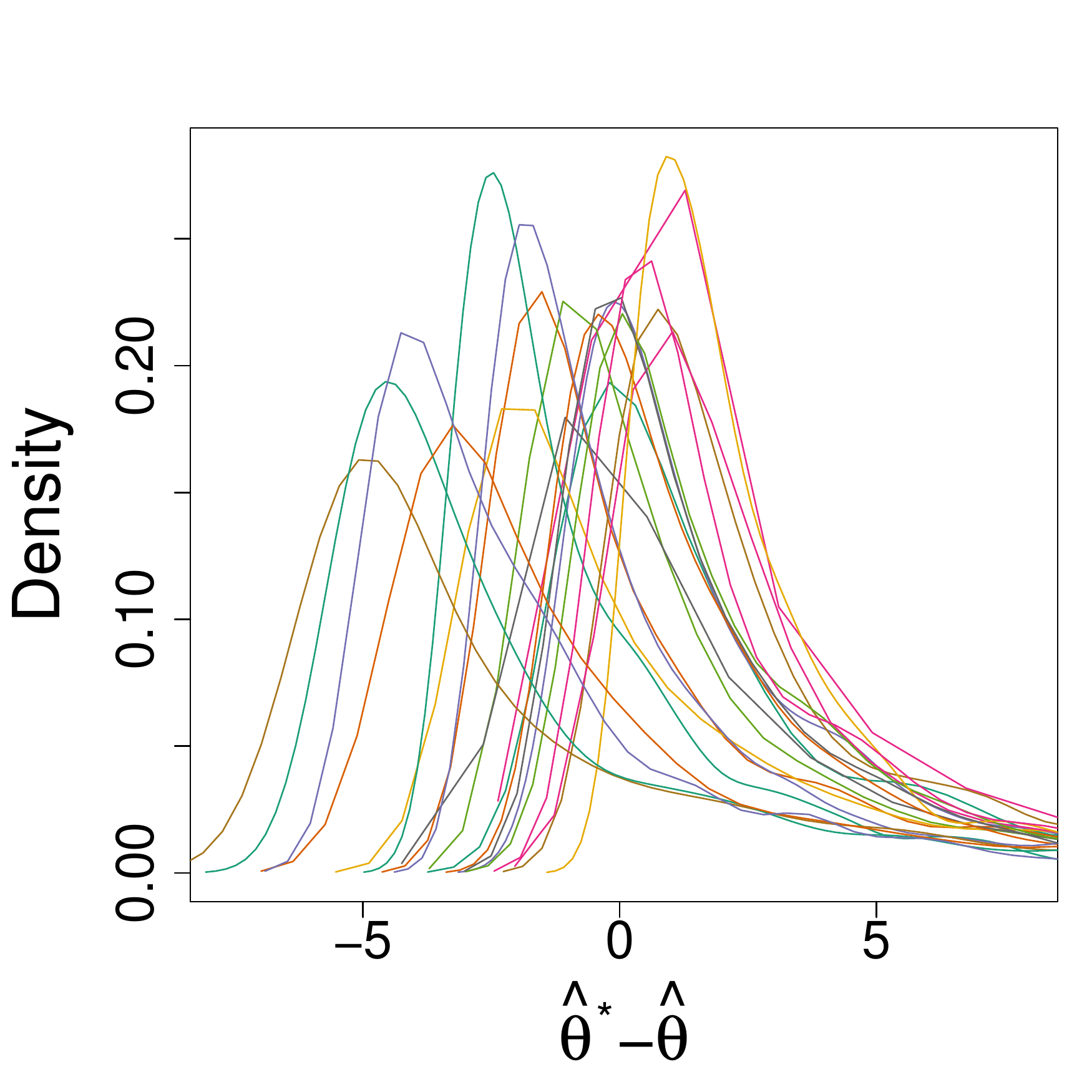} \label{subfig:bootDensityGapRatioEllipExp:0.03_0} }
	\\
	\subfloat[][\centering  $\lambda_1=1+3\sqrt{r}$ \par r=0.01]{\includegraphics[type=pdf,ext=.pdf,read=.pdf,width=.3\textwidth]{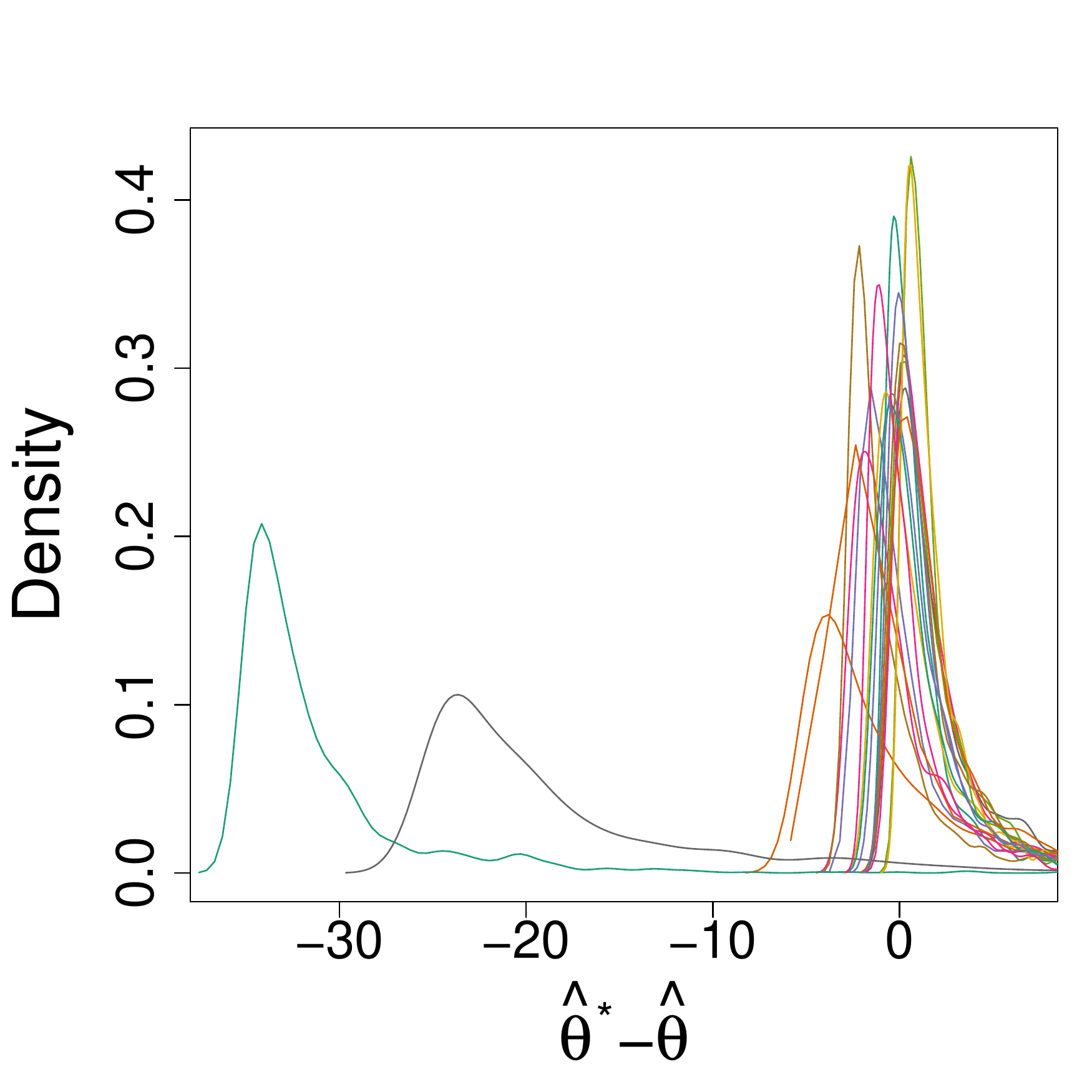} \label{subfig:bootDensityGapRatioEllipExp:0.01_3} }
	\subfloat[][\centering  $\lambda_1=1+3\sqrt{r}$ \par r=0.3]{\includegraphics[type=pdf,ext=.pdf,read=.pdf,width=.3\textwidth]{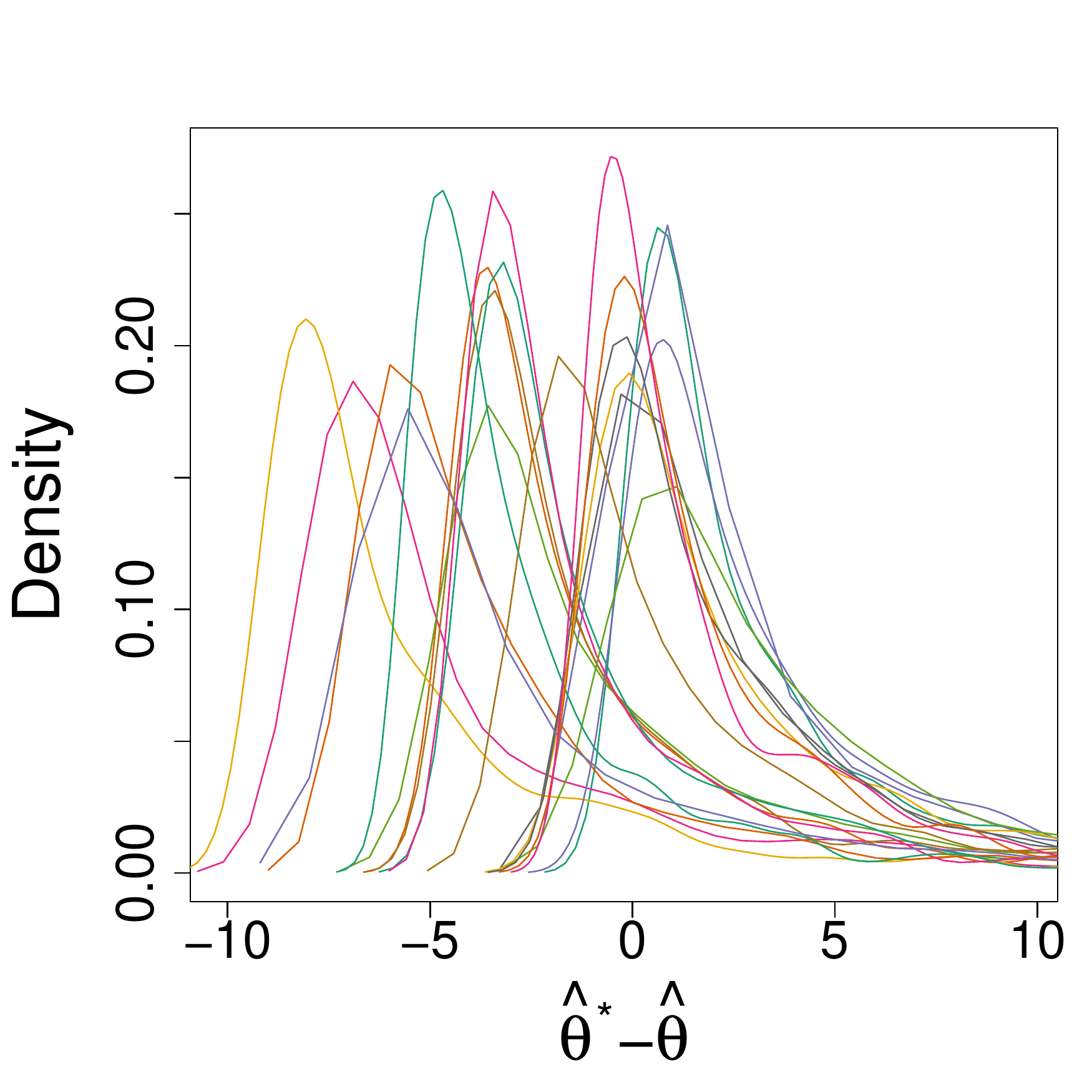} \label{subfig:bootDensityGapRatioEllipExp:0.03_3} }
\caption{
 \textbf{Gap Ratio Statistic: Bootstrap distribution, $X_i\sim$ Ellip Exp, n=1,000:  } Plotted are the estimated density of twenty simulations of the bootstrap distribution of $(\hat{\lambda}^{*b}_1-\hat{\lambda}^{*b}_2)-(\hat{\lambda}_1-\hat{\lambda}_2)$, with $b=1,\ldots,999$. The solid black line line represents the distribution of $(\hat{\lambda}_1-\hat{\lambda}_2)-(\lambda_1-\lambda_2)$ over 1,000 simulations. 
}\label{fig:bootDensityGapRatioEllipExp}
\end{figure}

\clearpage

\begin{center}
\textbf{\textsc{Supplementary Tables}}
\end{center}
\vspace{1cm}

\begin{table}
	\centering
	\subfloat[][Bootstrap Median Estimate of Bias]{
\begin{tabular}{rrrrrr}
  \hline
 & $\lambda_1=1$ & $\lambda_1=1+ 3 \sqrt{r}$ & $\lambda_1=1+ 11 \sqrt{r}$ & $\lambda_1=1+ 50 \sqrt{r}$ & $\lambda_1=1+ 100 \sqrt{r}$ \\ 
  \hline
$ r=0.01 $ & 0.08 & 0.04 & 0.02 & 0.01 & 0.01 \\ 
  $ r=0.1 $ & 0.42 & 0.23 & 0.13 & 0.11 & 0.10 \\ 
  $ r=0.3 $ & 1.04 & 0.60 & 0.37 & 0.31 & 0.31 \\ 
  $ r=0.5 $ & 1.70 & 1.00 & 0.60 & 0.52 & 0.51 \\ 
   \hline
\end{tabular}
		
	 }\\
	\subfloat[][True Bias]{
\begin{tabular}{rrrrrr}
  \hline
 & $\lambda_1=1$ & $\lambda_1=1+ 3 \sqrt{r}$ & $\lambda_1=1+ 11 \sqrt{r}$ & $\lambda_1=1+ 50 \sqrt{r}$ & $\lambda_1=1+ 100 \sqrt{r}$ \\ 
  \hline
$ r=0.01 $ & 0.17 & 0.04 & 0.02 & 0.02 & 0.02 \\ 
  $ r=0.1 $ & 0.70 & 0.20 & 0.12 & 0.12 & 0.12 \\ 
  $ r=0.3 $ & 1.37 & 0.48 & 0.36 & 0.28 & 0.38 \\ 
  $ r=0.5 $ & 1.88 & 0.73 & 0.55 & 0.48 & 0.45 \\ 
   \hline
\end{tabular}
		
	 }
\caption{\textbf{Top Eigenvalue: Median value of Bootstrap and True values of Bias, $Z\sim$ Normal} This tables give the median values of the boxplots plotted in Figure \ref{fig:bootBias}, as well as the true bias values (*) in the plots. See figure caption for more details. }
\label{tab:bootBiasNorm}
\end{table}

\begin{table}
	\centering
	\subfloat[][Bootstrap Median Estimate of Bias]{
\begin{tabular}{rrrrrr}
  \hline
 & $\lambda_1=1$ & $\lambda_1=1+ 3 \sqrt{r}$ & $\lambda_1=1+ 11 \sqrt{r}$ & $\lambda_1=1+ 50 \sqrt{r}$ & $\lambda_1=1+ 100 \sqrt{r}$ \\ 
  \hline
$ r=0.01 $ & 0.23 & 0.22 & 0.12 & 0.06 & 0.05 \\ 
  $ r=0.1 $ & 2.25 & 2.24 & 1.75 & 0.67 & 0.58 \\ 
  $ r=0.3 $ & 6.70 & 6.78 & 6.35 & 2.37 & 1.88 \\ 
  $ r=0.5 $ & 11.19 & 11.12 & 10.91 & 4.70 & 3.37 \\ 
   \hline
\end{tabular}

	 }\\
	\subfloat[][True Bias]{
\begin{tabular}{rrrrrr}
  \hline
 & $\lambda_1=1$ & $\lambda_1=1+ 3 \sqrt{r}$ & $\lambda_1=1+ 11 \sqrt{r}$ & $\lambda_1=1+ 50 \sqrt{r}$ & $\lambda_1=1+ 100 \sqrt{r}$ \\ 
  \hline
$ r=0.01 $ & 0.49 & 0.27 & 0.12 & 0.06 & 0.01 \\ 
  $ r=0.1 $ & 3.34 & 2.39 & 1.10 & 0.54 & 0.75 \\ 
  $ r=0.3 $ & 9.17 & 7.57 & 4.10 & 1.92 & 1.69 \\ 
  $ r=0.5 $ & 14.93 & 12.76 & 8.08 & 3.69 & 3.45 \\ 
   \hline
\end{tabular}
		
	 }
	 \caption{\textbf{Top Eigenvalue: Median value of Bootstrap and True values of Bias, $Z\sim$ Ellip Exp} This tables give the median values of the boxplots plotted in Figure \ref{fig:bootBias}, as well as the true bias values (*) in the plots. See figure caption for more details. }
	 \label{tab:bootBiasEllipExp}
	 \end{table}

	 \begin{table}
	 	\centering
	 	\subfloat[][Bootstrap Median Estimate of Bias]{
\begin{tabular}{rrrrrr}
  \hline
 & $\lambda_1=1$ & $\lambda_1=1+ 3 \sqrt{r}$ & $\lambda_1=1+ 11 \sqrt{r}$ & $\lambda_1=1+ 50 \sqrt{r}$ & $\lambda_1=1+ 100 \sqrt{r}$ \\ 
  \hline
$ r=0.01 $ & 0.16 & 0.13 & 0.05 & 0.03 & 0.03 \\ 
  $ r=0.1 $ & 1.38 & 1.33 & 0.55 & 0.32 & 0.30 \\ 
  $ r=0.3 $ & 4.17 & 4.16 & 2.30 & 1.00 & 0.93 \\ 
  $ r=0.5 $ & 6.95 & 6.96 & 4.80 & 1.72 & 1.57 \\ 
   \hline
\end{tabular}

	 	 }\\
	 	\subfloat[][True Bias]{
\begin{tabular}{rrrrrr}
  \hline
 & $\lambda_1=1$ & $\lambda_1=1+ 3 \sqrt{r}$ & $\lambda_1=1+ 11 \sqrt{r}$ & $\lambda_1=1+ 50 \sqrt{r}$ & $\lambda_1=1+ 100 \sqrt{r}$ \\ 
  \hline
$ r=0.01 $ & 0.32 & 0.13 & 0.05 & 0.04 & 0.07 \\ 
  $ r=0.1 $ & 1.65 & 0.83 & 0.41 & 0.35 & 0.13 \\ 
  $ r=0.3 $ & 4.14 & 2.52 & 1.20 & 0.78 & 0.92 \\ 
  $ r=0.5 $ & 6.57 & 4.45 & 2.01 & 1.64 & 1.76 \\ 
   \hline
\end{tabular}
		
	 	 }\\
\caption{\textbf{Top Eigenvalue: Median value of Bootstrap and True values of Bias, $Z\sim$ Ellip Norm}}
\label{tab:bootBiasEllipNorm}
\end{table}

	 \begin{table}
	 	\centering
	 	\subfloat[][Bootstrap Median Estimate of Bias]{
\begin{tabular}{rrrrrr}
  \hline
 & $\lambda_1=1$ & $\lambda_1=1+ 3 \sqrt{r}$ & $\lambda_1=1+ 11 \sqrt{r}$ & $\lambda_1=1+ 50 \sqrt{r}$ & $\lambda_1=1+ 100 \sqrt{r}$ \\ 
  \hline
$ r=0.01 $ & 0.09 & 0.05 & 0.02 & 0.01 & 0.01 \\ 
  $ r=0.1 $ & 0.55 & 0.34 & 0.17 & 0.14 & 0.13 \\ 
  $ r=0.3 $ & 1.56 & 1.07 & 0.49 & 0.40 & 0.39 \\ 
  $ r=0.5 $ & 2.64 & 1.96 & 0.82 & 0.67 & 0.65 \\ 
   \hline
\end{tabular}

	 	 }\\
	 	\subfloat[][True Bias]{
\begin{tabular}{rrrrrr}
  \hline
 & $\lambda_1=1$ & $\lambda_1=1+ 3 \sqrt{r}$ & $\lambda_1=1+ 11 \sqrt{r}$ & $\lambda_1=1+ 50 \sqrt{r}$ & $\lambda_1=1+ 100 \sqrt{r}$ \\ 
  \hline
$ r=0.01 $ & 0.20 & 0.05 & 0.02 & 0.00 & 0.02 \\ 
  $ r=0.1 $ & 0.83 & 0.27 & 0.16 & 0.12 & 0.22 \\ 
  $ r=0.3 $ & 1.66 & 0.62 & 0.45 & 0.45 & 0.28 \\ 
  $ r=0.5 $ & 2.33 & 0.98 & 0.72 & 0.63 & 0.63 \\ 
   \hline
\end{tabular}
		
	 	 }\\
\caption{\textbf{Top Eigenvalue: Median value of Bootstrap and True values of Bias, $Z\sim$ Ellip Uniform}}
\label{tab:bootBiasEllipUnif}
\end{table}

\begin{table}
	\centering
	\subfloat[][ $Z\sim$ Normal]{
\begin{tabular}{rrrrrr}
  \hline
 & $\lambda_1=1$ & $\lambda_1=1+ 3 \sqrt{r}$ & $\lambda_1=1+ 11 \sqrt{r}$ & $\lambda_1=1+ 50 \sqrt{r}$ & $\lambda_1=1+ 100 \sqrt{r}$ \\ 
  \hline
$ r=0.01 $ & 1.69 & 1.03 & 1.04 & 0.98 & 1.03 \\ 
  $ r=0.1 $ & 4.35 & 1.07 & 1.09 & 0.96 & 0.95 \\ 
  $ r=0.3 $ & 19.88 & 1.48 & 1.12 & 1.03 & 1.06 \\ 
  $ r=0.5 $ & 60.27 & 2.21 & 1.16 & 1.00 & 1.02 \\ 
   \hline
\end{tabular}

	 }
\\
	\subfloat[][$Z\sim$ Ellip Exp]{
\begin{tabular}{rrrrrr}
  \hline
 & $\lambda_1=1$ & $\lambda_1=1+ 3 \sqrt{r}$ & $\lambda_1=1+ 11 \sqrt{r}$ & $\lambda_1=1+ 50 \sqrt{r}$ & $\lambda_1=1+ 100 \sqrt{r}$ \\ 
  \hline
$ r=0.01 $ & 1.67 & 1.64 & 1.05 & 0.72 & 0.74 \\ 
  $ r=0.1 $ & 2.53 & 2.88 & 3.13 & 0.99 & 0.87 \\ 
  $ r=0.3 $ & 3.25 & 3.29 & 4.22 & 1.34 & 0.87 \\ 
  $ r=0.5 $ & 3.02 & 2.96 & 2.84 & 1.82 & 1.05 \\ 
   \hline
\end{tabular}

	 }\\
	 	\subfloat[][$Z\sim$ Ellip Norm]{
\begin{tabular}{rrrrrr}
  \hline
 & $\lambda_1=1$ & $\lambda_1=1+ 3 \sqrt{r}$ & $\lambda_1=1+ 11 \sqrt{r}$ & $\lambda_1=1+ 50 \sqrt{r}$ & $\lambda_1=1+ 100 \sqrt{r}$ \\ 
  \hline
$ r=0.01 $ & 2.07 & 1.51 & 0.98 & 0.98 & 0.92 \\ 
  $ r=0.1 $ & 8.98 & 8.33 & 2.01 & 0.96 & 1.03 \\ 
  $ r=0.3 $ & 10.27 & 11.37 & 6.63 & 1.12 & 0.95 \\ 
  $ r=0.5 $ & 9.78 & 10.87 & 11.76 & 1.20 & 1.05 \\ 
   \hline
\end{tabular}

	 	 }
\\
	 	\subfloat[][$Z\sim$ Ellip Uniform]{
\begin{tabular}{rrrrrr}
  \hline
 & $\lambda_1=1$ & $\lambda_1=1+ 3 \sqrt{r}$ & $\lambda_1=1+ 11 \sqrt{r}$ & $\lambda_1=1+ 50 \sqrt{r}$ & $\lambda_1=1+ 100 \sqrt{r}$ \\ 
  \hline
$ r=0.01 $ & 1.76 & 1.05 & 1.00 & 0.97 & 0.92 \\ 
  $ r=0.1 $ & 7.17 & 1.51 & 1.11 & 1.08 & 0.99 \\ 
  $ r=0.3 $ & 35.38 & 4.66 & 1.21 & 0.94 & 1.07 \\ 
  $ r=0.5 $ & 84.48 & 10.72 & 1.30 & 1.04 & 0.99 \\ 
   \hline
\end{tabular}

	 	 }

\caption{\textbf{Top Eigenvalue: Median value of ratio of bootstrap estimate of variance to true variance for $n=1000$} This tables give the median values of the boxplots plotted in Figure \ref{fig:bootVar}.}
\label{tab:bootVar}
\end{table}

\begin{landscape}
\begin{table}
	\centering
	\subfloat[][ $Z\sim$ Normal]{
\begin{tabular}{rrrrrr|rrrrr}
  \hline
 & \multicolumn{5}{c}{Percentile}&\multicolumn{5}{c}{Normal}\\ & $\lambda_1=1$ & $\lambda_1=3\sqrt{r}$ & $\lambda_1=11\sqrt{r}$ & $\lambda_1=50\sqrt{r}$ & $\lambda_1=100\sqrt{r}$ & $\lambda_1=1$ & $\lambda_1=3\sqrt{r}$ & $\lambda_1=11\sqrt{r}$ & $\lambda_1=50\sqrt{r}$ & $\lambda_1=100\sqrt{r}$ \\ 
  \hline
$r=0.01$ & 0.00 & 0.71 & 0.94 & 0.94 & 0.95 & 0.36 & 0.92 & 0.95 & 0.94 & 0.95 \\ 
  $r=0.1$ & 0.00 & 0.00 & 0.78 & 0.94 & 0.95 & 0.00 & 0.89 & 0.96 & 0.94 & 0.95 \\ 
  $r=0.3$ & 0.00 & 0.00 & 0.39 & 0.94 & 0.96 & 0.04 & 0.84 & 0.96 & 0.95 & 0.96 \\ 
  $r=0.5$ & 0.00 & 0.00 & 0.16 & 0.90 & 0.94 & 1.00 & 0.73 & 0.96 & 0.96 & 0.95 \\ 
   \hline
\end{tabular}

	 }
\\
	\subfloat[][$Z\sim$ Ellip Exp]{
\begin{tabular}{rrrrrr|rrrrr}
  \hline
 & \multicolumn{5}{c}{Percentile}&\multicolumn{5}{c}{Normal}\\ & $\lambda_1=1$ & $\lambda_1=3\sqrt{r}$ & $\lambda_1=11\sqrt{r}$ & $\lambda_1=50\sqrt{r}$ & $\lambda_1=100\sqrt{r}$ & $\lambda_1=1$ & $\lambda_1=3\sqrt{r}$ & $\lambda_1=11\sqrt{r}$ & $\lambda_1=50\sqrt{r}$ & $\lambda_1=100\sqrt{r}$ \\ 
  \hline
$r=0.01$ & 0.00 & 0.30 & 0.94 & 0.92 & 0.92 & 0.99 & 0.99 & 0.95 & 0.91 & 0.90 \\ 
  $r=0.1$ & 0.00 & 0.00 & 0.20 & 0.95 & 0.94 & 1.00 & 1.00 & 1.00 & 0.94 & 0.93 \\ 
  $r=0.3$ & 0.00 & 0.00 & 0.00 & 0.90 & 0.94 & 1.00 & 1.00 & 1.00 & 0.95 & 0.94 \\ 
  $r=0.5$ & 0.00 & 0.00 & 0.00 & 0.82 & 0.95 & 1.00 & 1.00 & 1.00 & 0.98 & 0.96 \\ 
   \hline
\end{tabular}
		
	 }\\
	 	\subfloat[][$Z\sim$ Ellip Norm]{
\begin{tabular}{rrrrrr|rrrrr}
  \hline
 & \multicolumn{5}{c}{Percentile}&\multicolumn{5}{c}{Normal}\\ & $\lambda_1=1$ & $\lambda_1=3\sqrt{r}$ & $\lambda_1=11\sqrt{r}$ & $\lambda_1=50\sqrt{r}$ & $\lambda_1=100\sqrt{r}$ & $\lambda_1=1$ & $\lambda_1=3\sqrt{r}$ & $\lambda_1=11\sqrt{r}$ & $\lambda_1=50\sqrt{r}$ & $\lambda_1=100\sqrt{r}$ \\ 
  \hline
$r=0.01$ & 0.00 & 0.39 & 0.93 & 0.95 & 0.94 & 0.88 & 0.96 & 0.94 & 0.95 & 0.95 \\ 
  $r=0.1$ & 0.00 & 0.00 & 0.52 & 0.94 & 0.95 & 1.00 & 1.00 & 0.98 & 0.94 & 0.94 \\ 
  $r=0.3$ & 0.00 & 0.00 & 0.01 & 0.93 & 0.94 & 1.00 & 1.00 & 1.00 & 0.95 & 0.95 \\ 
  $r=0.5$ & 0.00 & 0.00 & 0.00 & 0.89 & 0.94 & 1.00 & 1.00 & 1.00 & 0.97 & 0.96 \\ 
   \hline
\end{tabular}

	 	 }
\\
	 	\subfloat[][$Z\sim$ Ellip Uniform]{
\begin{tabular}{rrrrrr|rrrrr}
  \hline
 & \multicolumn{5}{c}{Percentile}&\multicolumn{5}{c}{Normal}\\ & $\lambda_1=1$ & $\lambda_1=3\sqrt{r}$ & $\lambda_1=11\sqrt{r}$ & $\lambda_1=50\sqrt{r}$ & $\lambda_1=100\sqrt{r}$ & $\lambda_1=1$ & $\lambda_1=3\sqrt{r}$ & $\lambda_1=11\sqrt{r}$ & $\lambda_1=50\sqrt{r}$ & $\lambda_1=100\sqrt{r}$ \\ 
  \hline
$r=0.01$ & 0.00 & 0.62 & 0.94 & 0.94 & 0.94 & 0.53 & 0.92 & 0.95 & 0.94 & 0.93 \\ 
  $r=0.1$ & 0.00 & 0.00 & 0.75 & 0.96 & 0.95 & 0.15 & 0.88 & 0.96 & 0.95 & 0.95 \\ 
  $r=0.3$ & 0.00 & 0.00 & 0.33 & 0.92 & 0.95 & 1.00 & 0.73 & 0.97 & 0.94 & 0.94 \\ 
  $r=0.5$ & 0.00 & 0.00 & 0.10 & 0.93 & 0.94 & 1.00 & 0.59 & 0.96 & 0.95 & 0.95 \\ 
   \hline
\end{tabular}

	 	 }
\caption{\textbf{Top Eigenvalue: Median value of 95\% CI Coverage of true $\lambda_1$ for $n=1000$.} This tables give the percentage of CI intervals (out of 1,000 simulations) that cover the true $\lambda_1$ as plotted in Figure \ref{fig:bootCI}.}
\label{tab:bootCITrue}
\end{table}
\end{landscape}

\begin{landscape}
\begin{table}
	\centering
	\subfloat[][ $Z\sim$ Normal]{
\begin{tabular}{rrrrrr|rrrrr}
  \hline
 & \multicolumn{5}{c}{Percentile}&\multicolumn{5}{c}{Normal}\\ & $\lambda_1=1$ & $\lambda_1=3\sqrt{r}$ & $\lambda_1=11\sqrt{r}$ & $\lambda_1=50\sqrt{r}$ & $\lambda_1=100\sqrt{r}$ & $\lambda_1=1$ & $\lambda_1=3\sqrt{r}$ & $\lambda_1=11\sqrt{r}$ & $\lambda_1=50\sqrt{r}$ & $\lambda_1=100\sqrt{r}$ \\ 
  \hline
$r=0.01$ & 0.00 & 0.00 & 0.00 & 0.00 & 0.00 & 0.36 & 0.00 & 0.00 & 0.00 & 0.00 \\ 
  $r=0.1$ & 0.00 & 0.00 & 0.00 & 0.00 & 0.00 & 0.00 & 0.00 & 0.00 & 0.00 & 0.00 \\ 
  $r=0.3$ & 0.00 & 0.00 & 0.00 & 0.00 & 0.00 & 0.04 & 0.00 & 0.00 & 0.00 & 0.00 \\ 
  $r=0.5$ & 0.00 & 0.00 & 0.00 & 0.00 & 0.00 & 1.00 & 0.00 & 0.00 & 0.00 & 0.00 \\ 
   \hline
\end{tabular}

	 }
\\
	\subfloat[][$Z\sim$ Ellip Exp]{
\begin{tabular}{rrrrrr|rrrrr}
  \hline
 & \multicolumn{5}{c}{Percentile}&\multicolumn{5}{c}{Normal}\\ & $\lambda_1=1$ & $\lambda_1=3\sqrt{r}$ & $\lambda_1=11\sqrt{r}$ & $\lambda_1=50\sqrt{r}$ & $\lambda_1=100\sqrt{r}$ & $\lambda_1=1$ & $\lambda_1=3\sqrt{r}$ & $\lambda_1=11\sqrt{r}$ & $\lambda_1=50\sqrt{r}$ & $\lambda_1=100\sqrt{r}$ \\ 
  \hline
$r=0.01$ & 0.00 & 0.00 & 0.00 & 0.00 & 0.00 & 0.99 & 0.92 & 0.04 & 0.00 & 0.00 \\ 
  $r=0.1$ & 0.00 & 0.00 & 0.00 & 0.00 & 0.00 & 1.00 & 1.00 & 0.73 & 0.00 & 0.00 \\ 
  $r=0.3$ & 0.00 & 0.00 & 0.00 & 0.00 & 0.00 & 1.00 & 1.00 & 1.00 & 0.01 & 0.00 \\ 
  $r=0.5$ & 0.00 & 0.00 & 0.00 & 0.00 & 0.00 & 1.00 & 1.00 & 1.00 & 0.04 & 0.00 \\ 
   \hline
\end{tabular}

	 }\\
	 	\subfloat[][$Z\sim$ Ellip Norm]{
\begin{tabular}{rrrrrr|rrrrr}
  \hline
 & \multicolumn{5}{c}{Percentile}&\multicolumn{5}{c}{Normal}\\ & $\lambda_1=1$ & $\lambda_1=3\sqrt{r}$ & $\lambda_1=11\sqrt{r}$ & $\lambda_1=50\sqrt{r}$ & $\lambda_1=100\sqrt{r}$ & $\lambda_1=1$ & $\lambda_1=3\sqrt{r}$ & $\lambda_1=11\sqrt{r}$ & $\lambda_1=50\sqrt{r}$ & $\lambda_1=100\sqrt{r}$ \\ 
  \hline
$r=0.01$ & 0.00 & 0.00 & 0.00 & 0.00 & 0.00 & 0.88 & 0.47 & 0.00 & 0.00 & 0.00 \\ 
  $r=0.1$ & 0.00 & 0.00 & 0.00 & 0.00 & 0.00 & 1.00 & 1.00 & 0.00 & 0.00 & 0.00 \\ 
  $r=0.3$ & 0.00 & 0.00 & 0.00 & 0.00 & 0.00 & 1.00 & 1.00 & 0.21 & 0.00 & 0.00 \\ 
  $r=0.5$ & 0.00 & 0.00 & 0.00 & 0.00 & 0.00 & 1.00 & 1.00 & 0.77 & 0.00 & 0.00 \\ 
   \hline
\end{tabular}

	 	 }
		 \\
		 	 	\subfloat[][$Z\sim$ Ellip Uniform]{
		 \begin{tabular}{rrrrrr|rrrrr}
		   \hline
		  & \multicolumn{5}{c}{Percentile}&\multicolumn{5}{c}{Normal}\\ & $\lambda_1=1$ & $\lambda_1=3\sqrt{r}$ & $\lambda_1=11\sqrt{r}$ & $\lambda_1=50\sqrt{r}$ & $\lambda_1=100\sqrt{r}$ & $\lambda_1=1$ & $\lambda_1=3\sqrt{r}$ & $\lambda_1=11\sqrt{r}$ & $\lambda_1=50\sqrt{r}$ & $\lambda_1=100\sqrt{r}$ \\ 
		   \hline
$r=0.01$ & 0.00 & 0.00 & 0.00 & 0.00 & 0.00 & 0.53 & 0.01 & 0.00 & 0.00 & 0.00 \\ 
  $r=0.1$ & 0.00 & 0.00 & 0.00 & 0.00 & 0.00 & 0.15 & 0.00 & 0.00 & 0.00 & 0.00 \\ 
  $r=0.3$ & 0.00 & 0.00 & 0.00 & 0.00 & 0.00 & 1.00 & 0.01 & 0.00 & 0.00 & 0.00 \\ 
  $r=0.5$ & 0.00 & 0.00 & 0.00 & 0.00 & 0.00 & 1.00 & 0.37 & 0.00 & 0.00 & 0.00 \\ 
   \hline
\end{tabular}

		 	 	 }

\caption{\textbf{Top Eigenvalue: Median value of 95\% CI Coverage of null value $1$ for $n=1000$.} This tables give the percentage of CI intervals (out of 1,000 simulations) that cover the value $\lambda_1=1$ for different values of the true $\lambda_1$.}
\label{tab:bootCINull}
\end{table}
\end{landscape}

\begin{table}
	\centering
	\subfloat[][Bootstrap Median Estimate of Bias]{
\begin{tabular}{rrrrrr}
  \hline
 & $\lambda_1=1$ & $\lambda_1=1+ 3 \sqrt{r}$ & $\lambda_1=1+ 11 \sqrt{r}$ & $\lambda_1=1+ 50 \sqrt{r}$ & $\lambda_1=1+ 100 \sqrt{r}$ \\ 
  \hline
$ r=0.01 $ & 0.03 & -0.02 & -0.05 & -0.06 & -0.06 \\ 
  $ r=0.1 $ & 0.05 & -0.17 & -0.28 & -0.30 & -0.31 \\ 
  $ r=0.3 $ & 0.10 & -0.41 & -0.67 & -0.72 & -0.73 \\ 
  $ r=0.5 $ & 0.18 & -0.61 & -1.08 & -1.17 & -1.18 \\ 
   \hline
\end{tabular}
		
	 }\\
	\subfloat[][True Bias]{
\begin{tabular}{rrrrrr}
  \hline
 & $\lambda_1=1$ & $\lambda_1=1+ 3 \sqrt{r}$ & $\lambda_1=1+ 11 \sqrt{r}$ & $\lambda_1=1+ 50 \sqrt{r}$ & $\lambda_1=1+ 100 \sqrt{r}$ \\ 
  \hline
$ r=0.01 $ & 0.05 & -0.12 & -0.14 & -0.14 & -0.15 \\ 
  $ r=0.1 $ & 0.04 & -0.49 & -0.58 & -0.58 & -0.58 \\ 
  $ r=0.3 $ & 0.05 & -0.88 & -1.01 & -1.08 & -0.98 \\ 
  $ r=0.5 $ & 0.05 & -1.15 & -1.33 & -1.40 & -1.43 \\ 
   \hline
\end{tabular}
		
	 }
\caption{\textbf{Gap Statistic: Median value of Bootstrap and True values of Bias, $Z\sim$ Normal} }
\label{tab:bootBiasNormGap}
\end{table}

\begin{table}
	\centering
	\subfloat[][Bootstrap Median Estimate of Bias]{
\begin{tabular}{rrrrrr}
  \hline
 & $\lambda_1=1$ & $\lambda_1=1+ 3 \sqrt{r}$ & $\lambda_1=1+ 11 \sqrt{r}$ & $\lambda_1=1+ 50 \sqrt{r}$ & $\lambda_1=1+ 100 \sqrt{r}$ \\ 
  \hline
$ r=0.01 $ & 0.14 & 0.13 & -0.04 & -0.12 & -0.14 \\ 
  $ r=0.1 $ & 0.93 & 0.92 & 0.40 & -1.33 & -1.47 \\ 
  $ r=0.3 $ & 2.74 & 2.76 & 2.54 & -3.50 & -4.27 \\ 
  $ r=0.5 $ & 4.56 & 4.44 & 4.41 & -4.82 & -6.76 \\ 
   \hline
\end{tabular}

	 }\\
	\subfloat[][True Bias]{
\begin{tabular}{rrrrrr}
  \hline
 & $\lambda_1=1$ & $\lambda_1=1+ 3 \sqrt{r}$ & $\lambda_1=1+ 11 \sqrt{r}$ & $\lambda_1=1+ 50 \sqrt{r}$ & $\lambda_1=1+ 100 \sqrt{r}$ \\ 
  \hline
$ r=0.01 $ & 0.19 & -0.07 & -0.30 & -0.37 & -0.41 \\ 
  $ r=0.1 $ & 0.85 & -0.11 & -1.88 & -2.63 & -2.41 \\ 
  $ r=0.3 $ & 2.32 & 0.64 & -3.51 & -6.70 & -7.02 \\ 
  $ r=0.5 $ & 3.76 & 1.71 & -3.74 & -10.42 & -11.33 \\ 
   \hline
\end{tabular}
		
	 }
	 \caption{\textbf{Gap Statistic: Median value of Bootstrap and True values of Bias, $Z\sim$ Ellip Exp} This tables give the median values of the boxplots plotted in Figure \ref{fig:bootBias}, as well as the true bias values (*) in the plots. See figure caption for more details. }
	 \label{tab:bootBiasEllipExpGap}
	 \end{table}

\begin{table}
	 	\centering
	 	\subfloat[][Bootstrap Median Estimate of Bias]{
\begin{tabular}{rrrrrr}
  \hline
 & $\lambda_1=1$ & $\lambda_1=1+ 3 \sqrt{r}$ & $\lambda_1=1+ 11 \sqrt{r}$ & $\lambda_1=1+ 50 \sqrt{r}$ & $\lambda_1=1+ 100 \sqrt{r}$ \\ 
  \hline
$ r=0.01 $ & 0.08 & 0.05 & -0.07 & -0.10 & -0.11 \\ 
  $ r=0.1 $ & 0.44 & 0.42 & -0.66 & -0.99 & -1.02 \\ 
  $ r=0.3 $ & 1.28 & 1.29 & -1.00 & -2.99 & -3.10 \\ 
  $ r=0.5 $ & 2.10 & 2.11 & -0.34 & -4.94 & -5.17 \\ 
   \hline
\end{tabular}

	 	 }\\
	 	\subfloat[][True Bias]{
\begin{tabular}{rrrrrr}
  \hline
 & $\lambda_1=1$ & $\lambda_1=1+ 3 \sqrt{r}$ & $\lambda_1=1+ 11 \sqrt{r}$ & $\lambda_1=1+ 50 \sqrt{r}$ & $\lambda_1=1+ 100 \sqrt{r}$ \\ 
  \hline
$ r=0.01 $ & 0.11 & -0.13 & -0.24 & -0.25 & -0.23 \\ 
  $ r=0.1 $ & 0.21 & -0.70 & -1.22 & -1.28 & -1.51 \\ 
  $ r=0.3 $ & 0.56 & -1.11 & -2.83 & -3.25 & -3.13 \\ 
  $ r=0.5 $ & 0.95 & -1.21 & -4.35 & -4.80 & -4.70 \\ 
   \hline
\end{tabular}
		
	 	 }\\
\caption{\textbf{Gap Statistic: Median value of Bootstrap and True values of Bias, $Z\sim$ Ellip Norm}}
\label{tab:bootBiasEllipNormGap}
\end{table}

	 \begin{table}
	 	\centering
	 	\subfloat[][Bootstrap Median Estimate of Bias]{
\begin{tabular}{rrrrrr}
  \hline
 & $\lambda_1=1$ & $\lambda_1=1+ 3 \sqrt{r}$ & $\lambda_1=1+ 11 \sqrt{r}$ & $\lambda_1=1+ 50 \sqrt{r}$ & $\lambda_1=1+ 100 \sqrt{r}$ \\ 
  \hline
$ r=0.01 $ & 0.03 & -0.02 & -0.06 & -0.07 & -0.07 \\ 
  $ r=0.1 $ & 0.08 & -0.17 & -0.37 & -0.41 & -0.41 \\ 
  $ r=0.3 $ & 0.24 & -0.30 & -1.04 & -1.14 & -1.15 \\ 
  $ r=0.5 $ & 0.43 & -0.28 & -1.77 & -1.94 & -1.96 \\ 
   \hline
\end{tabular}

	 	 }\\
	 	\subfloat[][True Bias]{
\begin{tabular}{rrrrrr}
  \hline
 & $\lambda_1=1$ & $\lambda_1=1+ 3 \sqrt{r}$ & $\lambda_1=1+ 11 \sqrt{r}$ & $\lambda_1=1+ 50 \sqrt{r}$ & $\lambda_1=1+ 100 \sqrt{r}$ \\ 
  \hline
$ r=0.01 $ & 0.06 & -0.12 & -0.16 & -0.17 & -0.17 \\ 
  $ r=0.1 $ & 0.05 & -0.55 & -0.66 & -0.70 & -0.60 \\ 
  $ r=0.3 $ & 0.06 & -1.02 & -1.20 & -1.20 & -1.37 \\ 
  $ r=0.5 $ & 0.07 & -1.34 & -1.60 & -1.70 & -1.69 \\ 
   \hline
\end{tabular}
		
	 	 }\\
\caption{\textbf{Gap Statistic: Median value of Bootstrap and True values of Bias, $Z\sim$ Ellip Uniform}}
\label{tab:bootBiasEllipUnifGap}
\end{table}

\begin{table}
	\centering
	\subfloat[][ $Z\sim$ Normal]{
\begin{tabular}{rrrrrr}
  \hline
 & $\lambda_1=1$ & $\lambda_1=1+ 3 \sqrt{r}$ & $\lambda_1=1+ 11 \sqrt{r}$ & $\lambda_1=1+ 50 \sqrt{r}$ & $\lambda_1=1+ 100 \sqrt{r}$ \\ 
  \hline
$ r=0.01 $ & 2.05 & 1.08 & 1.14 & 0.98 & 1.03 \\ 
  $ r=0.1 $ & 4.23 & 1.23 & 1.15 & 0.96 & 0.96 \\ 
  $ r=0.3 $ & 13.33 & 1.62 & 1.25 & 1.03 & 1.06 \\ 
  $ r=0.5 $ & 40.50 & 1.72 & 1.42 & 1.02 & 1.02 \\ 
   \hline
\end{tabular}

	 }
\\
	\subfloat[][$Z\sim$ Ellip Exp]{
\begin{tabular}{rrrrrr}
  \hline
 & $\lambda_1=1$ & $\lambda_1=1+ 3 \sqrt{r}$ & $\lambda_1=1+ 11 \sqrt{r}$ & $\lambda_1=1+ 50 \sqrt{r}$ & $\lambda_1=1+ 100 \sqrt{r}$ \\ 
  \hline
$ r=0.01 $ & 2.12 & 2.00 & 1.02 & 0.72 & 0.73 \\ 
  $ r=0.1 $ & 2.73 & 3.11 & 2.77 & 1.11 & 0.89 \\ 
  $ r=0.3 $ & 3.32 & 3.68 & 5.25 & 1.30 & 0.94 \\ 
  $ r=0.5 $ & 3.24 & 3.09 & 2.91 & 1.52 & 1.13 \\ 
   \hline
\end{tabular}

	 }\\
	 	\subfloat[][$Z\sim$ Ellip Norm]{
\begin{tabular}{rrrrrr}
  \hline
 & $\lambda_1=1$ & $\lambda_1=1+ 3 \sqrt{r}$ & $\lambda_1=1+ 11 \sqrt{r}$ & $\lambda_1=1+ 50 \sqrt{r}$ & $\lambda_1=1+ 100 \sqrt{r}$ \\ 
  \hline
$ r=0.01 $ & 2.67 & 1.68 & 1.11 & 0.99 & 0.91 \\ 
  $ r=0.1 $ & 12.22 & 10.60 & 1.92 & 1.01 & 1.02 \\ 
  $ r=0.3 $ & 11.52 & 14.83 & 3.22 & 1.44 & 0.97 \\ 
  $ r=0.5 $ & 11.07 & 12.41 & 5.36 & 1.67 & 1.15 \\ 
   \hline
\end{tabular}

	 	 }
\\
	 	\subfloat[][$Z\sim$ Ellip Uniform]{
\begin{tabular}{rrrrrr}
  \hline
 & $\lambda_1=1$ & $\lambda_1=1+ 3 \sqrt{r}$ & $\lambda_1=1+ 11 \sqrt{r}$ & $\lambda_1=1+ 50 \sqrt{r}$ & $\lambda_1=1+ 100 \sqrt{r}$ \\ 
  \hline
$ r=0.01 $ & 2.19 & 1.07 & 1.08 & 0.97 & 0.92 \\ 
  $ r=0.1 $ & 8.14 & 1.31 & 1.23 & 1.08 & 0.99 \\ 
  $ r=0.3 $ & 47.97 & 2.72 & 1.67 & 0.95 & 1.07 \\ 
  $ r=0.5 $ & 111.75 & 6.25 & 2.22 & 1.09 & 0.99 \\ 
   \hline
\end{tabular}

	 	 }

\caption{\textbf{Gap Statistic: Median value of ratio of bootstrap estimate of variance to true variance for $n=1000$} This tables give the median values of the boxplots plotted in SupplementaryFigure \ref{fig:bootVarGap}.}
\label{tab:bootVarGap}
\end{table}

\begin{landscape}
\begin{table}
	\centering
	\subfloat[][ $Z\sim$ Normal]{
\begin{tabular}{rrrrrr|rrrrr}
  \hline
 & \multicolumn{5}{c}{Percentile}&\multicolumn{5}{c}{Normal}\\ & $\lambda_1=1$ & $\lambda_1=3\sqrt{r}$ & $\lambda_1=11\sqrt{r}$ & $\lambda_1=50\sqrt{r}$ & $\lambda_1=100\sqrt{r}$ & $\lambda_1=1$ & $\lambda_1=3\sqrt{r}$ & $\lambda_1=11\sqrt{r}$ & $\lambda_1=50\sqrt{r}$ & $\lambda_1=100\sqrt{r}$ \\ 
  \hline
$r=0.01$ & 0.00 & 0.44 & 0.53 & 0.87 & 0.93 & 0.90 & 0.61 & 0.85 & 0.91 & 0.93 \\ 
  $r=0.1$ & 0.00 & 0.00 & 0.02 & 0.77 & 0.90 & 0.97 & 0.15 & 0.69 & 0.92 & 0.95 \\ 
  $r=0.3$ & 0.00 & 0.00 & 0.00 & 0.71 & 0.89 & 1.00 & 0.14 & 0.84 & 0.93 & 0.96 \\ 
  $r=0.5$ & 0.00 & 0.00 & 0.00 & 0.65 & 0.87 & 1.00 & 0.21 & 0.95 & 0.95 & 0.95 \\ 
   \hline
\end{tabular}

	 }
\\
	\subfloat[][$Z\sim$ Ellip Exp]{
\begin{tabular}{rrrrrr|rrrrr}
  \hline
 & \multicolumn{5}{c}{Percentile}&\multicolumn{5}{c}{Normal}\\ & $\lambda_1=1$ & $\lambda_1=3\sqrt{r}$ & $\lambda_1=11\sqrt{r}$ & $\lambda_1=50\sqrt{r}$ & $\lambda_1=100\sqrt{r}$ & $\lambda_1=1$ & $\lambda_1=3\sqrt{r}$ & $\lambda_1=11\sqrt{r}$ & $\lambda_1=50\sqrt{r}$ & $\lambda_1=100\sqrt{r}$ \\ 
  \hline
$r=0.01$ & 0.00 & 1.00 & 0.75 & 0.81 & 0.85 & 1.00 & 0.79 & 0.76 & 0.86 & 0.87 \\ 
  $r=0.1$ & 0.00 & 1.00 & 0.92 & 0.58 & 0.76 & 1.00 & 0.97 & 0.63 & 0.88 & 0.90 \\ 
  $r=0.3$ & 0.00 & 1.00 & 1.00 & 0.47 & 0.62 & 1.00 & 1.00 & 0.75 & 0.86 & 0.90 \\ 
  $r=0.5$ & 0.00 & 1.00 & 1.00 & 0.45 & 0.57 & 1.00 & 1.00 & 0.86 & 0.83 & 0.90 \\ 
   \hline
\end{tabular}
		
	 }\\
	 	\subfloat[][$Z\sim$ Ellip Norm]{
\begin{tabular}{rrrrrr|rrrrr}
  \hline
 & \multicolumn{5}{c}{Percentile}&\multicolumn{5}{c}{Normal}\\ & $\lambda_1=1$ & $\lambda_1=3\sqrt{r}$ & $\lambda_1=11\sqrt{r}$ & $\lambda_1=50\sqrt{r}$ & $\lambda_1=100\sqrt{r}$ & $\lambda_1=1$ & $\lambda_1=3\sqrt{r}$ & $\lambda_1=11\sqrt{r}$ & $\lambda_1=50\sqrt{r}$ & $\lambda_1=100\sqrt{r}$ \\ 
  \hline
$r=0.01$ & 0.00 & 1.00 & 0.64 & 0.86 & 0.91 & 0.96 & 0.59 & 0.83 & 0.91 & 0.93 \\ 
  $r=0.1$ & 0.00 & 1.00 & 0.12 & 0.65 & 0.82 & 1.00 & 0.42 & 0.85 & 0.94 & 0.93 \\ 
  $r=0.3$ & 0.00 & 1.00 & 0.45 & 0.42 & 0.72 & 1.00 & 0.83 & 0.81 & 0.97 & 0.95 \\ 
  $r=0.5$ & 0.00 & 1.00 & 0.91 & 0.31 & 0.67 & 1.00 & 0.97 & 0.75 & 0.99 & 0.97 \\ 
   \hline
\end{tabular}

	 	 }
\\
	 	\subfloat[][$Z\sim$ Ellip Uniform]{
\begin{tabular}{rrrrrr|rrrrr}
  \hline
 & \multicolumn{5}{c}{Percentile}&\multicolumn{5}{c}{Normal}\\ & $\lambda_1=1$ & $\lambda_1=3\sqrt{r}$ & $\lambda_1=11\sqrt{r}$ & $\lambda_1=50\sqrt{r}$ & $\lambda_1=100\sqrt{r}$ & $\lambda_1=1$ & $\lambda_1=3\sqrt{r}$ & $\lambda_1=11\sqrt{r}$ & $\lambda_1=50\sqrt{r}$ & $\lambda_1=100\sqrt{r}$ \\ 
  \hline
$r=0.01$ & 0.00 & 0.57 & 0.53 & 0.87 & 0.92 & 0.92 & 0.61 & 0.83 & 0.92 & 0.93 \\ 
  $r=0.1$ & 0.00 & 0.00 & 0.02 & 0.74 & 0.91 & 1.00 & 0.17 & 0.77 & 0.93 & 0.94 \\ 
  $r=0.3$ & 0.00 & 0.00 & 0.00 & 0.65 & 0.87 & 1.00 & 0.10 & 0.98 & 0.94 & 0.94 \\ 
  $r=0.5$ & 0.00 & 0.00 & 0.00 & 0.56 & 0.82 & 1.00 & 0.16 & 0.99 & 0.96 & 0.95 \\ 
   \hline
\end{tabular}

	 	 }
\caption{\textbf{Gap Statistic: Median value of 95\% CI Coverage of true Gap for $n=1000$.} This tables give the percentage of CI intervals (out of 1,000 simulations) that cover the true $\lambda_1-\lambda_2$ as plotted in Supplementary Figure \ref{fig:bootCIGap}.}
\label{tab:bootCITrueGap}
\end{table}
\end{landscape}

\begin{landscape}
\begin{table}
	\centering
	\subfloat[][ $Z\sim$ Normal]{
\begin{tabular}{rrrrrr|rrrrr}
  \hline
 & \multicolumn{5}{c}{Percentile}&\multicolumn{5}{c}{Normal}\\ & $\lambda_1=1$ & $\lambda_1=3\sqrt{r}$ & $\lambda_1=11\sqrt{r}$ & $\lambda_1=50\sqrt{r}$ & $\lambda_1=100\sqrt{r}$ & $\lambda_1=1$ & $\lambda_1=3\sqrt{r}$ & $\lambda_1=11\sqrt{r}$ & $\lambda_1=50\sqrt{r}$ & $\lambda_1=100\sqrt{r}$ \\ 
  \hline
$r=0.01$ & 0.00 & 0.00 & 0.85 & 0.00 & 0.00 & 0.00 & 0.00 & 0.95 & 0.00 & 0.00 \\ 
  $r=0.1$ & 0.00 & 0.00 & 0.00 & 0.00 & 0.00 & 0.00 & 0.08 & 0.00 & 0.00 & 0.00 \\ 
  $r=0.3$ & 0.00 & 0.00 & 0.00 & 0.00 & 0.00 & 0.00 & 0.78 & 0.00 & 0.00 & 0.00 \\ 
  $r=0.5$ & 0.00 & 0.01 & 0.00 & 0.00 & 0.00 & 0.00 & 0.12 & 0.00 & 0.00 & 0.00 \\ 
   \hline
\end{tabular}		
		
	 }
\\
	\subfloat[][$Z\sim$ Ellip Exp]{
\begin{tabular}{rrrrrr|rrrrr}
  \hline
 & \multicolumn{5}{c}{Percentile}&\multicolumn{5}{c}{Normal}\\ & $\lambda_1=1$ & $\lambda_1=3\sqrt{r}$ & $\lambda_1=11\sqrt{r}$ & $\lambda_1=50\sqrt{r}$ & $\lambda_1=100\sqrt{r}$ & $\lambda_1=1$ & $\lambda_1=3\sqrt{r}$ & $\lambda_1=11\sqrt{r}$ & $\lambda_1=50\sqrt{r}$ & $\lambda_1=100\sqrt{r}$ \\ 
  \hline
$r=0.01$ & 0.28 & 0.33 & 0.87 & 0.00 & 0.00 & 0.07 & 0.08 & 0.84 & 0.00 & 0.00 \\ 
  $r=0.1$ & 1.00 & 1.00 & 1.00 & 0.00 & 0.00 & 0.95 & 0.95 & 0.96 & 0.01 & 0.00 \\ 
  $r=0.3$ & 1.00 & 1.00 & 0.99 & 0.00 & 0.00 & 1.00 & 1.00 & 1.00 & 0.04 & 0.00 \\ 
  $r=0.5$ & 0.97 & 0.97 & 0.92 & 0.00 & 0.00 & 1.00 & 1.00 & 1.00 & 0.10 & 0.01 \\ 
   \hline
\end{tabular}

	 }\\
	 	\subfloat[][$Z\sim$ Ellip Norm]{
\begin{tabular}{rrrrrr|rrrrr}
  \hline
 & \multicolumn{5}{c}{Percentile}&\multicolumn{5}{c}{Normal}\\ & $\lambda_1=1$ & $\lambda_1=3\sqrt{r}$ & $\lambda_1=11\sqrt{r}$ & $\lambda_1=50\sqrt{r}$ & $\lambda_1=100\sqrt{r}$ & $\lambda_1=1$ & $\lambda_1=3\sqrt{r}$ & $\lambda_1=11\sqrt{r}$ & $\lambda_1=50\sqrt{r}$ & $\lambda_1=100\sqrt{r}$ \\ 
  \hline
$r=0.01$ & 0.00 & 0.00 & 0.79 & 0.00 & 0.00 & 0.00 & 0.00 & 0.92 & 0.00 & 0.00 \\ 
  $r=0.1$ & 1.00 & 1.00 & 0.96 & 0.00 & 0.00 & 0.27 & 0.35 & 0.07 & 0.00 & 0.00 \\ 
  $r=0.3$ & 1.00 & 1.00 & 0.99 & 0.00 & 0.00 & 1.00 & 1.00 & 0.39 & 0.00 & 0.00 \\ 
  $r=0.5$ & 1.00 & 1.00 & 1.00 & 0.00 & 0.00 & 1.00 & 1.00 & 0.78 & 0.00 & 0.00 \\ 
   \hline
\end{tabular}

	 	 }
		 \\
		 	 	\subfloat[][$Z\sim$ Ellip Uniform]{
		 \begin{tabular}{rrrrrr|rrrrr}
		   \hline
		  & \multicolumn{5}{c}{Percentile}&\multicolumn{5}{c}{Normal}\\ & $\lambda_1=1$ & $\lambda_1=3\sqrt{r}$ & $\lambda_1=11\sqrt{r}$ & $\lambda_1=50\sqrt{r}$ & $\lambda_1=100\sqrt{r}$ & $\lambda_1=1$ & $\lambda_1=3\sqrt{r}$ & $\lambda_1=11\sqrt{r}$ & $\lambda_1=50\sqrt{r}$ & $\lambda_1=100\sqrt{r}$ \\ 
		   \hline
$r=0.01$ & 0.00 & 0.00 & 0.83 & 0.00 & 0.00 & 0.00 & 0.00 & 0.94 & 0.00 & 0.00 \\ 
  $r=0.1$ & 0.00 & 0.00 & 0.00 & 0.00 & 0.00 & 0.00 & 0.10 & 0.00 & 0.00 & 0.00 \\ 
  $r=0.3$ & 0.07 & 0.03 & 0.00 & 0.00 & 0.00 & 0.00 & 0.90 & 0.00 & 0.00 & 0.00 \\ 
  $r=0.5$ & 1.00 & 1.00 & 0.00 & 0.00 & 0.00 & 0.00 & 0.98 & 0.00 & 0.00 & 0.00 \\ 
   \hline
\end{tabular}

		 	 	 }

\caption{\textbf{Gap Statistic: Median value of 95\% CI Coverage of null value $0$ for $n=1000$.} This tables give the percentage of CI intervals (out of 1,000 simulations) that cover the value $\lambda_1-\lambda_2=0$ for different values of the true $\lambda_1$.}
\label{tab:bootCINullGap}
\end{table}
\end{landscape}
\clearpage

\begin{center}
\textbf{\textsc{Bibliography}}
\end{center}

\bibliographystyle{plain}
\bibliography{research,Bioresearch,addBoot}

\def\cprime{$'$}
\begin{thebibliography}{10}

\bibitem{BootstrapEigenvaluesAlemayehu88}
Demissie Alemayehu.
\newblock Bootstrapping the latent roots of certain random matrices.
\newblock {\em Comm. Statist. Simulation Comput.}, 17(3):857--869, 1988.

\bibitem{anderson63}
T.~W. Anderson.
\newblock Asymptotic theory for principal component analysis.
\newblock {\em Ann. Math. Statist.}, 34:122--148, 1963.

\bibitem{anderson03}
T.~W. Anderson.
\newblock {\em An introduction to multivariate statistical analysis}.
\newblock Wiley Series in Probability and Statistics. Wiley-Interscience [John
  Wiley \& Sons], Hoboken, NJ, third edition, 2003.

\bibitem{bootstrapCIsInPCAChemometrics13}
Hamid Babamoradi, Frans van~den Berg, and \r{A}smund Rinnan.
\newblock Bootstrap based confidence limits in principal component analysis --
  a case study.
\newblock {\em Chemometrics and Intelligent Laboratory Systems}, 120:97--105,
  January 2013.

\bibitem{bai99}
Z.~D. Bai.
\newblock Methodologies in spectral analysis of large-dimensional random
  matrices, a review.
\newblock {\em Statist. Sinica}, 9(3):611--677, 1999.
\newblock With comments by G. J.\ Rodgers and Jack W.\ Silverstein; and a
  rejoinder by the author.

\bibitem{bbap}
J.~Baik, G.~Ben~Arous, and S.~P\'ech\'e.
\newblock Phase transition of the largest eigenvalue for non-null complex
  sample covariance matrices.
\newblock {\em Ann. Probab.}, 33(5):1643--1697, 2005.

\bibitem{BeranSrivastava85}
Rudolf Beran and Muni~S. Srivastava.
\newblock Bootstrap tests and confidence regions for functions of a covariance
  matrix.
\newblock {\em Ann. Statist.}, 13(1):95--115, 1985.

\bibitem{BeranSrivastava87Correction}
Rudolf Beran and Muni~S. Srivastava.
\newblock Correction: ``{B}ootstrap tests and confidence regions for functions
  of a covariance matrix'' [{A}nn.\ {S}tatist.\ {\bf 13} (1985), no.\ 1,
  95--115; {MR}0773155 (86g:62054)].
\newblock {\em Ann. Statist.}, 15(1):470--471, 1987.

\bibitem{bhatia97}
Rajendra Bhatia.
\newblock {\em Matrix analysis}, volume 169 of {\em Graduate Texts in
  Mathematics}.
\newblock Springer-Verlag, New York, 1997.

\bibitem{BickelFreedmanTheoryBootAoS81}
Peter~J. Bickel and David~A. Freedman.
\newblock Some asymptotic theory for the bootstrap.
\newblock {\em Ann. Statist.}, 9(6):1196--1217, 1981.

\bibitem{BoutetKhorunVasilchuk96}
A.~Boutet~de Monvel, A.~Khorunzhy, and V.~Vasilchuk.
\newblock Limiting eigenvalue distribution of random matrices with correlated
  entries.
\newblock {\em Markov Process. Related Fields}, 2(4):607--636, 1996.

\bibitem{BreimanBaggingPaper96}
Leo Breiman.
\newblock Bagging predictors.
\newblock {\em Machine Learning}, 24(2):123--140, 1996.

\bibitem{DavisonHinkley97}
A.~C. Davison and D.~V. Hinkley.
\newblock {\em Bootstrap methods and their application}.
\newblock Cambridge Series in Statistical and Probabilistic Mathematics.
  Cambridge University Press, Cambridge, 1997.

\bibitem{DiaconisFreedmanProjPursuit84}
Persi Diaconis and David Freedman.
\newblock Asymptotics of graphical projection pursuit.
\newblock {\em Ann. Statist.}, 12(3):793--815, 1984.

\bibitem{dieng04}
Momar Dieng.
\newblock Distribution functions for edge eigenvalues in orthogonal and
  symplectic ensembles: {P}ainlev\'e representations.
\newblock {\em Int. Math. Res. Not.}, 37(37):2263--2287, 2005.

\bibitem{DuembgenNondiffFuncsAndBootPTRF93}
Lutz D{\"u}mbgen.
\newblock On nondifferentiable functions and the bootstrap.
\newblock {\em Probab. Theory Related Fields}, 95(1):125--140, 1993.

\bibitem{EatonTyler91}
Morris~L. Eaton and David~E. Tyler.
\newblock On {W}ielandt's inequality and its application to the asymptotic
  distribution of the eigenvalues of a random symmetric matrix.
\newblock {\em Ann. Statist.}, 19(1):260--271, 1991.

\bibitem{EfronBootstrap1979AoS}
B.~Efron.
\newblock Bootstrap methods: another look at the jackknife.
\newblock {\em Ann. Statist.}, 7(1):1--26, 1979.

\bibitem{nek04}
Noureddine {El Karoui}.
\newblock On the largest eigenvalue of {W}ishart matrices with identity
  covariance when $n,p$ and $p/n\tendsto \infty$.
\newblock {\em arXiv:math.ST/0309355}, September 2003.

\bibitem{nekGencov}
Noureddine {El Karoui}.
\newblock Tracy-{W}idom limit for the largest eigenvalue of a large class of
  complex sample covariance matrices.
\newblock {\em The Annals of Probability}, 35(2):663--714, March 2007.

\bibitem{nekCorrEllipD}
Noureddine {El Karoui}.
\newblock Concentration of measure and spectra of random matrices: Applications
  to correlation matrices, elliptical distributions and beyond.
\newblock {\em The Annals of Applied Probability}, 19(6):2362--2405, December
  2009.

\bibitem{NEKEliBootRegression15}
Noureddine {El Karoui} and Elizabeth Purdom.
\newblock Can we trust the bootstrap in high-dimension?
\newblock Technical Report 824, UC Berkeley, Department of Statistics, February
  2015.
\newblock Submitted to AoS.

\bibitem{FastExactBootstrapPCA14}
Aaron Fisher, Brian Caffo, Brian Schwartz, and Vadim Zipunnikov.
\newblock Fast, exact bootstrap principal component analysis for $p>1$ million.
\newblock {\em Journal of the American Statistical Association}, 2015.

\bibitem{geman80}
S.~Geman.
\newblock A limit theorem for the norm of random matrices.
\newblock {\em Ann. Probab.}, 8(2):252--261, 1980.

\bibitem{geronimohill03}
Jeffrey~S. Geronimo and Theodore~P. Hill.
\newblock Necessary and sufficient condition that the limit of {S}tieltjes
  transforms is a {S}tieltjes transform.
\newblock {\em J. Approx. Theory}, 121(1):54--60, 2003.

\bibitem{GoetzeTikhomirov05}
F.~G{\"o}tze and A.~N. Tikhomirov.
\newblock Limit theorems for spectra of random matrices with martingale
  structure.
\newblock In {\em Stein's method and applications}, volume~5 of {\em Lect.
  Notes Ser. Inst. Math. Sci. Natl. Univ. Singap.}, pages 181--193. Singapore
  Univ. Press, Singapore, 2005.

\bibitem{HallBootstrapAndEdgeworthExpansion92}
Peter Hall.
\newblock {\em The bootstrap and {E}dgeworth expansion}.
\newblock Springer Series in Statistics. Springer-Verlag, New York, 1992.

\bibitem{HallPaulBootstrapEigenvalues2009}
Peter Hall, Young~K. Lee, Byeong~U. Park, and Debashis Paul.
\newblock Tie-respecting bootstrap methods for estimating distributions of sets
  and functions of eigenvalues.
\newblock {\em Bernoulli}, 15(2):380--401, 2009.

\bibitem{HallMarronNeemanJRSSb05}
Peter Hall, J.~S. Marron, and Amnon Neeman.
\newblock Geometric representation of high dimension, low sample size data.
\newblock {\em J. R. Stat. Soc. Ser. B Stat. Methodol.}, 67(3):427--444, 2005.

\bibitem{hj}
Roger~A. Horn and Charles~R. Johnson.
\newblock {\em Matrix analysis}.
\newblock Cambridge University Press, Cambridge, 1990.
\newblock Corrected reprint of the 1985 original.

\bibitem{JohnstoneReview07}
Iain~M. Johnstone.
\newblock High dimensional statistical inference and random matrices.
\newblock In {\em International {C}ongress of {M}athematicians. {V}ol. {I}},
  pages 307--333. Eur. Math. Soc., Z\"urich, 2007.

\bibitem{imj}
I.M. Johnstone.
\newblock On the distribution of the largest eigenvalue in principal component
  analysis.
\newblock {\em Ann. Statist.}, 29(2):295--327, 2001.

\bibitem{KatoPerturbTheory}
Tosio Kato.
\newblock {\em Perturbation theory for linear operators}.
\newblock Classics in Mathematics. Springer-Verlag, Berlin, 1995.
\newblock Reprint of the 1980 edition.

\bibitem{ledoux2001}
M.~Ledoux.
\newblock {\em The concentration of measure phenomenon}, volume~89 of {\em
  Mathematical Surveys and Monographs}.
\newblock American Mathematical Society, Providence, RI, 2001.

\bibitem{LeeSchnelli14}
Ji~oon Lee and Kevin Schnelli.
\newblock Tracy-widom distribution for the largest eigenvalue of real sample
  covariance matrices with general population.
\newblock {\em arxiv:1409.4979}, 2014.

\bibitem{mp67}
V.~A. Mar{\v{c}}enko and L.~A. Pastur.
\newblock Distribution of eigenvalues in certain sets of random matrices.
\newblock {\em Mat. Sb. (N.S.)}, 72 (114):507--536, 1967.

\bibitem{MichaudBook98}
Richard~O. Michaud.
\newblock {\em Efficient Asset Management: A Practical Guide to Stock Portfolio
  Optimization and Asset Allocation}.
\newblock {Oxford University Press, USA}, June 1998.

\bibitem{OnatskiTalkAtMIT06}
A.~Onatski.
\newblock Talk at mit conference, July 2006.

\bibitem{Onatski08}
Alexei Onatski.
\newblock The {T}racy-{W}idom limit for the largest eigenvalues of singular
  complex {W}ishart matrices.
\newblock {\em Ann. Appl. Probab.}, 18(2):470--490, 2008.

\bibitem{PajorPasturPub09}
A.~Pajor and L.~Pastur.
\newblock On the limiting empirical measure of eigenvalues of the sum of rank
  one matrices with log-concave distribution.
\newblock {\em Studia Math.}, 195(1):11--29, 2009.

\bibitem{debashis}
Debashis Paul.
\newblock Asymptotics of sample eigenstructure for a large dimensional spiked
  covariance model.
\newblock {\em Statistica Sinica}, 17(4):1617--1642, October 2007.

\bibitem{PaulSilversteinExactSeparation09}
Debashis Paul and Jack~W. Silverstein.
\newblock No eigenvalues outside the support of the limiting empirical spectral
  distribution of a separable covariance matrix.
\newblock {\em J. Multivariate Anal.}, 100(1):37--57, 2009.

\bibitem{PolitisRomanoWolfSubsampling99}
Dimitris~N. Politis, Joseph~P. Romano, and Michael Wolf.
\newblock {\em Subsampling}.
\newblock Springer Series in Statistics. Springer-Verlag, New York, 1999.

\bibitem{silverstein85}
Jack~W. Silverstein.
\newblock The smallest eigenvalue of a large-dimensional {W}ishart matrix.
\newblock {\em Ann. Probab.}, 13(4):1364--1368, 1985.

\bibitem{silverstein89}
Jack~W. Silverstein.
\newblock On the weak limit of the largest eigenvalue of a large-dimensional
  sample covariance matrix.
\newblock {\em J. Multivariate Anal.}, 30(2):307--311, 1989.

\bibitem{silverstein95}
Jack~W. Silverstein.
\newblock Strong convergence of the empirical distribution of eigenvalues of
  large-dimensional random matrices.
\newblock {\em J. Multivariate Anal.}, 55(2):331--339, 1995.

\bibitem{SoshnikovUniversalityWigner99}
Alexander Soshnikov.
\newblock Universality at the edge of the spectrum in {W}igner random matrices.
\newblock {\em Comm. Math. Phys.}, 207(3):697--733, 1999.

\bibitem{bootstrapScreeTests2006}
Stephen K.~Mitchell Sungjin.~Hong and Richard~A. Harshman.
\newblock Bootstrap scree tests: A monte carlo simulation and applications to
  published data.
\newblock {\em British Journal of Mathematical and Statistical Psychology},
  59(1):35--57, May 2006.

\bibitem{TimmermanBootstrapPC2007}
M.~E. Timmerman, H.~A. Kiers, and A.~K. Smilde.
\newblock Estimating confidence intervals for principal component loadings: A
  comparison between the bootstrap and asymptotic results.
\newblock {\em British Journal of Mathematical and Statistical Psychology},
  60(2):295--314, 2007.

\bibitem{wachter78}
Kenneth~W. Wachter.
\newblock The strong limits of random matrix spectra for sample matrices of
  independent elements.
\newblock {\em Annals of Probability}, 6(1):1--18, 1978.

\end{thebibliography}

\end{document}